\documentclass{article}

\usepackage{MyStyle}

\title{\textbf{A Particle-in-cell Method for Plasmas with a Generalized Momentum Formulation, \\ Part III: A family of Gauge conserving  methods} \thanks{The research of the authors was supported by AFOSR grants FA9550-19-1-0281 and  FA9550-17-1-0394, NSF grant DMS-1912183, and DOE grant DE-SC0023164.}} 

\author{Andrew J. Christlieb \thanks{Department of Computational Mathematics, Science and Engineering, Michigan State University, East Lansing, MI, 48824, United States; \href{mailto:christli@msu.edu}{christli@msu.edu}.}
\and William A. Sands \thanks{Department of Mathematical Sciences, University of Delaware, Newark, DE, 19716, United States;
\href{mailto:wsands@udel.edu}{wsands@udel.edu}.}
\and Stephen R. White \thanks{Department of Computational Mathematics, Science and Engineering, Michigan State University, East Lansing, MI, 48824, United States;
\href{mailto:whites73@msu.edu}
{whites73@msu.edu} (corresponding author).}
}

\date{}

\begin{document}

\maketitle

\nocite{christliebPIC2022pt1}
\begin{abstract}
    In this paper, we introduce a new family of spatially co-located field solvers for particle-in-cell applications which evolve the potential formulation of Maxwell's equations under the Lorenz gauge. Our recent work \cite{christliebPIC2022pt2} introduced the concept of time-consistency, which connects charge conservation to the preservation of the gauge at the semi-discrete level. It will be shown that there exists a large family of time discretizations which satisfy this property. Additionally, it will be further shown that for large classes of time marching methods, the satisfaction of the gauge condition automatically implies the satisfaction of Gauss's law for electricity, with the potential formulation ensuring that that Gauss's law for magnetism is satisfied by definition. We focus on popular time marching methods including centered differences, backward differences, and diagonally-implicit Runge-Kutta methods, which are coupled to a spectral discretization in space. We demonstrate the theory by testing the methods on a relativistic Weibel instability and a drifting cloud of electrons.
\end{abstract}


%
%
%
%
\section{Introduction}

In recent work \cite{christliebPIC2022pt1}, we developed a new particle-in-cell (PIC) method for a potential formulation of the relativistic Vlasov-Maxwell system using the Lorenz gauge. These equations were updated using a first-order backward difference formula (BDF), and the resulting modified Helmholtz operator was inverted using a dimensionally-split integral equation method. This approach, when compared against standard explicit PIC methods, demonstrated a number of desirable properties including a resilience against numerical heating, improved stability, and, though it remains unexplored, retains the geometric flexibility of the wave solver explored in our earlier work \cite{causley2017wave-propagation,MOLT-EB-2020}. In \cite{christliebPIC2022pt2}, we established other properties of the method, including a time-consistency property, which ensures that the method conserves charge if and only if it satisfies the semi-discrete Lorenz gauge condition. Additionally, we showed that Gauss's Law for the electric field is satisfied provided that the semi-discrete Lorenz gauge is satisfied. We explored a number of ways to take advantage of this, including a method which maps current on the mesh using the particles and defines the charge density using the continuity equation. This is the method that we will employ in the current work.

The exploration in part II was done entirely using a first-order BDF time discretization method for the fields. In this work we seek to achieve higher-order temporal accuracy using a variety of techniques, including BDF methods, central discretizations, and diagonally implicit Runge-Kutta (DIRK) methods, and we prove that many of these methods satisfy the same properties established in part II. For the general BDF and central difference methods, this will be fully achieved. For the DIRK methods, we will prove that charge conservation implies satisfaction of the Lorenz gauge. Though we do not show it, we have established that the two-stage method holds in the other direction, as well, suggesting that higher-order DIRK methods will also satisfy this property. 

Our primary objective in this work is to prove the time-consistency property for a family of wave solvers. As such, we restrict ourselves to problems with periodic boundary conditions, which allows us to use a spectral discretization that can be rapidly evaluated using a fast Fourier transform (FFT) rather than the fast convolution operator considered in previous work \cite{christliebPIC2022pt1,christliebPIC2022pt2}. Additionally, the restriction to periodic boundary conditions allows us to utilize high-order particle interpolation schemes, as is commonly done with high order non-diffusive wave solvers \cite{Benedetti_2008, Shalaby_2017, Davidson_2015, Sadeghirad_2024}. All wave solvers we implement will be fully-implicit, with the source term $\mathbf{J}^{n+1}$ being computed using an updated particle location $\mathbf{x}^{n+1}$ and old velocity $\mathbf{v}^{n}$.

The remainder of the paper is structured as follows. Section \ref{sec:2 formulation} gives a brief sketch of the problem formulation and the PIC framework employed to solve it. The problem formulation and the particle method have been discussed at great length in previous work \cite{christliebPIC2022pt1,christliebPIC2022pt2}, so we shall limit the presentation for the sake of brevity. More attention will be paid to the various wave solvers we employ and proofs of the aforementioned properties in section \ref{sec:3 wave solvers}, which occupies the bulk of the paper.  We briefly discuss how the particle updates are computed and how the continuity equation is satisfied in section \ref{sec:4 Particle Updates}.  Section \ref{sec:5 Numerical results} will present the numerical experiments confirming our findings, concentrating on the same periodic problems as in \cite{christliebPIC2022pt2}, and, finally, section \ref{sec:6 Conclusion} will wrap things up with some concluding remarks.

%
%
%
%
%
%
\section{Formulation and Background}
\label{sec:2 formulation}

This section contains a description of the problem formulation for the Vlasov-Maxwell system that is considered in this work, which is introduced in section \ref{subsec:2 VM-system}. A Hamiltonian formulation for this problem is described in section \ref{subsec:2 model}, which uses the Lorenz gauge to cast the system in terms of scalar and vector potentials. Under the Lorenz gauge, these equations take the form of wave equations, which are solved using a variety of time discretizations.  This formulation builds on our previous two papers, where we identified the need for time-consistency in order to satisfy the gauge condition and involutions \cite{christliebPIC2022pt1, christliebPIC2022pt2}. At the end of the section, we conclude with a brief summary of the important elements of the formulation.

%
%
\subsection{Relativistic Vlasov-Maxwell System}
\label{subsec:2 VM-system}

In this work, we develop numerical algorithms for plasmas described by the relativistic Vlasov-Maxwell system, which in SI units, reads as
\begin{empheq}[left=\empheqlbrace]{align}
    &\partial_t f_{s} + \frac{\mathbf{p}}{m_{s}\gamma_{s}} \cdot \nabla_{x} f_{s} + q_{s} \left( \mathbf{E} + \frac{\mathbf{p} \times \mathbf{B}}{m_{s}\gamma_{s}} \right)\cdot \nabla_{p} f_{s} = 0, \label{eq:Vlasov species}\\
    &\nabla \times \mathbf{E} = -\partial_t \mathbf{B}, \label{eq:Farady} \\
    &\nabla \times \mathbf{B} = \mu_{0}\left( \mathbf{J} + \epsilon_{0} \partial_t \mathbf{E} \right), \label{eq:Ampere} \\
    &\nabla \cdot \mathbf{E} = \frac{\rho}{\epsilon_0}, \label{eq:Gauss-E} \\
    &\nabla \cdot \mathbf{B} = 0. \label{eq:Gauss-B}
\end{empheq}
The first equation \eqref{eq:Vlasov species} is the relativistic Vlasov equation which describes the evolution of a probability distribution function $f_{s}\left( \mathbf{x}, \mathbf{p}, t\right)$ for particles of species $s$ in phase space which have mass $m_{s}$ and charge $q_{s}$. Here, we define $\gamma_{s} = 1/\sqrt{1 + \mathbf{p}^2/(m_{s} c)^2} =  1/\sqrt{1 - \mathbf{v}^2/c^2}$, which makes equation \eqref{eq:Vlasov species} Lorentz invariant. Physically, equation \eqref{eq:Vlasov species} describes the time evolution of a distribution function that represents the probability of finding a particle of species $s$ at the position $\mathbf{x}$, with linear momentum $\mathbf{p} = m_{s} \gamma_{s} \mathbf{v}$, at any given time $t$. Since the position and velocity data are vectors with 3 components, the distribution function is a scalar function of 6 dimensions plus time. While the equation itself has fairly simple structure, the primary challenge in numerically solving this equation is its high dimensionality. This growth in the dimensionality has posed tremendous difficulties for grid-based discretization methods, where one often needs to use many grid points to resolve relevant space and time scales in the problem. This difficulty is compounded by the fact that many plasmas of interest contain multiple species. Despite the lack of a collision operator on the right-hand side of \eqref{eq:Vlasov species}, collisions occur in a mean-field sense through the electric and magnetic fields, which appear as coefficients of the velocity gradient.

Equations \eqref{eq:Farady} - \eqref{eq:Gauss-B} are Maxwell's equations, which describe the evolution of the background electric and magnetic fields. Since the plasma is a collection of moving charges, any changes in the distribution function for each species will be reflected in the charge density $\rho(\mathbf{x},t)$, as well as the current density $\mathbf{J}(\mathbf{x},t)$, which, respectively, are the source terms for Gauss's law \eqref{eq:Gauss-E} and Amp\`ere's law \eqref{eq:Ampere}. For $N_{s}$ species, the total charge density and current density are defined by summing over the species
\begin{equation}
    \rho(\mathbf{x},t) = \sum_{s=1}^{N_{s}} \rho_{s}(\mathbf{x},t), \quad \mathbf{J}(\mathbf{x},t) = \sum_{s=1}^{N_{s}} \mathbf{J}_{s}(\mathbf{x},t), \label{eq:total charge + current densities}
\end{equation}
where the species charge and current densities are defined through moments of the distribution function $f_s$:
\begin{equation}
    \rho_{s}(\mathbf{x},t) = q_{s} \int_{\Omega_{p}} f_{s}(\mathbf{x}, \mathbf{p}, t) \, d\mathbf{p}, \quad \mathbf{J}_{s}(\mathbf{x},t) = q_{s} \int_{\Omega_{p}} \frac{\mathbf{p}}{m_{s} \gamma_{s}} f_{s}(\mathbf{x}, \mathbf{p}, t) \, d\mathbf{p}. \label{eq:species charge + current densities integrals}
\end{equation}
Here, the integrals are taken over the velocity components of phase space, which we have denoted by $\Omega_{p}$. The remaining parameters $\epsilon_{0}$ and $\mu_{0}$ describe the permittivity and permeability of the media in which the fields propagate, which we take to be free-space. In free-space, Maxwell's equations move at the speed of light $c$, and we have the useful relation $c^{2} = 1/\left( \mu_0 \epsilon_0 \right)$. The last two equations, \eqref{eq:Gauss-E} and \eqref{eq:Gauss-B}, are known as the involutions.  The involutions are constraints placed on the fields in terms of initial conditions.  If these conditions are satisfied at $t=0$, then equations \eqref{eq:Farady} and \eqref{eq:Ampere} preserve these conditions for all future time.  It is imperative that numerical schemes for Maxwell's equations satisfy these involutions to numerically ensure that the method maintains charge conservation and prevents the appearance of ``magnetic monopoles". This requirement is one of the reasons we adopt a gauge formulation for Maxwell's equations, which is presented in the next section.

%
%
%
%
\subsection{A Hamiltonian Formulation of the Vlasov-Maxwell System under the Lorenz Gauge}
\label{subsec:2 model}

The particle method introduced in part I \cite{christliebPIC2022pt1} was obtained from a Hamiltonian formulation of the Vlasov-Maxwell system in which Maxwell's equations were written in terms of the scalar and vector potential. The adoption of a Hamiltonian formulation for the particles was motivated by the desire to avoid computing time derivatives of the vector potential. Using the Hamiltonian for a single relativistic particle in a potential field
\begin{equation}
    \label{eq:Relativistic Hamiltonian particle potential field}
    \mathcal{H} = \sqrt{c^2 \left( \mathbf{P} - q \mathbf{A}\right)^2 + \left(mc^2\right)^2} + q \phi,
\end{equation}
we obtained the system
\begin{empheq}[left=\empheqlbrace]{align}
    &\frac{d\mathbf{x}_{i}}{dt} = \frac{c^2 \left(\mathbf{P}_{i} - q_{i}\mathbf{A}\right)}{\sqrt{c^2\left(\mathbf{P}_{i} - q_{i} \mathbf{A}\right)^2 + \left(m_{i} c^2\right)^2}}, \label{eq:Position equation relativistic form} \\
    &\frac{d\mathbf{P}_{i}}{dt} = -q_{i} \nabla \phi + \frac{q_{i} c^2 \left( \nabla\mathbf{A}\right) \cdot \left(\mathbf{P}_{i} - q_{i}\mathbf{A}\right)}{\sqrt{c^2\left(\mathbf{P}_{i} - q_{i} \mathbf{A}\right)^2 + \left(m_{i} c^2\right)^2}}, \label{eq:Generalized momentum equation relativistic form} \\
    &\frac{1}{c^2} \partial_{tt} \phi -\Delta \phi =  \frac{1}{\epsilon_{0}} \rho, \label{eq:scalar potential eqn lorenz} \\ 
    &\frac{1}{c^2} \partial_{tt} \mathbf{A} -\Delta \mathbf{A}= \mu_{0} \mathbf{J}, \label{eq:vector potential eqn lorenz} \\ 
    &\frac{1}{c^2} \partial_{t} \phi + \nabla \cdot \mathbf{A} =0 \label{eq:Lorenz gauge condition},
\end{empheq}
where $i = 1, 2, \cdots, N_{p}$ is the index of the simulation particle, $c$ is the speed of light, $\epsilon_{0}$ and $\mu_{0}$ represent, respectively, the permittivity and permeability of free-space. Further, we have used $\phi$ to denote the scalar potential and $\mathbf{A}$ is the vector potential, which are related to the electric and magnetic fields via
\begin{equation}
    \label{eq:Convert potentials to EB}
    \mathbf{E} = -\nabla \phi - \partial_{t} \mathbf{A}, \quad \mathbf{B} = \nabla \times \mathbf{A}.
\end{equation}
Though not written explicitly, the potentials and gradients that appear in equations \eqref{eq:Position equation relativistic form} and \eqref{eq:Generalized momentum equation relativistic form} for particle $i$ are evaluated at the corresponding position given by $\mathbf{x}_{i}$. The evolution equations for the potentials \eqref{eq:scalar potential eqn lorenz} and \eqref{eq:vector potential eqn lorenz} consists of a system of four scalar wave equations that are coupled to the Lorenz gauge condition \eqref{eq:Lorenz gauge condition}. The wave equations for the potentials $\phi$ and $\mathbf{A}$ are equivalent to Maxwell's equations provided the gauge condition is satisfied. While other gauge conditions can be used, we choose to work under the Lorenz gauge simply because it allows for a relatively straightforward application of our previous work developed for scalar wave equations \cite{causley2017wave-propagation,causley2014higher}. However, it may be advantageous to consider other gauge conditions, such as the Coulomb gauge.

In the Hamiltonian formulation presented above, the velocity of particle $i$ is obtained from the generalized momentum using the definition
\begin{equation}
    \mathbf{v}_{i} = \frac{c^2 \left( \mathbf{P}_{i} - q_{i} \mathbf{A}\right)}{\sqrt{ c^2\left( \mathbf{P}_{i} - q_{i} \mathbf{A}\right)^2 + \left(m_{i}c_{i}^{2}\right)^2}}, \quad i = 1, 2, \cdots, N_{p}, \label{eq:velocity from generalized momentum}
\end{equation}
where the potential $\mathbf{A}$ is to be evaluated at the position of particle $i$. Finally, the charge density $\rho$ and current density $\mathbf{J}$ in equations \eqref{eq:scalar potential eqn lorenz} and \eqref{eq:vector potential eqn lorenz} are related by the continuity equation
\begin{equation} \label{eq:continuity-equation}
    \partial_t\rho + \nabla\cdot\textbf{J} = 0,
\end{equation}
and are computed according to
\begin{align}
    \rho \left( \mathbf{x}, t \right) &= \sum_{p=1}^{N_{p}} q_{p} S \left(\mathbf{x} - \mathbf{x}_{p}(t) \right), \label{eq:PIC charge density} \\
    \mathbf{J} \left( \mathbf{x}, t \right) &= \sum_{p=1}^{N_{p}} q_{p} \mathbf{v}_{p}(t) S \left(\mathbf{x} - \mathbf{x}_{p}(t) \right). \label{eq:PIC current density}
\end{align}
Here, $S$ is the particle shape function, which, in addition to gathering the particle data to the mesh, is responsible for interpolating fields on the mesh to the particles. The use of a shape function can be considered as a spatial regularization of the Klimontovich function
\begin{equation*}
    f(\mathbf{x}, \mathbf{v},t) \approx \sum_{p = 1}^{N_{p}} S \left(\mathbf{x} - \mathbf{x}_{p}(t)\right)\delta \left(\mathbf{v} - \mathbf{v}_{p}(t)\right).
\end{equation*}
The definitions \eqref{eq:PIC charge density} and \eqref{eq:PIC current density} follow directly from the evaluation of the moments \eqref{eq:species charge + current densities integrals}, where the reference to the species is replaced with the more generic label $p$ that can be attributed to a particle of any species. In this work, the scatter and gather steps for the current density \eqref{eq:PIC current density} use tensor products of quadratic basis spline (B-spline) functions (see \cite{JPVerboncoeur2005review} or section 8.8 of \cite{BirdsallLangdon} for more details, or Chapter 9 of \cite{deBoor} for a much more rigorous discussion of the derivation of B-splines). The charge density, $\rho$, on the other hand will not be computed using \eqref{eq:PIC charge density}, but rather computed using the continuity equation \eqref{eq:continuity-equation}. The spatial component of the continuity equation may be solved in a variety of ways, but this work will use a spectral discretization. This will be examined briefly in section \ref{sec:4 Particle Updates}. The methods we use to discretize the system in time are presented in the next section.

%
%
\subsection{Summary}
\label{subsec:2 Summary}

In this section, we described the Vlasov-Maxwell system in a Hamiltonian formulation, transformed by the Lorenz gauge. Any PIC method simulating this system requires two main components: (1) a wave solver to update the fields and (2) a particle pusher to update the particle coordinates. We are now ready to discuss the main focus of this work, which is the wave solvers and their relationship with the continuity equation and Lorenz gauge, to which we now turn.

%
%
%
%
\section{A Family of Wave Solvers}
\label{sec:3 wave solvers}

In this section we develop a family of semi-discrete wave solvers, all of which will be shown to satisfy the semi-discrete Lorenz gauge if the semi-discrete continuity equation is satisfied.  We will establish that this is true for a variety of time-centered methods, BDF methods, and any $s$-stage DIRK method.  The converse will be shown for all methods except DIRK-$s$.  There does exist a proof for the converse for DIRK-2, but it is lengthy and not used in the remainder of the paper, so we exclude it.  With the exception of DIRK methods, these methods will also be shown to satisfy Gauss's law for electricity if they satisfy the semi-discrete Lorenz gauge. All wave solvers are approached in the context of the method-of-lines-transpose to derive self-consistent formulations. In particular, the semi-discrete models are obtained by only discretizing in time, which results in a collection of elliptic problems that we need to address once the formulation is identified.  

In this work, we make use of the FFT as a simple way to compute solutions to these boundary value problems.  As is common practice with high-order and spectral methods combined with particle methods, we also make use of second-order particle weightings to map the current density, $\mathbf{J}$, to the mesh \cite{Benedetti_2008, Shalaby_2017, Davidson_2015, Sadeghirad_2024}. This provides a simple platform to demonstrate that the theory presented holds at the fully-discrete level. Indeed, as will be shown in the results (section \ref{sec:5 Numerical results}), the methods constructed with CDF-2, BDF-1, BDF-2 and DIRK-2 integrators satisfy the Lorenz gauge condition at the fully-discrete level. Additionally, the CDF-2, BDF-1, and BDF-2 methods will also satisfy Gauss's law. 

In all cases, the key component that makes this possible is the idea of first mapping the current to the mesh and then directly solving the continuity equation for the charge density on the mesh. For both the time-centered and BDF methods, the proofs are constructive, whereas for the DIRK methods it is an inductive proof. In what follows, we will take both $\mathbf{J}^{n}$ and $\mathbf{J}^{n+1}$ as given and compute what is needed for the update using this information. In section \ref{sec:4 Particle Updates} on the overall algorithm, we provide additional details regarding the construction of $\mathbf{J}^{\{n,n+1\}}$.

%
%
%
%
\subsection{A Semi-discrete Time-centered Method for the Lorenz Gauge Formulation}
\label{subsec:semi-discrete CDF}





To start, consider the semi-discrete treatment of the Lorenz gauge formulation, as in \cite{christliebPIC2022pt2}, in which second-order approximations are used for $\partial_{t}$ and $\partial_{tt}$. As we will prove in this section, the satisfaction of the semi-discrete gauge condition as well as the involutions ($\nabla\cdot\mathbf{B}=0$ and $\nabla\cdot\mathbf{E}=\frac{\rho}{\epsilon_0}$) can be ensured using the time-centered difference approximations
$$
\partial_t u^{n+\frac{1}{2}} = \frac{u^{n+1}-u^{n}}{\Delta t} + \mathcal{O}(\Delta t^2), \quad \partial_{tt} u^n = \frac{u^{n+1}-2u^n+u^{n-1}}{\Delta t^2} +\mathcal{O}(\Delta t^2).
$$
A critical observation is that two applications of the first difference approximation to a discrete data set produces the discrete second order discrete second difference in time, i.e.,
$$
\partial_{tt} u = \frac{\partial_t u^{n+\frac{1}{2}}-\partial_t u^{n-\frac{1}{2}}}{\Delta t} + \mathcal{O}(\Delta t^2)
= \frac{\dfrac{u^{n+1}-u^{n}}{\Delta t}-\dfrac{u^{n}-u^{n-1}}{\Delta t}}{\Delta t} + \mathcal{O}(\Delta t^2).$$
This leads to the semi-discrete system in an implicit form 
\begin{empheq}[left=\empheqlbrace]{align}
    &\frac{\phi^{n+\frac{3}{2}} -2 \phi^{n+\frac{1}{2}} + \phi^{n-\frac{1}{2}}}{c^2\Delta t^2} -\Delta \left (\frac{ \phi^{n+\frac{3}{2}}+ \phi^{n-\frac{1}{2}}}{2} \right )= \frac{1}{\epsilon_0} \frac{\rho^{n+\frac{3}{2}}+\rho^{n-\frac{1}{2}}}{2}, \label{eq:semi-discrete phi2} \\
    &\frac{\mathbf{A}^{n+1} - 2 \mathbf{A}^{n} + \mathbf{A}^{n-1}}{c^2\Delta t^2} - \Delta \left ( \frac{\mathbf{A}^{n+1}+\mathbf{A}^{n-1}}{2} \right ) = \mu_0\frac{ \mathbf{J}^{n+1}+\mathbf{J}^{n-1}}{2}, \label{eq:semi-discrete A2} \\
    &\frac{\phi^{n+\frac{3}{2}} - \phi^{n+\frac{1}{2}}}{c^{2} \Delta t} + \nabla \cdot \mathbf{A}^{n+1} = 0, \label{eq:semi-discrete Lorenz}
\end{empheq}
a semi-implicit form
\begin{empheq}[left=\empheqlbrace]{align}
    &\frac{\phi^{n+\frac{3}{2}} - 2 \phi^{n+\frac{1}{2}} + \phi^{n-\frac{1}{2}}}{c^2\Delta t^2} - \Delta \left (\frac{ \phi^{n+\frac{3}{2}} + \phi^{n-\frac{1}{2}}}{2} \right )= \frac{1}{\epsilon_0} \rho^{n+\frac{1}{2}}, \label{eq:semi-discrete phi3} \\
    &\frac{\mathbf{A}^{n+1} - 2 \mathbf{A}^{n} + \mathbf{A}^{n-1}}{c^2\Delta t^2} - \Delta \left ( \frac{\mathbf{A}^{n+1}+\mathbf{A}^{n-1}}{2} \right ) = \mu_0 \mathbf{J}^{n}, \label{eq:semi-discrete A3} \\
    &\frac{\phi^{n+\frac{3}{2}} - \phi^{n+\frac{1}{2}}}{c^{2} \Delta t} + \nabla \cdot \mathbf{A}^{n+1} = 0, \label{eq:semi-discrete Lorenz3}
\end{empheq}
and, alternatively, an explicit form
\begin{empheq}[left=\empheqlbrace]{align}
    &\frac{\phi^{n+\frac{3}{2}} - 2 \phi^{n+\frac{1}{2}} + \phi^{n-\frac{1}{2}}}{c^2\Delta t^2} -\Delta \phi^{n+\frac{1}{2}} = \frac{1}{ \epsilon_0} \rho^{n+\frac{1}{2}}, \label{eq:semi-discrete phi} \\
    &\frac{\mathbf{A}^{n+1} - 2 \mathbf{A}^{n} + \mathbf{A}^{n-1}}{c^2\Delta t^2} - \Delta \mathbf{A}^{n} = \mu_0 \mathbf{J}^{n}, \label{eq:semi-discrete A} \\
    &\frac{\phi^{n+\frac{3}{2}} - \phi^{n+\frac{1}{2}}}{c^{2} \Delta t} + \nabla \cdot \mathbf{A}^{n+1} = 0, \label{eq:semi-discrete Lorenz1}
\end{empheq}
where here we assume $\mathbf{J}^{\{n,n+1\}}$ is given and use this to compute $\rho^{n+\frac{3}{2}}$ and the fields for the particle update.  All examples we explore in section \ref{sec:5 Numerical results} are implicit and compute an approximation of $\mathbf{J}^{n+1}$ (see section \ref{sec:4 Particle Updates} for more details).

Of particular relevance to this paper, and our previous work \cite{christliebPIC2022pt2},  is a theorem that establishes an equivalence, at the semi-discrete level, between the continuity equation
\begin{equation}
    \label{eq:staggered semi-discrete continuity equation}
    \frac{\rho^{n+\frac{3}{2}} - \rho^{n+\frac{1}{2}}}{\Delta t} + \nabla \cdot \mathbf{J}^{n+1} = 0,
\end{equation}
and the Lorenz gauge condition \eqref{eq:semi-discrete Lorenz} (see \cite{christliebPIC2022pt2}, Theorem 2.1). To establish similar results for these methods, we will need to identify groupings of terms. For convenience, we shall introduce the notation
\begin{equation}
    \epsilon_1^{k} = \frac{\phi^{k+\frac{1}{2}} - \phi^{k-\frac{1}{2}}}{c^{2} \Delta t} + \nabla \cdot \mathbf{A}^{k}, \;\;\; \epsilon_2^{k} = \frac{\rho^{k+\frac{1}{2}} - \rho^{k-\frac{1}{2}}}{\Delta t} + \nabla \cdot \mathbf{J}^{k},
\end{equation}
to simplify the presentation. Our first task is to establish the interdependence of satisfying the gauge condition with satisfying the continuity equation.  Then we can establish that if our method is time-consistent, and the gauge condition is satisfied, then we will ensure that Gauss's law for electricity is satisfied.  In the semi-discrete case, the derivatives are analytical, meaning Gauss's law for magnetism is satisfied by definition.  In the next section, we will examine how this looks in the fully-discrete case, and introduce the method for ensuring satisfaction of the continuity equation, and hence we will establish that the method preserves the gauge condition and the involutions.

We now define operator notation that will be used in the proofs of the semi-discrete theorem for the fully-implicit update. Consider now the linear wave equation
\begin{equation}
    \label{eq:diffusion equation}
   \frac{1}{c^2} \partial_{tt} u -  \Delta u = S.
\end{equation}
The coefficient $c > 0$ and assumed to be a constant. Additionally, we denote  $S(\mathbf{x},t)$ as the source function. An implicit second-order Crank-Nicolson method for $u^{n+1}$ centered at $t = t^{n}$ is
\begin{equation}
    \label{eq:semi-discrete wave}
    \left(\mathcal{I} - \frac{1}{\alpha^{2}} \Delta \right)\left( u^{n+1}+u^{n-1} \right) =  2u^{n} + \frac{1}{\alpha^{2}} \Gamma^n, 
\end{equation}
where we use $\alpha = \sqrt{2}/(c \Delta t)$ and $\Gamma^n= S^{n+1}+S^{n-1}$. We define 
$$
\mathcal{L}(\cdot) = \left(\mathcal{I} - \frac{1}{\alpha^{2}} \Delta \right) (\cdot)
$$
and denote the inverse as $\mathcal{L}^{-1}(\cdot)$.  The inverse can be computed with any number of methods. For our purposes, we will compute the inverse using the FFT.

\begin{lemma} \label{lemma:Implicit method-Green’s Function}
(Implicit method-Green’s Function)  For the implicit formulation in \eqref{eq:semi-discrete phi2}-\eqref{eq:semi-discrete A2}, the semi-discrete Lorenz gauge condition \eqref{eq:semi-discrete Lorenz} satisfies
\begin{equation}
    \label{eq:implicit semi-discrete Lorenz residual update}
    \epsilon_1^{n+1} = - \epsilon_1^{n-1} + \mathcal{L}^{-1} \left( 2 \epsilon_1^{n}  + \frac{\mu_0}{\alpha^2}\left ( \epsilon_2^{n+1}+\epsilon_2^{n-1}\right )\right).
\end{equation}
\label{lemma:implicit - residual equivalence}
\end{lemma}
\begin{proof}  This proof is a direct construction.  
We start with equation \eqref{eq:semi-discrete phi2}. Inverting the modified Helmholtz operator $\mathcal{L}$ gives the scalar potential as
\begin{equation}
    \phi^{n+\frac{3}{2}} = -\phi^{n-\frac{1}{2}} +  \mathcal{L}^{-1}  \left( 2 \phi^{n+\frac{1}{2}} +  \frac{1}{\alpha^{2} \epsilon_0} \left ( \rho^{n+\frac{3}{2}}+\rho^{n-\frac{1}{2}}\right) \right), \label{eq:semi-discrete psi update}
\end{equation}
which can be shifted to obtain
\begin{equation}
    \label{eq:semi-discrete psi update at t_n}
    \phi^{n+\frac{1}{2}} = -\phi^{n-\frac{3}{2}} +  \mathcal{L}^{-1}  \left( 2 \phi^{n-\frac{1}{2}} +  \frac{1}{\alpha^{2} \epsilon_0} \left ( \rho^{n-\frac{1}{2}}+\rho^{n-\frac{3}{2}}\right) \right).
\end{equation}
Next, we take the divergence of $\mathbf{A}$ in equation \eqref{eq:semi-discrete A2} and find that
\begin{equation*}
    \mathcal{L} \left( \nabla \cdot \mathbf{A}^{n+1} + \nabla \cdot \mathbf{A}^{n-1}\right) = 2 \nabla \cdot \mathbf{A}^{n}  + \frac{\mu_0}{\alpha^2} \left ( \nabla \cdot \mathbf{J}^{n+1} + \nabla \cdot \mathbf{J}^{n-1} \right ) .
\end{equation*}
Formally inverting the operator $\mathcal{L}$, we obtain the relation
\begin{equation}\label{eq:semi-discrete div A update}
    \nabla \cdot \mathbf{A}^{n+1}  = - \nabla \cdot \mathbf{A}^{n-1}+ \mathcal{L}^{-1} \left( 2 \nabla \cdot \mathbf{A}^{n}  + \frac{\mu_0}{\alpha^2} \left ( \nabla \cdot \mathbf{J}^{n+1} + \nabla \cdot \mathbf{J}^{n-1} \right ) \right).
\end{equation}

With the aid of equations \eqref{eq:semi-discrete psi update}, \eqref{eq:semi-discrete psi update at t_n}, and \eqref{eq:semi-discrete div A update}, along with the linearity of the operator $\mathcal{L}$, a direct calculation reveals that the residual at time level is given by
\begin{align*}
    \frac{\phi^{n+\frac{3}{2}} - \phi^{n+\frac{1}{2}}}{c^{2} \Delta t} &+ \nabla \cdot \mathbf{A}^{n+1} = \\ &- \epsilon_1^{n-1}+ \mathcal{L}^{-1} \Bigg( 2\epsilon_1^{n} + \left (\frac{1}{\alpha^{2} \epsilon_0} \left ( \frac{\rho^{n+\frac{3}{2}}-\rho^{n+\frac{1}{2}}}{c^2\Delta t}\right)+\frac{\mu_0}{\alpha^2}  \nabla \cdot \mathbf{J}^{n+1}  \right) + \left (\frac{1}{\alpha^{2} \epsilon_0} \left ( \frac{\rho^{n-\frac{1}{2}}-\rho^{n-\frac{3}{2}}}{c^2\Delta t}\right)+\frac{\mu_0}{\alpha^2}  \nabla \cdot \mathbf{J}^{n-1}  \right)  \Bigg).
\end{align*} 
From these calculations, we see that the above expression is the same \eqref{eq:implicit semi-discrete Lorenz residual update}, with the observation that we have used the following relation $c^2 = 1/(\mu_0 \epsilon_0)$.
We observe that the corresponding consistent semi-discrete continuity equation
\begin{equation}
    \label{eq:semi-discrete continuity}
    \frac{\rho^{n+\frac{3}{2}} - \rho^{n+\frac{1}{2}}}{\Delta t} + \nabla \cdot \mathbf{J}^{n+1} = \epsilon_2^{n+1},
\end{equation}
acts as a source for the residual, which completes the proof. 
\end{proof}
With the aid of Lemma \ref{lemma:implicit - residual equivalence}, we are now prepared to prove the following theorem that establishes the time-consistency of the semi-discrete system.
\begin{theorem}\label{thm:Implicit method-Green’s Function} (Implicit method-Green’s Function) Assuming $\epsilon_1^{n} = \epsilon_2^{n} = 0, \, \forall \, n \leq 0$, the semi-discrete Lorenz gauge formulation of Maxwell's equations \eqref{eq:semi-discrete phi2}-\eqref{eq:semi-discrete Lorenz} is time-consistent in the sense that the semi-discrete Lorenz gauge condition \eqref{eq:semi-discrete Lorenz} is satisfied at any discrete time if and only if the corresponding semi-discrete continuity equation \eqref{eq:semi-discrete continuity} is also satisfied.
\end{theorem}
\begin{proof}

We use a simple inductive argument to prove both directions. In the case of the forward direction, we assume that the semi-discrete gauge condition is satisfied, so $\epsilon_1^{n} \equiv 0, \quad \forall n \geq 0$. From our assumption $\epsilon_2^{-1} \equiv 0$. Combining this with equation \eqref{eq:implicit semi-discrete Lorenz residual update} established by Lemma \ref{lemma:implicit - residual equivalence}, it follows that the next time level satisfies
\begin{equation*}
    0 = \mathcal{L}^{-1} \Bigg( \frac{\mu_0}{\alpha^2} \left( \frac{\rho^{\frac{1}{2}} - \rho^{-\frac{1}{2}}}{\Delta t} + \nabla \cdot \mathbf{J}^{0} \right) \Bigg).
\end{equation*}
Applying the operator $\mathcal{L}$ to both sides leads to
\begin{equation*}
    \frac{\rho^{\frac{1}{2}} - \rho^{-\frac{1}{2}}}{\Delta t} + \nabla \cdot \mathbf{J}^{0}  = 0.
\end{equation*}
This argument can be iterated $n$ times to show that
\begin{equation*}
    \frac{\rho^{n+\frac{1}{2}} - \rho^{n-\frac{1}{2}}}{\Delta t} + \nabla \cdot \mathbf{J}^{n} = 0,
\end{equation*}
also holds at any discrete time level $n$, which establishes the forward direction.

A similar argument can be used for the converse. Here, we show that if the semi-discrete continuity equation \eqref{eq:semi-discrete continuity} is satisfied for any time level $n$, i.e., $\epsilon_2^{n}  \equiv 0$, then the residual for the semi-discrete gauge condition also satisfies $\epsilon_1^{n+1} \equiv 0$. First, we assume that the initial data and starting values satisfy $\epsilon_1^{-1} \equiv \epsilon_1^{0} \equiv 0$. Appealing to equation \eqref{eq:implicit semi-discrete Lorenz residual update} with this initial data, it is clear that after a single time step, the residual in the gauge condition satisfies
\begin{equation*}
    \epsilon_1^{1} =  - \epsilon_1^{-1}+\mathcal{L}^{-1} \Bigg( 2\epsilon_1^{0} + \frac{\mu_0}{\alpha^2} \left( \epsilon_2^1+ \epsilon_2^{-1}\right) \Bigg) \equiv \mathcal{L}^{-1} ( 0 ).
\end{equation*}
This argument can also be iterated $n$ more times to obtain the result, which finishes the proof.
\end{proof}

\begin{lemma} (Semi-Implicit method-Green’s Function)
\label{lemma:semi-implicit residual equivalence}
For the semi-implicit formulation in \eqref{eq:semi-discrete phi3}-\eqref{eq:semi-discrete A3}, the semi-discrete Lorenz gauge condition \eqref{eq:semi-discrete Lorenz3} satisfies
\begin{equation}
    \label{eq:semi-implicit semi-discrete Lorenz residual update}
    \epsilon_1^{n+1} = - \epsilon_1^{n-1} + \mathcal{L}^{-1} \left( 2 \epsilon_1^{n}  + \frac{\mu_0}{\alpha^2} \epsilon_2^{n}\right).
\end{equation}
\end{lemma}
\begin{proof}
The proof is nearly identical to that of Lemma \ref{lemma:Implicit method-Green’s Function}, so we exclude it.
\end{proof}

\begin{theorem}\label{thm:semi-implicit method-Green’s Function}
(Semi-Implicit method-Green’s Function)
Assuming $\epsilon_1^{n} = \epsilon_2^{n} = 0, \quad \forall n \leq 0$, the semi-discrete Lorenz gauge formulation of Maxwell's equations \eqref{eq:semi-discrete phi3}-\eqref{eq:semi-discrete Lorenz3} is time-consistent in the sense that the semi-discrete Lorenz gauge condition \eqref{eq:semi-discrete Lorenz3} is satisfied at any discrete time if and only if the corresponding semi-discrete continuity equation \eqref{eq:semi-discrete continuity} is also satisfied.
\end{theorem}
\begin{proof}
The proof is nearly identical to that of Theorem \ref{thm:Implicit method-Green’s Function}, so we exclude it.
\end{proof}

\begin{lemma} \label{lemma:explicit - residual equivalence}
(Explicit method)
For the explicit formulation in \eqref{eq:semi-discrete phi}-\eqref{eq:semi-discrete A}, the semi-discrete Lorenz gauge condition \eqref{eq:semi-discrete Lorenz1} satisfies
\begin{equation}
    \label{eq:explicit semi-discrete Lorenz residual update}
    \epsilon_1^{n+1} = 2 \epsilon_1^{n} - \epsilon_1^{n-1} + \left ( c^2 \Delta t^2 \right ) \Delta \epsilon_1^{n}  +     \left ( c^2 \Delta t^2 \right )\mu_0 \epsilon_2^{n}.
\end{equation}
\end{lemma}
\begin{proof}
Start with 
\begin{equation*}
 \epsilon_1^{n+1}=\frac{\phi^{n+\frac{3}{2}} - \phi^{n+\frac{1}{2}}}{c^{2} \Delta t} + \nabla \cdot \mathbf{A}^{n+1},
\end{equation*}
and substitute equation \eqref{eq:semi-discrete phi} in for $\phi^{n+\frac{3}{2}}$ and $\phi^{n+\frac{1}{2}}$, then substitute equation \eqref{eq:semi-discrete A} in for $\mathbf{A}^{n+1}$. We have
\begin{align*}
    \frac{\phi^{n+\frac{3}{2}} - \phi^{n+\frac{1}{2}}}{c^2\Delta t}& = \\ 
    &\frac{2 \phi^{n+\frac{1}{2}} - \phi^{n-\frac{1}{2}}  -(c^2\Delta t^2)\Delta \phi^{n+\frac{1}{2}} + (c^2\Delta t^2)\frac{1}{\epsilon_0} \rho^{n+\frac{1}{2}}
    -\left (2 \phi^{n-\frac{1}{2}} - \phi^{n-\frac{3}{2}}  - (c^2\Delta t^2)\Delta \phi^{n-\frac{1}{2}}+ (c^2\Delta t^2)\frac{1}{\epsilon_0} \rho^{n-\frac{1}{2}} \right )}{c^2 \Delta t},
\end{align*}
and 
\begin{equation*}
\nabla \cdot \mathbf{A}^{n+1}= 2 \nabla \cdot\mathbf{A}^{n} - \nabla \cdot\mathbf{A}^{n-1}- c^2\Delta t^2 \Delta \nabla \cdot\mathbf{A}^{n} + c^2 \Delta t^2 \mu_0 \nabla \cdot \mathbf{J}^{n}.
\end{equation*}
After grouping terms in the form of $\epsilon_1^{k}$ and $\epsilon_2^{k}$ as in Lemma \ref{lemma:Implicit method-Green’s Function}, we arrive at \eqref{eq:explicit semi-discrete Lorenz residual update}.  This completes the proof.
\end{proof}

\begin{theorem} \label{thm:explicit method-Green’s Function} (Explicit method) 
Assuming $\epsilon_1^{n} = \epsilon_2^{n} = 0, \, \forall \, n \leq 0$, the explicit semi-discrete Lorenz gauge formulation of Maxwell's equations \eqref{eq:semi-discrete phi}-\eqref{eq:semi-discrete Lorenz1} is time-consistent in the sense that the semi-discrete Lorenz gauge condition \eqref{eq:semi-discrete Lorenz1} is satisfied at any discrete time if and only if the corresponding semi-discrete continuity equation \eqref{eq:semi-discrete continuity} is also satisfied.
\end{theorem}
\begin{proof}
The proof is an identical construction to that in Theorem \ref{thm:Implicit method-Green’s Function}.
\end{proof}

\begin{theorem}
    Under the assumption that $\epsilon_1^k = \epsilon_2^k=0, \, \forall \, k \leq n+1$, the implicit semi-discrete methods \eqref{eq:semi-discrete phi2}-\eqref{eq:semi-discrete Lorenz} satisfy Gauss's Law
    \begin{equation}
        \nabla \cdot \mathbf{E}^{n+\frac{1}{2}} = \frac{1}{\epsilon_0}\left(\frac{\rho^{n+\frac{3}{2}} + \rho^{n-\frac{1}{2}}}{2}\right).
    \end{equation}
\end{theorem}

\begin{proof}
    Starting with the identity $\mathbf{E}=-\nabla\phi -\partial_t \mathbf{A}$, we have
    \begin{equation}
        \textbf{E}^{n+\frac{1}{2}} = -\nabla\left(\frac{\phi^{n+\frac{3}{2}}+\phi^{n-\frac{1}{2}}}{2}\right) - \frac{\textbf{A}^{n+1} - \textbf{A}^{n}}{\Delta t}.
    \end{equation}
    Taking the divergence, we obtain
    \begin{equation}
        \nabla \cdot \textbf{E}^{n+\frac{1}{2}} = -\Delta\left(\frac{\phi^{n+\frac{3}{2}}+\phi^{n-\frac{1}{2}}}{2}\right) - \frac{\nabla\cdot\textbf{A}^{n+1} - \nabla\cdot\textbf{A}^{n}}{\Delta t}.
    \end{equation}
    Then, from the assumption, we know that $$\epsilon_1^{n} = 0 \implies \frac{1}{c^2}\frac{\phi^{n+\frac{1}{2}} - \phi^{n-\frac{1}{2}}}{\Delta t} + \nabla\cdot\textbf{A}^{n} = 0.$$ From this, we see that
    \begin{align}
        \begin{split}
            \nabla \cdot \textbf{E}^{n+\frac{1}{2}} &= -\Delta\left(\frac{\phi^{n+\frac{3}{2}}+\phi^{n-\frac{1}{2}}}{2}\right) - \frac{\left(-\dfrac{1}{c^2}\dfrac{\phi^{n+\frac{3}{2}} - \phi^{n+\frac{1}{2}}}{\Delta t}\right) - \left(-\dfrac{1}{c^2}\dfrac{\phi^{n+\frac{1}{2}} - \phi^{n-\frac{1}{2}}}{\Delta t}\right)}{\Delta t} \\
            &= -\Delta\left(\frac{\phi^{n+\frac{3}{2}}+\phi^{n-\frac{1}{2}}}{2}\right) + \frac{1}{c^2}\frac{\phi^{n+\frac{3}{2}} - 2\phi^{n+\frac{1}{2}} + \phi^{n-\frac{1}{2}}}{\Delta t^2} \\
            &= \frac{1}{\epsilon_0}\left(\frac{\rho^{n+\frac{3}{2}} + \rho^{n-\frac{1}{2}}}{2}\right).
        \end{split}
    \end{align}
    The last step comes from the fully implicit CDF-2 scalar wave equation \eqref{eq:semi-discrete phi2}. We thus see Gauss's Law is satisfied.
\end{proof}

\begin{theorem}  Under the assumption that $\epsilon_2^{k} = 0$, for $k \leq n+1$, both the explicit and semi-implicit semi-discrete methods satisfy Gauss's law
\begin{equation}
    \nabla \cdot \mathbf{E}^{n+\frac{1}{2}} = \frac{\rho^{n+\frac{1}{2}}}{\epsilon_0}.
\end{equation}
\end{theorem}
\begin{proof}
    We start by noting, under the assumption that  $\epsilon_2^{k} = 0$, we have   $\epsilon_1^{k}=0$, for $k=n+1$, from the previous theorems. We carry out the proof for the explicit case and note that the semi-implicit case is identical except for the fact that $\nabla \phi^{n+\frac{1}{2}}$ will be replaced by $\nabla \left(\frac{1}{2}(\phi^{n+\frac{3}{2}}+\phi^{n-\frac{1}{2}})\right)$ in the identity $$\mathbf{E}=-\nabla\phi -\partial_t \mathbf{A}.$$ Evaluating this identity at time $t^{n+\frac{1}{2}}$ and making use of the semi-discrete potentials, we find
    \begin{equation*}
        \mathbf{E}^{n+\frac{1}{2}} = -\nabla \phi^{n+\frac{1}{2}} - \frac{\mathbf{A}^{n+1}-\mathbf{A}^{n}}{\Delta t}.
    \end{equation*}                            
    Again, taking the divergence gives
    \begin{equation*}
        \nabla \cdot \mathbf{E}^{n+\frac{1}{2}} = -\Delta \phi^{n+\frac{1}{2}} - \frac{\nabla \cdot\mathbf{A}^{n+1}-\nabla \cdot\mathbf{A}^{n}}{\Delta t}.
    \end{equation*}
    given that $\epsilon_1^{n+1}=0$, then we can replace $\nabla \cdot\mathbf{A}^{n+1}$ with $-\frac{\phi^{n+\frac{3}{2}} - \phi^{n+\frac{1}{2}}}{c^{2} \Delta t}$ and $\nabla \cdot\mathbf{A}^{n}$ with $-\frac{\phi^{n+\frac{1}{2}} - \phi^{n-\frac{1}{2}}}{c^{2} \Delta t}$, so
    \begin{equation*}
        \nabla \cdot \mathbf{E}^{n+\frac{1}{2}} = -\Delta \phi^{n+\frac{1}{2}} - \frac{\phi^{n+\frac{3}{2}}-2\phi^{n+\frac{1}{2}}+\phi^{n-\frac{1}{2}}}{c^2\Delta t^2}=\frac{\rho^{n+\frac{1}{2}}}{\epsilon_0}.
    \end{equation*}
    Hence, we arrive at the conclusion of the theorem.
\end{proof}

\begin{remark} The semi-discrete formulation method automatically satisfies $\nabla\cdot \mathbf{B}=\nabla\cdot (\nabla\times \mathbf{A})=0.$ 
\end{remark}

\begin{remark} All three time-consistent methods satisfy all of the involutions and the gauge condition, as long as the method satisfies the semi-discrete continuity condition.
\end{remark}

\begin{remark}
    For all the methods we discuss, we incorporate the improved asymmetric Euler method for our particle update method. A time-centered wave advance opens the door for time-centering the particle pusher, which increases the accuracy of the solver and can be used to explore aspects of symplecticity. This exploration is beyond the scope of this paper but will be considered in future work.
\end{remark}

%
%
\subsection{The General Semi-Discrete Formulation for BDF Methods Under the Lorenz Gauge}
\label{subsec:3 BDFk}

In \cite{christliebPIC2022pt2}, we introduced a lemma and two theorems that showcased properties associated with the first-order BDF method in a semi-discrete setting. That is, conservation of charge is satisfied if and only if the semi-discrete Lorenz gauge condition is satisfied, and that Gauss's law for electricity is satisfied if said gauge condition is.  In this section, we generalize these properties to any BDF methods constructed with uniform step sizes.  Vital to the original is the consistent nature of the temporal derivatives, that is, applying the first order derivative BDF method twice results in the second order derivative BDF method in the sense that $D_{t}[D_{t}[\cdot]] = D_{t}^{2}[\cdot]$; the generalization will also take advantage of this.  We start with a time-consistent $k$-step BDF formulation
\begin{align}
    &\frac{du}{dt} = \frac{1}{\Delta t}\sum_{i=n-k}^{n}a_{n-i}u^{i} + O(\Delta t^p), \label{eq:BDFk-first-derivative} \\
    &\frac{d^2u}{dt^2} = \frac{1}{\Delta t^2}\sum_{i=n-k}^{n}a_{n-i}\sum_{j=i-k}^{i}a_{i-j}u^{j}+O(\Delta t^p). \label{eq:BDFk-second-derivative}
\end{align}
Note that BDF-1 through BDF-6 have this form.

For purposes of bookkeeping later on, we now prove a lemma regarding the structure of the coefficients in the method, which, absorbing the $\Delta t$ terms into the coefficients, we write as
\begin{equation}
    \sum_{i=n-k}^{n}{a_{n-i}\sum_{j=i-k}^{i}a_{i-j}u^{j}} \equiv \sum_{i=n-2k}^{n}C_iu^i \label{eq:BDFk-update-u}.
\end{equation}
This term arises in the right hand side  of the implicit solution to semi-discrete Maxwell's equations.  It is important that we sort out this term first, as it is embedded inside of other backwards differences we encounter when doing the constructive proof for the theorems regarding preservation of the gauge condition, continuity equation, and Gauss's law. 

\begin{lemma} \label{lemma:BDFk-summation-reordering}

    We consider the BDF-k method applied to $u^n$ twice:
    \begin{equation}
        \sum_{i=n-k}^{n}a_{n-i}\sum_{j=i-k}^{i}a_{i-j}u^j.
    \end{equation}
    Considering \eqref{eq:BDFk-update-u}, the coefficient $C_i$ for a term $u^i$ at the corresponding 
     $\left(n-i\right)$th time level below $n$ takes the form
    \begin{equation} \label{eq:BDFk-coefficients}
        C_i \coloneqq \begin{cases}
                \sum\limits_{j = 0}^{n-i}a_{j}a_{n-i-j}, \quad n-i \leq k, \\
                \sum\limits_{j = 0}^{2k-(n-i)}a_{k-j}a_{n-i-k+j}, \quad n-i > k.
            \end{cases}
    \end{equation}
    
\end{lemma}

\begin{proof}

    We prove by induction. The base case $k=1$ results in the sum
    \begin{equation}
        a_0a_0u^n + (a_0a_1 + a_1a_0)u^{n-1} + a_1a_1u^{n-2}.
    \end{equation}
    This clearly has the form given in equation \eqref{eq:BDFk-coefficients}.
    
    For the inductive step, we assume the hypothesis is true for $k \leq m$. Next, let us consider $k=m+1$ for which we have the summation
    \begin{equation}
        \sum_{j=n-\left(m+1\right)}^{n}a_{n-j}\sum_{\ell=j-\left(m+1\right)}^{j}a_{j-\ell}u^\ell.
    \end{equation}
    We can easily peel the outer $(m+1)$-th term out, which gives
    \begin{equation}
        \left( \sum_{j=n-m}^{n}a_{n-j}\sum_{\ell=j-\left(m+1\right)}^{j}a_{j-\ell}u^\ell \right) + \left( a_{m+1}\sum_{j=n-2(m+1)}^{n-(m+1)}a_{n-(m+1)-j}u^j \right).
    \end{equation} 
    Similarly, we then peel the inner $(m+1)$-th term out to obtain
    \begin{equation}
        \left( \sum_{j=n-m}^{n}a_{n-j}\sum_{\ell=j-m}^{j}a_{j-\ell}u^\ell \right) + \left( a_{m+1}\sum_{j=n-2(m+1)}^{n-(m+1)}a_{n-(m+1)-j}u^{j} \right) + \left( a_{m+1}\sum_{j=n-m}^{n}a_{n-j}u^{j-(m+1)} \right).
    \end{equation}
    To prove the lemma, we must consider three cases, namely, $n-i < m+1$, $n-i = m+1$, and $n-i > m+1$.

    Consider $C_i$ for $n-i < m+1$. We know from the inductive hypothesis that the first nested sum is equal to $$\sum_{\substack{j=0 \\ n-i < m+1}}^{n-i}a_j a_{n-i-j}.$$ The second summation contains $u^j$ for $n-2(m+1) \leq j \leq n-(m+1)$. Since $n-(m+1) < i$, this summation cannot contain $u^i$. The third summation contains $u^j$ for $n-m \leq j \leq n$, and therefore contains $u^i$ when $j-(m+1) = i$, or $j=i+(m+1)$. So
    \begin{align*}
        C_i &= \sum_{\substack{j=0 \\ n-i < m+1}}^{n-i}a_j a_{n-i-j} + a_{n-i-(m+1)}a_{m+1}, \\
            &= \sum_{j=0}^{n-i}a_ja_{n-i-j}.
    \end{align*}

    Next, consider $C_i$ for $n-i = m+1$. Again, from the inductive hypothesis, the first nested sum is equal to $$\sum_{\substack{j=0 \\ n-i < m+1}}^{n-i}a_j a_{n-i-j}.$$ In the second summation, the term containing $u^{i}$ occurs when $j = i$ or $j = n-(m+1)$. Therefore, the coefficients in second sum reduce to
    \begin{equation*}
        a_{m+1}\sum_{j=n-2(m+1)}^{n-(m+1)}a_{n-(m+1)-j} = a_{m+1} a_{0}.
    \end{equation*}
    Similar, we find that the terms containing $u^{i}$ in third summation reduce to
    \begin{equation*}
        \sum_{j=n-m}^{n}a_{n-j}a_{m+1}u^{j-(m+1)} = a_{0}a_{m+1}.
    \end{equation*}
    Combining these terms, we find that this is equivalent to
    \begin{align*}
        C_i &= \sum_{\substack{j=0 \\ n-i < m+1}}^{n-i}a_j a_{n-i-j} + a_{m+1}a_{0} + a_{0}a_{m+1}, \\
            &= \sum_{j=0}^{n-i}a_ja_{n-i-j},
    \end{align*}
    when $n-i = m+1$.

     Consider $C_i$ for $n-i > k = m+1$. We know from the inductive hypothesis that the first nested sum is equal to $\sum_{j=0}^{2m-(n-i)}a_{m-j}a_{n-i-m+j}$ for $n-i \geq m$. We can adjust this index by one without any change to the actual value:

     \begin{equation}
         \sum_{j=0}^{2m-(n-i)}a_{m-j}a_{n-i-m+j} = \sum_{j=1}^{2m-(n-i)+1}a_{m+1-j}a_{n-i-(m+1)+j}
     \end{equation}
     
     Given $i < n-k = n-(m+1)$, we know the corresponding $u^j$ term is in both the second and third summation term and have the coefficients $a_{m+1}a_{n-(m+1)-i}$ and $a_{m+1}a_{n-i-(m+1)}$, respectively. The former corresponds to $a_{m+1-j}a_{n-i-(m+1)+j}$ for $j=2(m+1)-(n-i)$, the latter corresponds to $j=0$. We thus fill in our summation and see

     \begin{equation}
        C_i = \sum_{j=0}^{2(m+1)-(n-i)}a_{m+1-j}a_{n-i-(m+1)+j}.
     \end{equation}

    In all cases we see the nested sum takes the following structure.

    \begin{equation}
        C_i \coloneqq \begin{cases}
                \sum_{j = 0}^{n-i}a_{j}a_{n-i-j}, n-i \leq k, \\
                \sum_{j = 0}^{2k-(n-i)}a_{k-j}a_{n-i-k+j}, n-i > k.
            \end{cases}
    \end{equation}

    The base case and inductive step have been demonstrated, completing the proof.    
\end{proof}

\begin{table}[t]
    \centering
    \begin{tabular}{||c|c|c||}
        \hline
        $i$ & $n-i$ & $C_i$ \\
        \hline
        $n$ & $0$ & $a_0a_0$ \\
        $n-1$ & $1$ & $a_1a_0$ + $a_0a_1$ \\
        $n-2$ & $2$ & $a_2a_0 + a_1a_1 + a_0a_2$ \\
        \vdots & \vdots & \vdots \\
        $n-k$ & $k$ & $a_ka_0 + a_{k-1}a_1 + \cdots + a_1a_{k-1} + a_0a_k$ \\
        \vdots & \vdots & \vdots \\
        $n-2k+1$ & $2k-1$ & $a_{k-1}a_k + a_ka_{k-1}$ \\
        $n-2k$ & $2k$ & $a_ka_k$ \\
        \hline
    \end{tabular}
    \caption{Table of coefficients for $i$th time level of the $k$th order Backward Difference method for timestep $n$.}
    \label{tab:coefficients}
\end{table}

An immediate consequence is the following corollary:

\begin{corollary} \label{corollary:BDFk-coefficients-lemma}
    The following identity holds:
    \begin{equation}
        \sum_{i=n-k}^{n-1}a_{n-i}\sum_{j=i-k}^{i}a_{i-j}u^{j} + a_0\sum_{i-n-k}^{n-1}a_{n-i}u^{i} = \sum_{i=n-2k}^{n-1}C_iu^i.
    \end{equation}
\end{corollary}
\begin{proof}
    From Lemma \ref{lemma:BDFk-summation-reordering} we know
    \begin{align}
        \begin{split}
            \sum_{i=n-2k}^{n}C_i &= \sum_{i=n-k}^{n}{a_{n-i}\sum_{j=i-k}^{i}a_{i-j}u^{j}} \\
            &= \sum_{i=n-k}^{n-1}{a_{n-i}\sum_{j=i-k}^{i}a_{i-j}u^{j}} + a_0\sum_{j=i-k}^{n-1}a_{i-j}u^{j} + a_0a_0u^n.
        \end{split}
    \end{align}
    Subtracting $a_0a_0u^n \equiv C_{n}u^n$ from each side demonstrates the identity, concluding the proof.
\end{proof}

Applying a general backward difference method equations to the vector potential, scalar potential, gauge condition and continuity equation yields:
\begin{align}
    &\frac{1}{\Delta c^2}\frac{1}{\Delta t^2}\sum_{i=n-k}^{n}a_{n-i}\sum_{j=i-k}^{i}a_{i-j}\phi^{j} - \Delta \phi^{n} = \frac{\rho^{n}}{\epsilon_0}, \label{eq:BDFk-phi} \\
    &\frac{1}{\Delta c^2}\frac{1}{\Delta t^2}\sum_{i=n-k}^{n}a_{n-i}\sum_{j=i-k}^{i}a_{i-j}\textbf{A}^{j} - \Delta \textbf{A}^{n} = \mu_0\textbf{J}^{n}, \label{eq:BDFk-A} \\
    &\frac{1}{c^2}\frac{1}{\Delta t}\sum_{i=n-k}^{n}a_{n-i}\phi^{i} + \nabla\cdot\textbf{A}^{n} = 0, \label{eq:BDFk-gauge} \\
    &\frac{1}{\Delta t}\sum_{i=n-k}^{n}a_{n-i}\rho^{i} + \nabla\cdot\textbf{J}^{n} = 0 \label{eq:BDFk-continuity}.
\end{align}
where here we assume $\mathbf{J}^{n}$ is given and compute $\rho^{n}$ and the fields for the particle update.  Additionally, we define the residuals
\begin{align}
    &\epsilon_{1}^{n} \coloneqq \frac{1}{c^2}\frac{1}{\Delta t}\sum_{i=n-k}^{n}a_{n-i}\phi^{i} + \nabla\cdot\textbf{A}^{n}, \label{eq:BDFk-gauge-residual} \\
    &\epsilon_{2}^{n} \coloneqq \frac{1}{\Delta t}\sum_{i=n-k}^{n}a_{n-i}\rho^{i} + \nabla\cdot\textbf{J}^{n} \label{eq:BDFk-continuity-residual}.
\end{align}

We now prove an additional lemma which will link the residuals of the semi-discrete continuity equation and semi-discrete gauge condition.

\begin{lemma} \label{lemma:BDFk-residual-equivalence}
    In the system \eqref{eq:BDFk-phi} - \eqref{eq:BDFk-A}, the residual $\epsilon_1^n$ is a linear combination of $\epsilon_2^n$ and $\epsilon_2^i$, where $n-2k \leq i < n$.
\end{lemma}

\begin{proof}
    
    From \eqref{eq:BDFk-phi} and \eqref{eq:BDFk-A}, we get the update schemes
    \begin{align}
        &\phi^{n} = \mathcal{L}^{-1}\left[\frac{1}{\alpha^2}\frac{\rho^n}{\epsilon_0} - \sum_{i=n-k}^{n-1}{a_{n-i}\sum_{j=i-k}^{i}a_{i-j}\phi^{j}} - a_0\sum_{i=n-k}^{n-1}a_{n-i}\phi^{i}\right], \label{eq:BDFk-update-phi} \\
        &\textbf{A}^{n} = \mathcal{L}^{-1}\left[\frac{1}{\alpha^2}\frac{\rho^n}{\epsilon_0} - \sum_{i=n-k}^{n-1}{a_{n-i}\sum_{j=i-k}^{i}a_{i-j}\textbf{A}^{j}} - a_0\sum_{i=n-k}^{n-1}a_{n-i}\textbf{A}^{i}\right]. \label{eq:BDFk-update-A}
    \end{align}
    We plug these into the definition of $\epsilon_1^n$, \eqref{eq:BDFk-gauge-residual}, and see
    \begin{align} \label{eq:BDFk-residual-eps1-1}
        \begin{split}
            \epsilon_1^n &= \frac{1}{c^2}\frac{1}{\Delta t}\sum_{i=n-k}^{n}a_{n-i}\phi^{i} + \nabla\cdot\textbf{A}^{n} \\
            &= \frac{1}{c^2}\frac{1}{\Delta t}\sum_{i=n-k}^{n}a_{n-i}\mathcal{L}^{-1}\left[\frac{1}{\alpha^2}\frac{\rho^{i}}{\epsilon_0} - \sum_{j=i-k}^{i-1}{a_{i-j}\sum_{l=j-k}^{j}{a_{j-l}\phi^l}} - a_0\sum_{j=i-k}^{i-1}a_{i-j}\phi^{j}\right] \\
            &+ \mathcal{L}^{-1}\left[\frac{\mu_0}{\alpha^2}\textbf{J}^{n} - \sum_{i=n-k}^{n-1}{a_{n-i}\sum_{j=i-k}^{i}a_{i-j}\nabla\cdot\textbf{A}^{j}} - a_0\sum_{i=n-k}^{n-1}a_{n-i}\nabla\cdot\textbf{A}^{i}\right]            
        \end{split}
    \end{align}
    This becomes
    \begin{subequations}  \label{eq:BDFk-residual-eps1-2}
        \begin{align}
            &= \mathcal{L}^{-1}\left[\frac{\mu_0}{\alpha^2}\left(\frac{1}{\Delta t}\sum_{i=n-k}^{n}a_{n-i}\rho^i + \nabla\cdot\textbf{J}^{n}\right)\right] \label{eq:BDFk-residual-eps1-2a} \\
            &- \frac{1}{c^2}\frac{1}{\Delta t}\sum_{i=n-k}^{n}a_{n-i}\mathcal{L}^{-1}\left[\sum_{j=i-k}^{i-1}a_{i-j}\sum_{l=j-k}^{j}a_{j-l}\phi^{l} + a_0\sum_{j=i-k}^{i-1}a_{i-j}\phi^{j}\right] \label{eq:BDFk-residual-eps1-2b} \\
            &- \mathcal{L}^{-1}\left[\sum_{i=n-k}^{n-1}a_{n-i}\sum_{j=i-k}^{i}a_{i-j}\nabla\cdot\textbf{A}^{i} + a_0\sum_{i=n-k}^{n-1}a_{n-i}\nabla\cdot\textbf{A}^{i}\right]. \label{eq:BDFk-residual-eps1-2c}
        \end{align}
    \end{subequations}
    The first term is clearly $\mathcal{L}^{-1}\left[\frac{\mu_0}{\alpha^2}\epsilon_2^n\right]$. Disregarding the $\frac{1}{c^2}\frac{1}{\Delta t}$ coefficient and $\mathcal{L}^{-1}$, we apply Corollary \ref{corollary:BDFk-coefficients-lemma} to the argument passed to $\mathcal{L}^{-1}$, and, breaking up the outermost summation, we see:
    \begin{align} \label{eq:BDFk-residual-term-2-rearrange-1}
        \begin{split}
              &a_0\left(\sum_{j=n-k}^{n-1}a_{n-j}\sum_{l=j-k}^{j}a_{j-l}\phi^{l} + a_0\sum_{j=n-k}^{n-1}a_{n-j}\phi^{j}\right) \\
            + &a_1\left(\sum_{j=n-k-1}^{n-2}a_{n-1-j}\sum_{l=j-k}^{j}a_{j-l}\phi^{l} + a_0\sum_{j=n-k-1}^{n-2}a_{n-1-j}\phi^{j}\right) \\
            + &a_2\left(\sum_{j=n-k-2}^{n-3}a_{n-2-j}\sum_{l=j-k}^{j}a_{j-l}\phi^{l} + a_0\sum_{j=n-k-2}^{n-3}a_{n-2-j}\phi^{j}\right) \\
            + &\cdots \\
            + &a_k\left(\sum_{j=n-2k}^{n-k-1}a_{n-k-j}\sum_{l=j-k}^{j}a_{j-l}\phi^{l} + a_0\sum_{j=n-2k}^{n-k-1}a_{n-k-j}\phi^{j}\right).
        \end{split}
    \end{align}
    We apply Corollary \ref{corollary:BDFk-coefficients-lemma} to each nested summation, and we see
    \begin{align}\label{eq:BDFk-residual-term-2-rearrange-2}
        \begin{split}
            &a_0\left(C_{n-2k}\phi^{n-2k} + C_{n-2k+1}\phi^{n-2k+1} + \cdots + C_{n-k}\phi^{n-k} + \cdots + C_{n-1}\phi^{n-1}\right) \\
            + &a_1\left(C_{n-2k}\phi^{n-2k-1} + C_{n-2k+1}\phi^{n-2k} + \cdots + C_{n-k}\phi^{n-k-1} + \cdots + C_{n-1}\phi^{n-2}\right) \\
            + &\cdots \\
            + &a_k\left(C_{n-2k}\phi^{n-3k} + C_{n-2k+1}\phi^{n-3k+1} + \cdots + C_{n-k}\phi^{n-k-2} + \cdots + C_{n-1}\phi^{n-k}\right)
        \end{split}
    \end{align}
    We then group like terms with like, and find that we have a linear combination of first order BDF-k derivatives.
    \begin{align} \label{eq:BDFk-residual-term-2-rearrange-3}
        \begin{split}
            &C_{n-1}\left(a_0\phi^{n-1} + a_1\phi^{n-2} + \cdots + a_k\phi^{n-k-1}\right) \\
            + &C_{n-2}\left(a_0\phi^{n-2} + a_1\phi^{n-3} + \cdots + a_k\phi^{n-k-2}\right) \\
            + &\cdots \\
            + &C_{n-k+1}\left(a_0\phi^{n-k+1} + a_1\phi^{n-k} + \cdots + a_k\phi^{n-2k+1}\right) \\
            + &C_{n-k}\left(a_0\phi^{n-k} + a_1\phi^{n-k-1} + \cdots + a_k\phi^{n-2k}\right) \\
            + &C_{n-k-1}\left(a_0\phi^{n-k-1} + a_1\phi^{n-k-2} + \cdots + a_k\phi^{n-2k-1}\right) \\
            + &\cdots \\
            + &C_{n-2k+1}\left(a_0\phi^{n-2k+1} + a_1\phi^{n-2k} + \cdots + a_k\phi^{n-3k+1}\right) \\
            + &C_{n-2k}\left(a_0\phi^{n-2k} + a_1\phi^{n-2k-1} + \cdots + a_k\phi^{n-3k}\right).
        \end{split}
    \end{align}
    This may be more compactly written as
    \begin{equation}  \label{eq:BDFk-residual-term-2-rearrange-4}
        \sum_{i=n-2k}^{n-1}C_{i}\sum_{j=0}^{k}a_j\phi^{i-j}
    \end{equation}
    A simpler process of just applying Corollary \ref{corollary:BDFk-coefficients-lemma} to the argument in \eqref{eq:BDFk-residual-eps1-2c}, gives similar results:
    \begin{align}\label{eq:BDFk-residual-term-2-rearrange-5}
        \sum_{i=n-2k}^{n-1}C_{i}\nabla\cdot\mathbf{A}^{i}
    \end{align}
    Multiplying the $\frac{1}{c^2}\frac{1}{\Delta t}$ coefficients we left out from \eqref{eq:BDFk-residual-eps1-2b} against \eqref{eq:BDFk-residual-term-2-rearrange-4}, we add this to \eqref{eq:BDFk-residual-term-2-rearrange-5} and see \eqref{eq:BDFk-residual-eps1-2b} + \eqref{eq:BDFk-residual-eps1-2c} becomes
    \begin{equation}
        \mathcal{L}^{-1}\left[\sum_{i=n-2k}^{n-1}C_i\left(\frac{1}{c^2}\frac{1}{\Delta t}\sum_{j=0}^{k}\left(a_j\phi^{i-j}\right) + \nabla\cdot\mathbf{A}^{i}\right)\right].
    \end{equation}
    This contains the exact definitions for the gauge residuals \eqref{eq:BDFk-gauge-residual}, $\epsilon_1^{i}$. Thus we see combining \eqref{eq:BDFk-residual-eps1-2a}, \eqref{eq:BDFk-residual-eps1-2b}, and \eqref{eq:BDFk-residual-eps1-2c} yields
    \begin{align}
        \begin{split}
            &\epsilon_1^n = \mathcal{L}^{-1}\left[\frac{\mu_0}{\alpha^2}\epsilon_2^n + \sum_{i=n-2k}^{n-1}C_i\epsilon_1^i\right], \\
            &C_i \coloneqq \begin{cases}
                \sum_{j = 0}^{i}a_{j}a_{i-j}, i \leq k, \\
                \sum_{j = 0}^{r}a_{k-j}a_{k-r+j}, i > k, r \coloneqq 2k-i.
            \end{cases}
        \end{split}
    \end{align}
    We have demonstrated $\epsilon_1^n$ is a linear combination of $\epsilon_2^n$ and previous steps of the gauge residual.
\end{proof}

Having proven this lemma, we now can prove that any BDF method satisfies the semi-discrete continuity equation satisfies the semi-discrete gauge condition and vice versa.

\begin{theorem}
    Under the system \eqref{eq:BDFk-phi}-\eqref{eq:BDFk-A}, \eqref{eq:BDFk-continuity} is satisfied if \eqref{eq:BDFk-gauge} is, and, assuming the initial conditions of \eqref{eq:BDFk-gauge} are satisfied, \eqref{eq:BDFk-gauge} is satisfied if \eqref{eq:BDFk-continuity} is satisfied.
\end{theorem}

\begin{proof}

    Assume \eqref{eq:BDFk-gauge} is satisfied, that is, $\epsilon_1^i = 0 \, \forall \, i$. From Lemma \ref{lemma:BDFk-residual-equivalence} we have
    \begin{align}
        \begin{split}
            &\epsilon_1^n = \mathcal{L}^{-1}\left[\epsilon_2^{n} + \sum_{i=n-k-1}^{n-1}C_i\epsilon_1^i\right], \\
            \implies &0 = \mathcal{L}^{-1}\left[\epsilon_2^{n} + \sum_{i=n-k-1}^{n-1}C_i\left(0\right)\right], \\
            \implies &0 = \mathcal{L}^{-1}\left[\epsilon_2^{n}\right].
        \end{split}
    \end{align}
    Clearly $\mathcal{L}^{-1}$ is invertible, so $\epsilon_2^n = 0$.

    Assume both \eqref{eq:BDFk-continuity} and the initial conditions of \eqref{eq:BDFk-gauge} are satisfied. That is, $\epsilon_2^i = 0 \forall i$ and $\epsilon_1^n = 0$ $\forall$ $n \leq 0$. We have
    \begin{align}
        \begin{split}
            &\epsilon_1^n = \mathcal{L}^{-1}\left[\epsilon_2^{n} + \sum_{i=n-k-1}^{n-1}C_i\epsilon_1^i\right], \\
            \implies &\epsilon_1^n = \mathcal{L}^{-1}\left[\left(0\right) + \sum_{i=n-k-1}^{n-1}C_i\epsilon_1^i\right], \\
            \implies &\epsilon_1^n = \mathcal{L}^{-1}\left[\sum_{i=n-k-1}^{n-1}C_i\epsilon_1^i\right].
        \end{split}
    \end{align}
    Given the initial conditions are all zero, we know $\epsilon_1^1$ is zero, as is $\epsilon_1^2$, and so on. This logic can be stepped to arbitrary $n$.

    We have demonstrated \eqref{eq:BDFk-gauge} $\implies$ \eqref{eq:BDFk-continuity} if \eqref{eq:BDFk-gauge} is satisfied initially, and \eqref{eq:BDFk-continuity} $\implies$ \eqref{eq:BDFk-gauge}.    
\end{proof}

We have linked the semi-discrete continuity equation with the semi-discrete gauge condition for arbitrary order BDF methods. Now we will show Gauss's law for electricity follows from the satisfaction of the semi-discrete gauge condition \eqref{eq:BDFk-gauge}.

\begin{theorem}
    If the semi-discrete gauge condition \eqref{eq:BDFk-gauge} is satisfied for time level $n+1$, then $\nabla\cdot\textbf{E}^{n+1} = \frac{\rho^{n+1}}{\epsilon_0}$.
\end{theorem}

\begin{proof}
    We assume \eqref{eq:BDFk-gauge}:
    \begin{equation}
        \frac{1}{c^2}\frac{1}{\Delta t}\sum_{i=n-k}^{n}a_{n-i}\phi^{i} + \nabla\cdot\textbf{A}^{n} = 0.
    \end{equation}
    By definition we have $\mathbf{E} = -\nabla\phi - \frac{\partial\mathbf{A}}{\partial t}$, whose BDF-k time discrete form is as follows:
    \begin{equation}
        \textbf{E}^{n} = -\nabla\phi^n - \frac{1}{\Delta t}\sum_{i=n-k}^{n}a_{n-i}\textbf{A}^{i}.
    \end{equation}
    We take the divergence, which yields
    \begin{equation}
        \nabla\cdot\textbf{E}^{n} = -\Delta \phi^n - \frac{1}{\Delta t}\sum_{i=n-k}^{n}a_{n-i}\nabla\cdot\textbf{A}^{i}
    \end{equation}
    Using \eqref{eq:BDFk-gauge}, we replace $\nabla\cdot\textbf{A}^{i}$ with the BDF-k of $\phi^{i}$:
    \begin{align}
        \begin{split}
            \nabla\cdot\textbf{E}^{n} &= -\Delta \phi^n - \frac{1}{\Delta t}\sum_{i=n-k}^{n}a_{n-i}\left(-\frac{1}{c^2}\frac{1}{\Delta t}\sum_{j=i-k}^{i}a_{i-j}\phi^{j}\right) \\
            &= \frac{1}{c^2}\frac{1}{\Delta t^2}\sum_{i=n-k}^{n}a_{n-i}\sum_{j=i-k}^{i}a_{i-j}\phi^{j} - \Delta \phi^n \\
            &= \frac{\rho^n}{\epsilon_0}.
        \end{split}
    \end{align}
    The last step was taken by definition of the scalar potential wave equation \eqref{eq:BDFk-phi}. 
    
    We have demonstrated Gauss's Law for electricity is satisfied, completing the proof.
\end{proof}

\begin{remark}
    We note that all BDF methods satisfy this theorem.  However, while first and second order BDF methods are A-stable, third order BDF methods and above are A$\alpha$-stable.  In practice, this means that using first and second order BDF methods can be used with any CFL one wants, while third order BDF and above need to take time steps that are greater than $t_{min}$ to be stable when solving a wave equation, which has eigenvalues on the imaginary axis.  This is because third order BDF and above have stability diagrams that slightly move into the left half of the plane of the region of absolute stability and then back out again. A large enough CFL will avoid this issue.    
\end{remark}

\begin{remark}
For our numerical experiments we use the second order BDF method, BDF-2. The semi-discrete wave equations, Lorenz gauge condition, and continuity equation take the following form:
\begin{align*}
    \frac{1}{c^2}\frac{\phi^{n+1} - \frac{8}{3}\phi^{n} + \frac{22}{9}\phi^{n-1} - \frac{8}{9}\phi^{n-2} + \frac{1}{9}\phi^{n-3}}{\left((2/3)\Delta t\right)^2} - \Delta \phi^{n+1} &= \frac{\rho^{n+1}}{\epsilon_0}, \\
    \frac{1}{c^2}\frac{\textbf{A}^{n+1} - \frac{8}{3}\textbf{A}^{n} + \frac{22}{9}\textbf{A}^{n-1} - \frac{8}{9}\textbf{A}^{n-2} + \frac{1}{9}\textbf{A}^{n-3}}{\left((2/3)\Delta t\right)^2} - \Delta \textbf{A} &= \mu_0\textbf{J}^{n+1}, \\
    \frac{1}{c^2}\frac{\phi^{n+1}-\frac{4}{3}\phi^{n} + \frac{1}{3}\phi^{n-1}}{\left(2/3\right)\Delta t} + \nabla\cdot\textbf{A}^{n+1} &= 0, \\
    \frac{\rho^{n+1} - \frac{4}{3}\rho^{n} + \frac{1}{3}\rho^{n-1}}{\left(2/3\right)\Delta t} + \nabla\cdot\textbf{J}^{n+1} &= 0.
\end{align*}
\end{remark}

%
%
\subsection{The Generalized s-stage Diagonal Implicit Runge-Kutta Method}
\label{subsec:DIRKs-formulation}

In this section we will first establish an $s$-stage method for solving the wave equations and continuity equation. We will then prove the properties this method, namely that conservation of charge implies satisfaction of the Lorenz gauge condition.  As with the CDF and BDF methods, we will assume we are given $\mathbf{J}$ during the construction.  In practice, $\mathbf{J}$ will come from a particle advance, and there are several ways to go about constructing $\mathbf{J}$ at stage values.  
\subsubsection{The s-stage DIRK Method}
Consider $\frac{1}{c^2}\frac{\partial^2u}{\partial t}-\Delta u = S$, which we rewrite as a first order system: 
\begin{align}
    \frac{\partial}{\partial t}
    \begin{pmatrix}
        v \\ u
    \end{pmatrix}
    =
    \begin{pmatrix}
        c^2\left(S + \Delta u\right) \\
        v
    \end{pmatrix}.
\end{align}
We apply an $s$-stage Runge-Kutta method, 
%
%
where defining $S^{(i)} \coloneqq S\left(x,t^n + c_ih\right)$, we arrive at
\begin{align}
    k_i^v &= c^2\left(S^{(i)} + \Delta\left(u^n + h\sum_{j=1}^ia_{ij}k_j^u\right)\right), \label{eq:DIRKs-kv-def} \\
    k_i^u &= v^n + h\sum_{j=1}^{i}a_{ij}k_j^v, \label{eq:DIRKs-ku-def} \\
    u^{n+1} &= u^n + h\sum_{i=1}^{s}b_ik_i^u, \\
    v^{n+1} &= v^n + h\sum_{i=1}^{s}b_ik_i^v.
\end{align}
Here the $a_{ij}$, $b_i$, and $c_i$ terms come from whatever Butcher tableau with which our Runge-Kutta method is working.  When solving for $k_i^v$ we use the definition of $k_i^u$ and vice versa to arrive at,
%
%
\begin{align}
    \mathcal{L}_i\left[k_i^u\right] &= v^n + h\sum_{j=1}^{i-1}a_{ij}k_j^v + ha_{ii}\left(c^2\left(S^{(i)} + \Delta\left(u^n + h\sum_{j=1}^{i-1}a_{ij}k_j^u\right)\right)\right), \\
    \mathcal{L}_i\left[k_i^v\right] &= c^2\left(S^{(i)} + \Delta\left(u^n + h\sum_{j=1}^{i-1}a_{ij}k_j^u\right)\right) + ha_{ii}c^2\Delta \left(v^n + h\sum_{j=1}^{i-1}a_{ij}k_j^v\right).
\end{align}
Here we have defined $\mathcal{L}_{i} \coloneqq \left(\mathcal{I} - \frac{1}{\alpha_i^2}\Delta\right)$ and $\alpha_i \coloneqq \frac{1}{ha_{ii}c}$. Note the $k_j$ values are known for $j < i$.


We apply this to our system of wave equations, the Lorenz gauge, and the continuity equation. To ease the comparison, we define $v_\phi \coloneqq \frac{\partial \phi}{\partial t}$ and $v_\mathbf{A} \coloneqq \frac{\partial \mathbf{A}}{\partial t}$. Doing so, we see the update statements for $\phi$ and $\mathbf{A}$ are

\begin{align}
    \phi^{n+1} &= \phi^n + h\sum_{i=1}^{s}b_ik_i^\phi, \label{eq:DIRKs-phi-update} \\
    \mathbf{A}^{n+1} &= \mathbf{A}^n + h\sum_{i=1}^{s}b_ik_i^\mathbf{A}, \label{eq:DIRKs-A-update} \\
    v_\phi^{n+1} &= v_\phi^n + h\sum_{i=1}^{s}b_ik_i^{v_\phi}, \label{eq:DIRKs-v_phi-update} \\
    \mathbf{v}_\mathbf{A}^{n+1} &= \mathbf{v}_\mathbf{A}^n + h\sum_{i=1}^{s}b_ik_i^{\mathbf{v}_\mathbf{A}}, \label{eq:DIRKs-v_A-update}
\end{align}
with the $i=1,\cdots,s$ stage values computed according to
\begin{align}
    k_i^{\phi} &= \mathcal{L}^{-1}_i\left[v_\phi^n + h\sum_{j=1}^{i-1}a_{ij}k_j^{v_\phi} + ha_{ii}c^2\frac{\rho^{(i)}}{\epsilon_0} + ha_{ii}c^2\Delta\left(\phi^n + h\sum_{j=1}^{i-1}a_{ij}k_j^\phi\right)\right], \label{eq:DIRKs-k-phi-update} \\
    k_i^{v_\phi} &= \mathcal{L}^{-1}_i\left[\frac{\rho^{(i)}}{\epsilon_0} + \Delta\left(\phi^n + h\sum_{j=1}^{i-1}a_{ij}k_j^\phi\right) + ha_{ii}c^2\Delta\left(v_\phi^n + h\sum_{j=1}^{i-1}a_{ij}k_j^{v_\phi}\right)\right], \label{eq:DIRKs-k-v_phi-update} \\
    \mathbf{k}_i^{\mathbf{A}} &= \mathcal{L}^{-1}_i\left[\mathbf{v}_\mathbf{A}^n + h\sum_{j=1}^{i-1}a_{ij}\mathbf{k}_j^{\mathbf{v}_\mathbf{A}} + ha_{ii}c^2\mu_0\mathbf{J}^{(i)} + ha_{ii}c^2\Delta\left(\mathbf{A}^n + h\sum_{j=1}^{i-1}a_{ij}\mathbf{k}_j^\mathbf{A}\right)\right], \label{eq:DIRKs-k-A-update} \\
    \mathbf{k}_i^{v_\mathbf{A}} &= \mathcal{L}^{-1}_i\left[\mu_0\mathbf{J}^{(i)} + \Delta\left(\mathbf{A}^n + h\sum_{j=1}^{i-1}a_{ij}\mathbf{k}_j^\mathbf{A}\right) + ha_{ii}c^2\Delta\left(\mathbf{v}_\mathbf{A}^n + h\sum_{j=1}^{i-1}a_{ij}\mathbf{k}_j^{\mathbf{v}_\mathbf{A}}\right)\right]. \label{eq:DIRKs-k-v_A-update}
\end{align}
We do a much more straightforward update statement for the gauge and continuity update equations:
\begin{align}
    \phi^{n+1} &= \phi^{n} - c^2h\sum_{i=1}^{s}b_i\nabla\cdot\mathbf{A}^{(i)}, \label{eq:lorenz-gauge-res-DIRKs} \\
    \rho^{n+1} &= \rho^{n} - h\sum_{i=1}^{s}b_i\nabla\cdot\mathbf{J}^{(i)}. \label{eq:lorenz-continuity-res-DIRKs}
\end{align}

In the above updates we have defined $\left\{\mathbf{A},\mathbf{J},\rho\right\}^{(i)}$ as the $i$-th linear interpolation between time levels $n$ and $n+1$ given by the Runge-Kutta method, i.e.,

\begin{equation}
    X^{(i)} = \left(1-c_i\right)X^{n} + c_iX^{n+1}.
\end{equation}
So it follows
\begin{align}
    \begin{split}
        \sum_{i=1}^{r}b_iX^{(i)} &= \sum_{i=1}^{r}b_i\left(\left(1-c_i\right)X^{n} + c_iX^{n+1}\right) \\
        &\equiv R^{r,n}X^{n} + R^{r,n+1}X^{n+1}.
    \end{split}
\end{align}
\subsubsection{The Properties of the s-stage Method}

We wish to link the satisfaction of the DIRK-$s$ formulation of the Lorenz gauge condition \eqref{eq:lorenz-gauge-res-DIRKs} with the DIRK-$s$ continuity equation \eqref{eq:lorenz-continuity-res-DIRKs}.  To do so, we define the following five residuals:

\begin{align}
    \epsilon_1^{n+1} &\coloneqq -\phi^{n+1} + \phi^{n} - c^2h\left(R^{s,n}\nabla\cdot\mathbf{A}^{n} + R^{s,n+1}\nabla\cdot\mathbf{A}^{n+1}\right) \label{eq:DIRKs-gauge-res} \\
    \epsilon_2^{n+1} &\coloneqq -v_\phi^{n+1} + v_\phi^{n} - c^2h\left(R^{s,n}\nabla\cdot\mathbf{v}_\mathbf{A}^{n} + R^{s,n+1}\nabla\cdot\mathbf{v}_\mathbf{A}^{n+1}\right) \label{eq:DIRKs-ddt-gauge-res} \\
    \epsilon_3^{n+1} &\coloneqq -\rho^{n+1} + \rho^{n} - h\left(R^{s,n}\nabla\cdot\mathbf{J}^{n} + R^{s,n+1}\nabla\cdot\mathbf{J}^{n+1}\right) \label{eq:DIRKs-continuity-res} \\
    \epsilon_4^{n+1} &\coloneqq -k_i^{\phi,n+1} + k_i^{\phi,n} - c^2h\left(R^{s,n}\nabla\cdot\mathbf{k}_i^{\mathbf{A},n} + R^{s,n}\nabla\cdot\mathbf{k}_i^{\mathbf{A},n+1}\right) \label{eq:DIRKs-k-res} \\
    \epsilon_5^{n+1} &\coloneqq -k_i^{v_\phi,n+1} + k_i^{v_\phi,n} - c^2h\left(R^{s,n}\nabla\cdot\mathbf{k}_i^{\mathbf{v}_\mathbf{A},n} + R^{s,n+1}\nabla\cdot\mathbf{k}_i^{\mathbf{v}_\mathbf{A},n+1}\right) \label{eq:DIRKs-ddt-k-res}
\end{align}

We have a host of lemmas to prove that will assist with the main theorem. The first, Lemma \ref{lemma:continuity-res-identity-DIRKs}, is partially related to the residual of the continuity equation \eqref{eq:DIRKs-continuity-res} and will assist in proving Lemmas \ref{lemma:k-linear-combination-DIRKs} and \ref{lemma:ddt-k-linear-combination-DIRKs}, which are concerned with the residuals of the $k$ variables \eqref{eq:DIRKs-k-res} and \eqref{eq:DIRKs-ddt-k-res}. These lemmas assist in proving Lemmas \ref{lemma:gauge-linear-combination-DIRKs} and \ref{lemma:ddt-gauge-linear-combination-DIRKs}, two additional lemmas concerned with the residuals of the gauge condition \eqref{eq:DIRKs-gauge-res} and its time derivative \eqref{eq:DIRKs-ddt-gauge-res}, respectively, and these will finally assist in proving Theorem \ref{thm:DIRKs-continuity-implies-gauge}, that satisfaction of the semi-discrete continuity equation implies satisfaction of the semi-discrete gauge condition.

\begin{lemma} \label{lemma:continuity-res-identity-DIRKs}
    For any substep $i$, the following identity holds:
    \begin{equation}
        -\frac{\rho^{(i),n}}{\epsilon_0} + \frac{\rho^{(i),n-1}}{\epsilon_0} - c^2h\mu_0\left(R^{s,n}\nabla\cdot\mathbf{J}^{(i),n-1} + R^{s,n+1}\nabla\cdot\mathbf{J}^{(i),n}\right) = \mu_0c^2\left(\left(1-c_i\right)\epsilon_3^{n} + c_i\epsilon_3^{n+1}\right).
    \end{equation}
\end{lemma}
\begin{proof}
    We use the identity $c^2 = \frac{1}{\mu_0\epsilon_0}$ and see
    \begin{align}
        \begin{split}
            &-\frac{\rho^{(i),n}}{\epsilon_0} + \frac{\rho^{(i),n-1}}{\epsilon_0} - c^2h\mu_0\left(R^{s,n}\nabla\cdot\mathbf{J}^{(i),n-1} + R^{s,n+1}\nabla\cdot\mathbf{J}^{(i),n}\right) \\
            &= \mu_0\left(-\frac{\rho^{(i),n}}{\mu_0\epsilon_0} + \frac{\rho^{(i),n-1}}{\mu_0\epsilon_0} - c^2h\mu_0\left(R^{s,n}\nabla\cdot\mathbf{J}^{(i),n-1} + R^{s,n+1}\nabla\cdot\mathbf{J}^{(i),n}\right)\right) \\
            &= \mu_0c^2\left(-\rho^{(i),n} + \rho^{(i),n-1} - h\left(R^{s,n}\nabla\cdot\mathbf{J}^{(i),n-1} + R^{s,n+1}\nabla\cdot\mathbf{J}^{(i),n}\right)\right).
        \end{split}
    \end{align}
    Taking the interior portion, we further derive
    \begin{align}
        \begin{split}
            - &\left(\left(1-c_i\right)\rho^{n} + c_i\rho^{n+1}\right) + \left(\left(1-c_i\right)\rho^{n-1} + c_i\rho^{n}\right) - h\left(R^{s,n}\nabla\cdot\left(\left(1-c_i\right)\mathbf{J}^{n-1} + c_i\mathbf{J}^{n}\right) + R^{s,n+1}\nabla\cdot\left(\left(1-c_i\right)\mathbf{J}^{n} + c_i\mathbf{J}^{n+1}\right)\right) \\
            = &\left(1-c_i\right)\left(-\rho^{n} + \rho^{n-1}\right) + c_i\left(-\rho^{n+1} + \rho^{n}\right) - h\left(\left(1-c_i\right)\left(R^{s,n}\nabla\cdot \mathbf{J}^{n-1} + R^{s,n+1}\nabla\cdot \mathbf{J}^{n}\right) + c_i\left(R^{s,n}\nabla\cdot \mathbf{J}^{n} + R^{s,n+1}\nabla\cdot \mathbf{J}^{n+1}\right)\right) \\
            = &\left(1-c_i\right)\left(-\rho^{n} + \rho^{n-1} - h\left(R^{s,n}\nabla\cdot \mathbf{J}^{n-1} + R^{s,n+1}\nabla\cdot \mathbf{J}^{n}\right)\right) + c_i\left(-\rho^{n+1} + \rho^{n} - h\left(R^{s,n}\nabla\cdot \mathbf{J}^{n} + R^{s,n+1}\nabla\cdot \mathbf{J}^{n+1}\right)\right) \\
            = &\left(1-c_i\right)\epsilon_3^{n} + c_i\epsilon_3^{n+1}.
        \end{split}
    \end{align}
    We have thus demonstrated the identity.
\end{proof}

\begin{lemma} \label{lemma:k-linear-combination-DIRKs}
    The $\epsilon_4$ residual \eqref{eq:DIRKs-k-res} at time level $n$ is a linear combination of the residuals of the Lorenz gauge \eqref{eq:DIRKs-gauge-res} at time level $n$, the time derivative of the Lorenz gauge \eqref{eq:DIRKs-ddt-gauge-res} at time level $n$, and the continuity equation at time levels $n$ and $n+1$.
\end{lemma}
\begin{proof}
    We prove by induction. For the base case, we see
    \begin{align}
        \begin{split}
            \epsilon_4^n &= - k_1^{\phi,n} + k_1^{\phi,n-1} - hc^2\left(R^{n}\nabla\cdot \mathbf{k}_1^{\mathbf{A},n-1} + R^{n+1}\nabla\cdot \mathbf{k}_1^{\mathbf{A},n}\right) \\
            &= -\mathcal{L}_1^{-1}\left[v_\phi^n + ha_{11}c^2\left(\Delta \phi^n + \frac{\rho^{(1),n}}{\epsilon_0}\right)\right] \\
            &+ \mathcal{L}_1^{-1}\left[v_\phi^{n-1} + ha_{11}c^2\left(\Delta \phi^{n-1} + \frac{\rho^{(1),n-1}}{\epsilon_0}\right)\right] \\
            &- hc^2\left(R^{n}\nabla\cdot\mathcal{L}_1^{-1}\left[\mathbf{v}_\mathbf{A}^{n-1} + ha_{11}c^2\left(\Delta \mathbf{A}^{n-1} + \mu_0 \mathbf{J}^{(1),n-1}\right)\right] + R^{n+1}\nabla\cdot\mathcal{L}_1^{-1}\left[\mathbf{v}_\mathbf{A}^{n} + ha_{11}c^2\left(\Delta \mathbf{A}^n + \mu_0 \mathbf{J}^{(1),n}\right)\right]\right) \\
            &= \mathcal{L}_1^{-1}\left[-v_\phi^{n} + v_\phi^{n-1} - hc^2\left(R^n \nabla\cdot \mathbf{v}_\mathbf{A}^{n-1} + R^{n+1} \nabla\cdot \mathbf{v}_\mathbf{A}^n\right)\right] \\
            &+ ha_{11}c^2\mathcal{L}_1^{-1}\left[-\Delta \phi^n + \Delta \phi^{n-1} - hc^2\left(R^{n}\Delta\left(\nabla\cdot \mathbf{A}^{n-1}\right) + R^{n+1}\Delta\left(\nabla\cdot \mathbf{A}^{n}\right)\right)\right] \\
            &+ ha_{11}c^2\mathcal{L}_1^{-1}\left[-\frac{\rho^{(1),n-1}}{\epsilon_0} + \frac{\rho^{(1),n}}{\epsilon_0} - hc^2\mu_0\left(R^{n}\nabla\cdot \mathbf{J}^{(1),n-1} + R^{n+1}\nabla\cdot \mathbf{J}^{(1),n}\right)\right] \\
            &= \mathcal{L}_1^{-1}\left[\epsilon_2^n + ha_{11}c^2\Delta\epsilon_1^n + ha_{11}c^4\mu_0\left(\left(1-c_1\right)\epsilon_3^{n} + c_1\epsilon_3^{n+1}\right)\right]
        \end{split}
    \end{align}
    The last step is justified definitionally and by Lemma \ref{lemma:continuity-res-identity-DIRKs}. We have shown the base case, now we wish to show the inductive step. Assuming true for $i=l-1$, we now consider $i=l$:
    \begin{align}
        \begin{split}
            \epsilon_4^n &= -k_l^{\phi,n} + k_l^{\phi,n-1} - hc^2\left(R^{s,n}\nabla\cdot\mathbf{k}_l^{\mathbf{A},n-1} + R^{s,n+1}\nabla\cdot\mathbf{k}_l^{\mathbf{A},n}\right) \\
            = &-\mathcal{L}_l^{-1}\left[v_\phi^{n-1} + h\sum_{j=1}^{l-1}a_{lj}k_j^{v_\phi,n-1} + ha_{ll}c^2\frac{\rho^{(l),n-1}}{\epsilon_0} + ha_{ll}c^2\Delta\left(\phi^{n-1} + h\sum_{j=1}^{l-1}a_{lj}k_j^{\phi,n-1}\right)\right] \\
            + &\mathcal{L}_l^{-1}\left[v_\phi^{n-2} + h\sum_{j=1}^{l-1}a_{lj}k_j^{v_\phi,n-2} + ha_{ll}c^2\frac{\rho^{(l),n-2}}{\epsilon_0} + ha_{ll}c^2\Delta\left(\phi^{n-2} + h\sum_{j=1}^{l-1}a_{lj}k_j^{\phi,n-2}\right)\right] \\
            - &hc^2R^{s,n}\nabla\cdot\left(\mathcal{L}_l^{-1}\left[\mathbf{v}_\mathbf{A}^{n-2} + h\sum_{j=1}^{l-1}a_{lj}\mathbf{k}_j^{\mathbf{v}_\mathbf{A},n-2} + ha_{ll}c^2\mu_0\mathbf{J}^{(l),n-2} + ha_{ll}c^2\Delta\left(\mathbf{A}^{n-2} + h\sum_{j=1}^{l-1}a_{lj}\mathbf{k}_j^{\mathbf{A},n-2}\right)\right]\right) \\
            - &hc^2R^{s,n+1}\nabla\cdot\left(\mathcal{L}_l^{-1}\left[\mathbf{v}_\mathbf{A}^{n-1} + h\sum_{j=1}^{l-1}a_{lj}\mathbf{k}_j^{\mathbf{v}_\mathbf{A},n-1} + ha_{ll}c^2\mu_0\mathbf{J}^{(l),n-1} + ha_{ll}c^2\Delta\left(\mathbf{A}^{n-1} + h\sum_{j=1}^{l-1}a_{lj}\mathbf{k}_j^{\mathbf{A},n-1}\right)\right]\right).
        \end{split}
    \end{align}
    Grouping all of these inside $\mathcal{L}^{-1}_l$, we can break all these apart, rearrange, and either using definitions or Lemma \ref{lemma:continuity-res-identity-DIRKs} to see the argument passed to this operator takes the form:
    \begin{align}
        \begin{split}
            &-v_\phi^{n-1} + v_\phi^{n-2} - hc^2\left(hc^2R^{s,n}\mathbf{v}_\mathbf{A}^{n-1} + R^{s,n+1}\mathbf{v}_\mathbf{A}^{n-1}\right) \\
            &+ha_{ll}c^2\left(-\frac{\rho^{(l),n-1}}{\epsilon_0} + \frac{\rho^{(l),n-2}}{\epsilon_0} - hc^2\mu_0\left(R^{s,n}\nabla\cdot\mathbf{J}^{(l),n-2} + R^{s,n+1}\nabla\cdot\mathbf{J}^{(l),n-1}\right)\right) \\
            &+h^2a_{ll}c^2\Delta\left(-\phi^{n-1} + \phi^{n-2} -hc^2\left(R^{s,n}\nabla\cdot\mathbf{A}^{n-2} + R^{s,n+1}\nabla\cdot\mathbf{A}^{n-1}\right)\right) \\
            &+h^2a_{ll}c^2\Delta\sum_{j=1}^{l-1}{a_{lj}\left(-k_j^{\phi,n-1} + k_j^{\phi,n-2} - hc^2\left(R^{s,n}\nabla\cdot\mathbf{k}^{\mathbf{A},n-2} + R^{s,n+1}\nabla\cdot\mathbf{k}_j^{\mathbf{A},n-1}\right)\right)} \\
            &=\epsilon_2^{n-1} + ha_{ll}c^2\mu_0c^2\left(\left(1-c_l\right)\epsilon_3^{n-1} + c_l\epsilon_3^{n}\right) + ha_{ll}c^2\Delta\epsilon_1^{n-1} + h^2a_{11}c^2\Delta\sum_{j=1}^{l-1}a_{lj}\alpha_l^{n-1}.
        \end{split}
    \end{align}
    We define $\alpha_l^{n-1}$ to be some residual that is a linear combination of $\epsilon_1^{n-1}$, $\epsilon_2^{n-1}$, and $\epsilon_3^{n-1}$. This is justified by our inductive hypothesis. We have demonstrated $\epsilon_4^n$ is a linear combination of residuals of previous timesteps for any number of substeps.
\end{proof}

\begin{lemma} \label{lemma:ddt-k-linear-combination-DIRKs}
    The $\epsilon_5$ residual \eqref{eq:DIRKs-ddt-k-res} at time level $n$ is a linear combination of the residuals of the Lorenz gauge \eqref{eq:DIRKs-gauge-res} at time level $n$, the time derivative of the Lorenz gauge \eqref{eq:DIRKs-ddt-gauge-res} at time level $n$, the continuity equation at time levels $n$ and $n+1$, and the $\epsilon_4$ residual (itself a linear combination of the other three residuals) at time level $n$.
\end{lemma}
\begin{proof}
    We prove by induction. For the base case, we see
    \begin{align}
        \begin{split}
            \epsilon_5^n &= - k_1^{v_\phi,n} + k_1^{v_\phi,n-1} - hc^2\left(R^{n}\nabla\cdot \mathbf{k}_1^{\mathbf{v}_\mathbf{A},n-1} + R^{n+1}\nabla\cdot \mathbf{k}_1^{\mathbf{v}_\mathbf{A},n}\right) \\
            &= -\mathcal{L}_1^{-1}\left[c^2\left(\Delta\phi^n + \frac{\rho^{(1),n}}{\epsilon_0} + ha_{11}\Delta v_\phi^{n}\right)\right] + \mathcal{L}_1^{-1}\left[c^2\left(\Delta\phi^{n-1} + \frac{\rho^{(1),n-1}}{\epsilon_0} + ha_{11}\Delta v_\phi^{n-1}\right)\right] \\
            &- hc^2\left(R^{n}\nabla\cdot\mathcal{L}_1^{-1}\left[c^2\left(\Delta\mathbf{A}^{n-1} + \mu_0\mathbf{J}^{(1),n-1} + ha_{11}\Delta\mathbf{v}_\mathbf{A}^{n-1}\right)\right] + R^{n+1}\nabla\cdot\mathcal{L}_1^{-1}\left[c^2\left(\Delta\mathbf{A}^{n} + \mu_0\mathbf{J}^{(1),n} + ha_{11}\Delta\mathbf{v}_\mathbf{A}^{n}\right)\right]\right) \\
            &= c^2\mathcal{L}_1^{-1}\left[\Delta\left(-\phi^{n} + \phi^{n-1} -hc^2\left(R^{n}\nabla\cdot\mathbf{A}^{n-1} + R^{n+1}\nabla\cdot\mathbf{A}^{n}\right)\right)\right] \\
            &+ c^2\mathcal{L}_1^{-1}\left[-\frac{\rho^{(1),n}}{\epsilon_0} + \frac{\rho^{(1),n-1}}{\epsilon_0} - hc^2\mu_0\left(R^{n}\nabla\cdot\mathbf{J}^{(1),n-1} + R^{n+1}\nabla\cdot\mathbf{J}^{(1),n}\right)\right] \\
            &+ ha_{11}c^2\mathcal{L}_1^{-1}\left[\Delta\left(-v_\phi^{n} + v_\phi^{n-1} - hc^2\left(R^n\nabla\cdot\mathbf{A}_\mathbf{A}^{n} + R^{n+1}\nabla\cdot\mathbf{v}_\mathbf{A}^{n-1}\right)\right)\right] \\
            &= c^2\mathcal{L}_1^{-1}\left[\Delta \epsilon_1^{n}\right] + c^4\mu_0\mathcal{L}_1^{-1}\left[\left(1-c_1\right)\epsilon_3^{n} + c_1\epsilon_3^{n+1}\right] + ha_{11}c^2\mathcal{L}_1^{-1}\left[\epsilon_2^{n}\right]
        \end{split}
    \end{align}
    The last step is justified definitionally and by Lemma \ref{lemma:continuity-res-identity-DIRKs}. We have proven the base case, now we assume true for $i = l-1$ substeps. We consider $i=l$ and see
    \begin{align}
        \begin{split}
            &-k_l^{v_\phi,n} + k_l^{v_\phi,n-1} - hc^2\left(R^{s,n}\nabla\cdot\mathbf{k}_l^{\mathbf{v}_\mathbf{A},n-1} + R^{s,n+1}\nabla\cdot\mathbf{k}_l^{\mathbf{v}_\mathbf{A},n}\right) \\
            &= -\mathcal{L}_l^{-1}\left[\frac{\rho^{(l),n-1}}{\epsilon_0} + \Delta\left(\phi^{n-1} + h\sum_{j=1}^{l-1}a_{lj}k_j^{\phi,n-1}\right) + ha_{ll}c^2\Delta\left(v_\phi^{n-1} + h\sum_{j=1}^{l-1}a_{lj}k_j^{v_\phi,n-1}\right)\right] \\
            &+ \mathcal{L}_l^{-1}\left[\frac{\rho^{(l),n-2}}{\epsilon_0} + \Delta\left(\phi^{n-2} + h\sum_{j=1}^{l-1}a_{lj}k_j^{\phi,n-2}\right) + ha_{ll}c^2\Delta\left(v_\phi^{n-2} + h\sum_{j=1}^{l-1}a_{lj}k_j^{v_\phi,n-2}\right)\right] \\
            &- hc^2R^{s,n}\nabla\cdot\mathcal{L}_l^{-1}\left[\mu_0\mathbf{J}^{(l),n-2} + \Delta\left(\mathbf{A}^{n-2} + h\sum_{j=1}^{l-1}a_{lj}\mathbf{k}_j^{\mathbf{A},n-2}\right) + ha_{ll}c^2\Delta\left(\mathbf{v}_\mathbf{A}^{n-2} + h\sum_{j=1}^{l-1}a_{lj}\mathbf{k}_j^{\mathbf{v}_\mathbf{A},n-1}\right)\right] \\
            &- hc^2R^{s,n+1}\nabla\cdot\mathcal{L}_l^{-1}\left[\mu_0\mathbf{J}^{(l),n-1} + \Delta\left(\mathbf{A}^{n-1} + h\sum_{j=1}^{l-1}a_{lj}\mathbf{k}_j^{\mathbf{A},n-1}\right) + ha_{ll}c^2\Delta\left(\mathbf{v}_\mathbf{A}^{n-1} + h\sum_{j=1}^{l-1}a_{lj}\mathbf{k}_j^{\mathbf{v}_\mathbf{A},n-1}\right)\right]
        \end{split}
    \end{align}
    We can group all of this within the $\mathcal{L}_{l}^{-1}$ operator and, taking the entire argument, we rearrange and see
    \begin{align}
        \begin{split}
            &-\frac{\rho^{(l),n-1}}{\epsilon_0} + \frac{\rho^{(l)},n-1}{\epsilon_0} - hc^2\mu_0\left(R^{s,n}\nabla\cdot\mathbf{J}^{(l),n-2} + R^{s,n+1}\nabla\cdot\mathbf{J}^{n-1}\right) \\
            &+\Delta\left(-\phi^{n-1} + \phi^{n-2} - hc^2\left(R^{s,n}\nabla\cdot\mathbf{A}^{n-2} + R^{s,n+1}\nabla\cdot\mathbf{A}^{n-1}\right)\right) \\
            &+h\Delta\sum_{j=1}^{l-1}{a_{lj}\left(-k_j^{\phi,n-1} + k_j^{\phi,n-2} - hc^2\left(R^{s,n}\nabla\cdot\mathbf{k}_j^{\mathbf{A},n-2} + R^{s,n+1}\nabla\cdot\mathbf{k}_j^{\mathbf{A},n-1}\right)\right)} \\
            &+ha_{ll}c^2\Delta\left(-v_\phi^{n-1} + v_\phi^{n-2} - hc^2\left(R^{s,n}\nabla\cdot\mathbf{v}_\mathbf{A}^{n-1} + R^{s,n+1}\nabla\cdot\mathbf{v}_\mathbf{A}^{n-1}\right)\right) \\
            &+h^2a_{ll}c^2\Delta\sum_{j=1}^{l-1}{a_{lj}\left(-k_j^{v_\phi,n-1} + k_j^{v_\phi,n-2} - hc^2\left(R^{s,n}\nabla\cdot\mathbf{k}_j^{\mathbf{v}_\mathbf{A},n-2} + R^{s,n+1}\nabla\cdot\mathbf{k}_j^{\mathbf{v}_\mathbf{A},n-1}\right)\right)} \\
            &=\mu_0c^2\left(\left(1-c_l\right)\epsilon_3^{n-1} + c_l\epsilon_3^n\right) + \Delta\epsilon_1 + h\Delta\sum_{j=1}^{l-1}\alpha_l^{n-1} + ha_{ll}c^2\Delta\epsilon_2^{n-1} + h^2a_{ll}c^2\Delta\sum_{j=1}^{l-1}\beta_l^{n-1}.
        \end{split}
    \end{align}
    
    Here we have set $\beta_l$ as some linear combination of $\epsilon_1^{n-1}$, $\epsilon_2^{n-1}$, and $\epsilon_3^{n-1}$, justified by our inductive hypothesis. $\alpha_l^{n-1}$ is likewise a linear combination of these residuals and is justified by Lemma \ref{lemma:k-linear-combination-DIRKs}. We have demonstrated $\epsilon_5^n$ is a linear combination of residuals of previous timesteps for any amount of substeps.
\end{proof}

With these two lemmas in hand, we will prove the following lemmas further relating the residuals.

\begin{lemma} \label{lemma:gauge-linear-combination-DIRKs}
    The residual of the Lorenz gauge condition \eqref{eq:DIRKs-gauge-res} at time level $n+1$ is a linear combination of the residuals of the Lorenz gauge \eqref{eq:DIRKs-gauge-res}, the time derivative of the Lorenz gauge \eqref{eq:DIRKs-ddt-gauge-res}, and the continuity equation \eqref{eq:DIRKs-continuity-res}.
\end{lemma}
\begin{proof}
    Consider the Lorenz gauge at time $n+1$:
    
    \begin{equation}
        \epsilon_1^{n+1} = -\phi^{n+1} + \phi^{n} - c^2h\left(R^{s,n}\nabla\cdot\mathbf{A}^{n} + R^{s,n+1}\nabla\cdot\mathbf{A}^{n+1}\right)
    \end{equation}
    We plug in our update equations \eqref{eq:DIRKs-phi-update} and \eqref{eq:DIRKs-A-update}:
    \begin{align}
        \begin{split}
            \epsilon_1^{n+1} &= -\left(\phi^{n} + h\sum_{i=1}^{s}b_ik_i^{\phi,n}\right) + \left(\phi^{n-1} + h\sum_{i=1}^{s}b_ik_i^{\phi,n-1}\right) \\
            &- c^2h\left(R^{s,n}\nabla\cdot\left(\mathbf{A}^{n-1} + h\sum_{i=1}^{s}b_i\mathbf{k}_i^{\mathbf{A},n-1}\right) + R^{s,n+1}\nabla\cdot\left(\mathbf{A}^n + h\sum_{i=1}^{s}b_i\mathbf{k}_i^{\mathbf{A},n}\right)\right) \\
            &= -\phi^n + \phi^{n-1} - c^2h\left(R^{s,n}\nabla\cdot\mathbf{A}^{n-1} + R^{s,n+1}\nabla\cdot\mathbf{A}^{n}\right) \\
            &+h\sum_{i=1}^{s}{b_i\left(-k_i^{\phi,n} + k_i^{\phi,n-1} - c^2h\left(R^{s,n}\mathbf{k}_i^{\mathbf{A},n-1} + R^{s,n+1}\nabla\cdot\mathbf{k}_i^{\mathbf{A},n}\right)\right)}
        \end{split}
    \end{align}

    The first term is by definition $\epsilon_1^n$, the second term is a linear combination of $\epsilon_1^n$, $\epsilon_2^n$, $\epsilon_3^n$, and $\epsilon_3^{n+1}$ by Lemma \ref{lemma:k-linear-combination-DIRKs}.
\end{proof}

\begin{lemma} \label{lemma:ddt-gauge-linear-combination-DIRKs}
    The residual of the time derivative of the Lorenz gauge condition \eqref{eq:DIRKs-ddt-gauge-res} at time level $n+1$ is a linear combination of the residuals of the Lorenz gauge \eqref{eq:DIRKs-gauge-res} at time level $n$, the time derivative of the Lorenz gauge \eqref{eq:DIRKs-ddt-gauge-res} at time level $n$, the continuity equation \eqref{eq:DIRKs-continuity-res} at time levels $n$ and $n+1$.
\end{lemma}
\begin{proof}
    Consider the residual of the time derivative of the Lorenz gauge at time $n+1$:
    \begin{equation}
        \epsilon_2^{n+1} = -v_\phi^{n+1} + v_\phi^{n} - c^2h\left(R^{s,n}\nabla\cdot\mathbf{v}_\mathbf{A}^{n} + R^{s,n+1}\nabla\cdot\mathbf{v}_\mathbf{A}^{n+1}\right) 
    \end{equation}
    We plug in our update equations \eqref{eq:DIRKs-v_phi-update} and \eqref{eq:DIRKs-v_A-update}:
    \begin{align}
        \begin{split}
            \epsilon_2^{n+1} &= -\left(v_\phi^{n} + h\sum_{i=1}^sb_ik_i^{v_\phi,n}\right) + \left(v_\phi^{n-1} + h\sum_{i=1}^sb_ik_i^{v_\phi,n-1}\right) \\
            &- c^2h\left(R^{s,n}\nabla\cdot\left(\mathbf{v}_\mathbf{A}^{n-1} + h\sum_{i=1}^sb_i\mathbf{k}_i^{\mathbf{v}_\mathbf{A},n-1}\right) + R^{s,n+1}\nabla\cdot\left(\mathbf{v}_\mathbf{A}^{n} + h\sum_{i=1}^sb_i\mathbf{k}_i^{\mathbf{v}_\mathbf{A},n}\right)\right) \\
            &= -v_\phi^n + v_\phi^{n-1} - c^2h\left(R^{s,n}\nabla\cdot\mathbf{v}_\mathbf{A}^{n-1} + R^{s,n+1}\nabla\cdot\mathbf{v}_\mathbf{A}^{n}\right) \\
            &+h\sum_{i=1}^{s}b_i\left(-k_i^{v_\phi,n} + k_i^{v_\phi,n-1} - c^2h\left(R^{s,n}\mathbf{k}_i^{\mathbf{v}_\mathbf{A},n-1} + R^{s,n+1}\nabla\cdot\mathbf{k}_i^{\mathbf{v}_\mathbf{A},n}\right)\right).
        \end{split}
    \end{align}
    The first term is by definition $\epsilon_2^{n}$, the second term is a linear combination of $\epsilon_1^n$, $\epsilon_2^n$, $\epsilon_3^{n}$, $\epsilon_3^{n+1}$, and $\epsilon_4^{n}$ by Lemma \ref{lemma:ddt-k-linear-combination-DIRKs}. The same lemma shows $\epsilon_4^n$ is a linear combination of the other four residuals.
\end{proof}

With Lemmas \ref{lemma:gauge-linear-combination-DIRKs} and \ref{lemma:ddt-gauge-linear-combination-DIRKs} at hand, we are now ready to prove the following theorem.

\begin{theorem} \label{thm:DIRKs-continuity-implies-gauge}
    Assuming the residuals of the continuity and gauge condition \eqref{eq:DIRKs-gauge-res}-\eqref{eq:DIRKs-continuity-res} are zero for timestep $n=0$, the residuals of the Lorenz gauge condition \eqref{eq:DIRKs-gauge-res}-\eqref{eq:DIRKs-ddt-gauge-res} are zero for all time steps $n>0$ if the residual of the continuity equation \eqref{eq:DIRKs-continuity-res} is zero for all time $n > 0$.
\end{theorem}
\begin{proof}
    Assume the residuals are zero for the starting conditions. Additionally, assume the residual for the continuity equation is zero for all time. By Lemma \ref{lemma:gauge-linear-combination-DIRKs}, we know $\epsilon_1^{n+1}$ is a linear combination of $\epsilon_1^n$, $\epsilon_2^n$, $\epsilon_3^n$, and $\epsilon_3^{n+1}$. ie

    \begin{equation}
        \epsilon_1^{n+1} = \mathcal{L}_i^{-1}\left[A\epsilon_1^{n} + B\epsilon_2^{n} + C\epsilon_3^{n} + D\epsilon_3^{n+1}\right].
    \end{equation}
    By our assumption, $\epsilon_3^k = 0 \, \forall \, k$. Thus
    \begin{equation}
        \epsilon_1^{n+1} = \mathcal{L}_i^{-1}\left[A\epsilon_1^{n} + B\epsilon_2^{n}\right].
    \end{equation}

    Clearly, if $\epsilon_1^{0} = \epsilon_2^{0} = 0$, then by induction we have $\epsilon^{n+1}_1 = 0$. Thus $\epsilon_1^k = 0 \, \forall \, k$.

    Lemma \ref{lemma:ddt-gauge-linear-combination-DIRKs} gives that $\epsilon_2^{n+1}$ is a linear combination of $\epsilon_1^n$, $\epsilon_2^n$, $\epsilon_3^n$, and $\epsilon_3^{n+1}$. Similar logic will show that our assumptions imply $\epsilon_2^k = 0 \, \forall \, k$.

    If the residual of the continuity equation \eqref{eq:DIRKs-continuity-res} is zero for all time, then the residuals for the Lorenz gauge \eqref{eq:DIRKs-gauge-res}-\eqref{eq:DIRKs-ddt-gauge-res} are zero for all time.
\end{proof}

\begin{remark}
    In this generalization we have proven only one direction of the ``if and only if" relationship indicated by the above schemes. However, for our purposes we are only interested in satisfaction of the continuity equation implying satisfaction of the Lorenz gauge, and so this suffices for now.  The other direction has been proven for DIRK-$2$, indicating generalization to $s$ stages is possible, but this will require more work.
\end{remark}

\begin{remark}
    Given this unidirectional relation, we were able to prove the relationship using only a generic linear combination of residuals, we did not need any specifics as to the coefficients.  To prove the other direction, these details will be necessary.
\end{remark}

\begin{remark} \label{rem:DIRK-Gauss-Law}
    For neither DIRK-$2$ nor the generalized DIRK-$s$ methods have we established that Gauss's law for electricity is satisfied if the semi-discrete Lorenz gauage condition is satisfied.  Stages beyond the first involve compounded inverse linear operators that render such a property intractable to prove.  This is not to suggest that Runge-Kutta methods exclude this property, for example the implicit midpoint rule gives this relationship. Further work needs to be done to design Runge-Kutta methods that demonstrably have this property.
\end{remark}

%
%
%
\subsection{Summary}
\label{subsec:3 Summary}

In this section, composing the bulk of our work, we generalized the properties proven in \cite{christliebPIC2022pt2} to an arbitrary order of BDF methods as well as a set of second-order time centered methods. Additionally, we proved one direction (continuity implying the gauge) of this relationship to DIRK methods with $s$ stages. Having proven these properties, we will now briefly discuss the way we will be using these fields to solve the Vlasov Maxwell system and then move on to examining the numerical results.

%
%
%
%
\section{Overview of the Algorithm}
\label{sec:4 Particle Updates}


In this section we introduce the particle push we use with all wave solvers and give an overview of the overall algorithm.  As we are interested in how to construct time discrete formulations that ensure Gauge, Gauss and Continuity equations are satisfied, we choose to simplify our discussion from our second paper \cite{christliebPIC2022pt2} by making use of spectral (FFT based) solvers for the fields as opposed to the kernel based methods.  
In the following sections, we first review the IAEM and then go over the algorithm as a whole regardless of the time integration method chosen for the field equations.

%
%
\subsection{The Improved Asymmetrical Euler Method}
\label{subsec:Improved AEM method description}

In parts I and II \cite{christliebPIC2022pt1, christliebPIC2022pt2} we presented a semi-implicit particle update scheme called the Improved Asymmetrical Euler Method (IAEM), which modifies the Asymmetrical Euler Method (AEM) of Gibbon et al \cite{Gibbon2017Hamiltonian} by introducing a Taylor expansion term to the momentum update equation. This method is reproduced as follows:
\begin{empheq}[left=\empheqlbrace]{align}
    \mathbf{x}_{i}^{n+1} &= \mathbf{x}_{i}^{n} + \mathbf{v}_{i}^{n} \Delta t, \label{eq:IAEM x update} \\
    \mathbf{v}_{i}^{*} &= 2\mathbf{v}_{i}^{n} - \mathbf{v}_{i}^{n-1}, \label{eq:IAEM v extrapolation}\\
    \mathbf{P}_{i}^{n+1} &= \mathbf{P}_{i}^{n} + q_i \Bigg( - \nabla \phi^{n+1} + \nabla \mathbf{A}^{n+1} \cdot \mathbf{v}_{i}^{*} \Bigg)\Delta t, \label{eq:IAEM P update} \\
    \mathbf{v}_{i}^{n+1} &\equiv \frac{c^2 \left( \mathbf{P}_{i}^{n+1} - q_{i} \mathbf{A}^{n+1}\right)}{\sqrt{ c^2\left( \mathbf{P}_{i}^{n+1} - q_{i} \mathbf{A}^{n+1}\right)^2 + \left(m_{i}c^2\right)^2}}. \label{eq:IAEM v defn}
\end{empheq}
This method, though not proven to be symplectic, has been shown to have volume preserving behavior and reduces the error significantly when compared to \cite{Gibbon2017Hamiltonian}, these comparisons may be found in \cite{christliebPIC2022pt1}.

There are a number of spatial derivatives that are taken throughout this scheme, like our approach to computing $\rho$ from the continuity equation, (section \ref{sec:3 wave solvers}).  It is important to compute the derivatives in the same manner as the wave equations are solved, ie if the $\mathcal{L}$ operator is inverted using the FFT, the derivatives here need to be computed with an FFT, otherwise the above theorems break.  We will be computing these spatial derivatives spectrally.



For all particle updates, whether the CDF, BDF or DIRK field integrators are used, the particles are advanced using the IAEM.  $\mathbf{J}^{n+1}$ on the mesh comes from a symmetric second order B-Spline weighting of the current to the mesh \cite{birdsall1991particle,huang2006quickpic,Benedetti_2008,cormier2008unphysical,lehe2014laser,Shalaby_2017}, where the particle location at time $t^{n+1}$ determines where the spline weights, and the particle velocity at time $t^{n}$ determines the value that is interpolated (along with charge).  One can alternatively use  $\mathbf{v}^{*}$ when computing $\mathbf{J}^{n+1}$, but in all tests in this paper this made no discernible difference in the solution.  Unlike staggered methods, we do not introduce a reduction in the order of the weighting spline for the current $\mathbf{J}_i$ in directions perpendicular to the direction $i\in{\{x,y,z\}}$ as in \cite{VillasenorChargeConservation92}.  The spline is detailed in \cite{birdsall1991particle} and section 8.8 in \cite{BirdsallLangdon}. In 2D and 3D the spline is given by tensor products of 1D splines in x, y and z directions.  For further details on B-splines themselves, see Chapter 9 of \cite{deBoor}.

We note that for staggered in time fields, $\mathbf{J}$ lives at integrator times and we compute $\rho$ at the half time levels using the time staggered continuity equation.  This lets us update the particles with IAEM, where $\phi^n=\frac{1}{2}(\phi^{n+1/2}+\phi^{n-1/2})$ for second order in time approximation to $\phi^n$, which is all that is needed for the CDF method.

%
%
%
%
\subsection{Algorithm} 
\label{subsec:4 Algorithm-nonstaggered}

In this section we outline the overall algorithm, see Algorithm \ref{Algorithm1}.  The outline covers the choices made that are consistent with the original two papers, but cast the the context of the generalized framework.  

\begin{algorithm}[ht!]
    \caption{Outline of the PIC algorithm with the improved asymmetric Euler method (IAEM).}\label{Algorithm1}
    Perform one time step of the PIC cycle using the improved asymmetric Euler method. 
    \label{alg:IAEM}
    \begin{algorithmic}[1]

    \State \textbf{Given}: $(\mathbf{x}_{i}^{0}, \mathbf{P}_{i}^{0}, \mathbf{v}_{i}^{0})$, as well as the fields $\left( \phi^{0}, \nabla \phi^{0} \right)$ and $ \left( \mathbf{A}^{0}, \nabla \mathbf{A}^{0} \right)$

    \State Initialize $\mathbf{v}_{i}^{-1} = \mathbf{v}_{i}(-\Delta t)$ using a Taylor approximation.

    \While{stepping}
    
    \vspace{10pt}
    
    \State \label{start time loop} Update the particle positions with $$\mathbf{x}_{i}^{n+1} = \mathbf{x}_{i}^{n} + \mathbf{v}_{i}^{n} \Delta t.$$
    
    \State Using the position data $\mathbf{x}_{i}^{n+1}$, and velocity data $\mathbf{v}_{i}^{n}$, compute the current density $\tilde{\mathbf{J}}^{n+1} = \mathbf{J}^{n+1} + \mathcal{O}(\Delta t),$
    by using quadratic weighting of $\tilde{\mathbf{J}}_i^{n+1}$ to the mesh for all $i$.
    
    \State Using the continuity equation compute the charge density from the current density.  Here we employ the  integrator that is consistent with the update the field:
    \begin{itemize}
        \item CDF: $$\rho^{n+3/2} = \rho^{n+1/2} - \Delta t\nabla\cdot\tilde{\mathbf{J}}^{n+1},$$
        \item BDF: 
        $$\rho^{n+1} = \sum_{i=n-k}^{n}\frac{a_{n-i}}{a_0}\rho^{i} + \frac{\Delta t}{a_0} \nabla\cdot\tilde{\textbf{J}}^{n+1} $$
        \item DIRK: 
        $$\rho^{n+1} = \rho^{n} - \Delta t \sum_{i=1}^k\left(b_i\nabla\cdot \textbf{J}^{(i)} \right )$$
        where $b_i$ are the DIRK coefficients and the $\textbf{J}^{(i)} $ are constructed via linear interpolation using the data $\{\tilde{\mathbf{J}}^{n}, \tilde{\mathbf{J}}^{n+1}\}$.  Higher order construction could use more data if desired, or one could substep the IAEM (or other particle stepper) to get $\textbf{J}^{(i)}$ at each stage.
    \end{itemize}
    
    \State Compute the potentials $\phi^{*+1}$ and $\mathbf{A}^{n+1}$ at time level $t^{n+1}$ using the proposed time stepping method (see section \ref{sec:3 wave solvers}) and the FFT solution to the boundary value problem framed as a Helmholtz operator.  Here $*\in \{ n, n+1/2\}$ depending on whether it is BDF and DIRK or the CDF method. 
    
    \State Compute the spatial derivatives of the potentials, $\nabla \phi^{*+1}$ and $\nabla \mathbf{A}^{n+1}$, on the mesh at time level $t^{n+1}$ using an FFT.

    \State Compute $\nabla \tilde{\phi}_i^{n+1}$ and $\nabla \mathbf{A}_i^{n+1}$ at the particles using quadratic weighting from the mesh to the particles, where $\tilde{\phi}_i^{n+1} =  \phi_i^{n+1} $ for BDF and DIRK methods, and $\tilde{\phi}_i^{n+1} =  \frac{1}{2}(\phi_i^{n+3/2}+\phi_i^{n+1/2}) $  for the CDF method.
    
    \State Evaluate the Taylor corrected particle velocities $$\mathbf{v}_{i}^{*} = 2 \mathbf{v}_{i}^{n} - \mathbf{v}_{i}^{n-1}.$$
    
    \State Calculate the new generalized momentum according to $$\mathbf{P}_{i}^{n+1} = \mathbf{P}_{i}^{n} + q_i \Big( - \nabla \tilde{\phi}_i^{n+1} + \nabla \mathbf{A}_i^{n+1} \cdot \mathbf{v}_{i}^{*} \Big)\Delta t.$$

    \State Convert the new generalized momenta into new particle velocities with $$ \mathbf{v}_{i}^{n+1} =  \frac{c^2 \left( \mathbf{P}_{i}^{n+1} - q_{i} \mathbf{A}_i^{n+1}\right)}{\sqrt{ c^2\left( \mathbf{P}_{i}^{n+1} - q_{i} \mathbf{A}_i^{n+1}\right)^2 + \left(m_{i}c^2\right)^2}}.$$
    
    \State Shift the time history data and return to step \ref{start time loop} to begin the next time step. 
    
    \EndWhile
    \end{algorithmic}
\end{algorithm}

%
%
\subsection{Summary}
\label{subsec:4 Summary}

In this section, we outlined the time stepping method used with the particles and the overall algorithm.  The method as outlined conserves the gauge condition for all three types for time discretization methods and preserves Gauss's law for both BDF and CDF methods.  In the next section we present numerical results that support this statement for the BDF-$1$, BDF-$2$, CDF-$2$ and DIRK-$2$ consistent time discretization strategies.   

%
%
%
%
\section{Numerical Experiments}
\label{sec:5 Numerical results}

This section contains the numerical results for our PIC method using several wave solvers described above.  We consider two different test problems with dynamic and steady state qualities. We first consider the Weibel instability \cite{Weibel1959}, which is a streaming instability that occurs in a periodic domain. The second example is a simulation of a non-relativistic drifting cloud of electrons in a periodic domain. For each example, we compare the performance of the different methods by inspecting the fully discrete Lorenz gauge condition and tracking its behavior as a function of time.  Our exploration concentrates on the BDF-$1$, BDF-$2$, CDF-$2$, and DIRK-$2$ methods, typically with a mesh refinement of $128\times 128$. The DIRK-$2$ method's weights are those explored by Qin and Zhang \cite{qin_zhang_1992} and has the following Butcher tableau:

\begin{equation*}
    \begin{array}
    {c|cccc}
    1/4 & 1/2 & 0\\
    1/4 & 1/4 & 1/2\\
    \hline
    & 1/2& 1/2
    \end{array}
\end{equation*}
All methods use a spectral solve to invert the $\mathcal{L}$ operator and compute gradients spectrally as well. As mentioned above, we make use of quadratic particle weighting as in \cite{Benedetti_2008,Shalaby_2017,birdsall1991particle,huang2006quickpic,cormier2008unphysical,lehe2014laser}, and we compute $\mathbf{J}^{n+1}$ using $\mathbf{x}^{n+1}_i$ as the location of the interpolant point and $\mathbf{v}^{n}_i$ and $q_i$ as the value that is interpolated.

We conclude the section by summarizing the key results of the experiments.

%
%
\subsection{Weibel Instability}
\label{subsec:weibel instability}

The Weibel instability \cite{Weibel1959} is a type of filamentation instability known to occur in electromagnetic plasmas in which there is an anistropic distribution of momenta. This anisotropy is prevalent in many applications of high-energy-density physics including astrophysical plasmas \cite{WeibelAstro2021} and fusion applications \cite{FluidKineticICFsims2005}. In such applications, the momenta in different directions can vary by several orders of magnitude. The strong currents resulting from these momenta create filaments that eventually interact due to the growth in the magnetic field. Over time, the magnetic field can become quite turbulent and the currents self-organize into larger connected networks. During this self-organization phase, there is an energy conversion mechanism that transforms the kinetic energy from the beams into turbulent magnetic fields in an attempt to make the momentum isotropic. The resulting instability creates the formation of magnetic islands and other structures which store massive amounts of energy. In such highly turbulent regions, this can lead to the emergence of other plasma phenomena such as magnetic reconnection, in which energy is released from the fields back into the plasma \cite{FonsecaWeibel03}.

We set up the simulation following a setup that is similar to \cite{MorseWeibel1971} which can be summarized as follows. The domain for the problem is the periodic square $[L_x,L_y)^2$, where $L_{x} = L_{y} = 1.153896$ are in units of $c/\omega_{pe}$. Here, $c$ is the speed of light and $\omega_{pe}$ is the angular plasma frequency for electrons. The system is initially current neutral and consists of two infinite counter-propagating sheets of electrons whose velocities are $\pm v_{e}$ in the $x$ component of the velocity. The drift velocity of the electrons $v_{b}$ is chosen to give a Lorentz factor of $\gamma_{e} = 2.5$, which translates to $v_{b} \approx 0.92 c$. The electrons are prescribed a small thermal velocity $v_{th} = 1\times 10^{-4} v_{b}$ along only the $y$ direction, which acts as a small perturbation causing the streams to interact. In other papers, this perturbation can be introduced by changing the velocity of all particles in a given cell \cite{Chen-ImplicitPIC-Darwin2014}. A neutralizing background of stationary ions is present to ensure that the system is initially charge neutral, as well. For simplicity, we shall assume that the sheets have the same number density $\bar{n}$, but this assumption is not necessary. In fact, some earlier studies by \cite{LeeFilamentation1973} explored the structure of the instability for interacting streams with different densities and drift velocities. The particular values of the plasma parameters used in this experiment are summarized in Table \ref{tab:weibel plasma parameters}. In our previous work \cite{christliebPIC2022pt2} we ran this problem to examine the response of the Lorenz gauge in the context of the BDF-$1$ method. A fuller discussion of the derivation of the setup can be found there.

\begin{table}[!ht]
    \centering
    \def\arraystretch{1.2}
    \begin{tabular}{ | c || c | }
        \hline
        \textbf{Parameter}  & \textbf{Value} \\
        \hline
        Average number density ($\bar{n}$) [m$^{-3}$] & $1.0\times 10^{10}$ \\
        Average electron temperature ($\bar{T}$) [K] & $1.0\times 10^{4}$ \\
        Electron angular plasma period ($\omega_{pe}^{-1}$) [s/rad] & $1.772688\times 10^{-7}$ \\
        Electron skin depth ($c/\omega_{pe})$ [m] & $5.314386\times 10^{1}$ \\
        Electron drift velocity in $x$ ($v_{\perp}$) [m/s] & $c/2$ \\
        Maximum electron velocity in $y$ ($v_{\parallel}$) [m/s] & $c/100$ \\
        \hline
    \end{tabular}
    \caption{Plasma parameters used in the simulation of the Weibel instability. All simulation particles are prescribed a drift velocity corresponding to $v_{\perp}$ in the $x$ direction while the $y$ component of their velocities are sampled from a uniform distribution scaled to the interval $[-v_{\parallel}, v_{\parallel})$.}
    \label{tab:weibel plasma parameters}
\end{table}

\begin{figure}[!ht]
    \textbf{Weibel Instability Magnetic Magnitude vs Angular Plasma Period, Mesh Resolution: $128 \times 128$}
    \centering
    \subfloat[][Naive Update, Quadratic Weighting]{
    \includegraphics[clip, trim={0cm, 0cm, 21cm, 0cm}, scale=0.185]{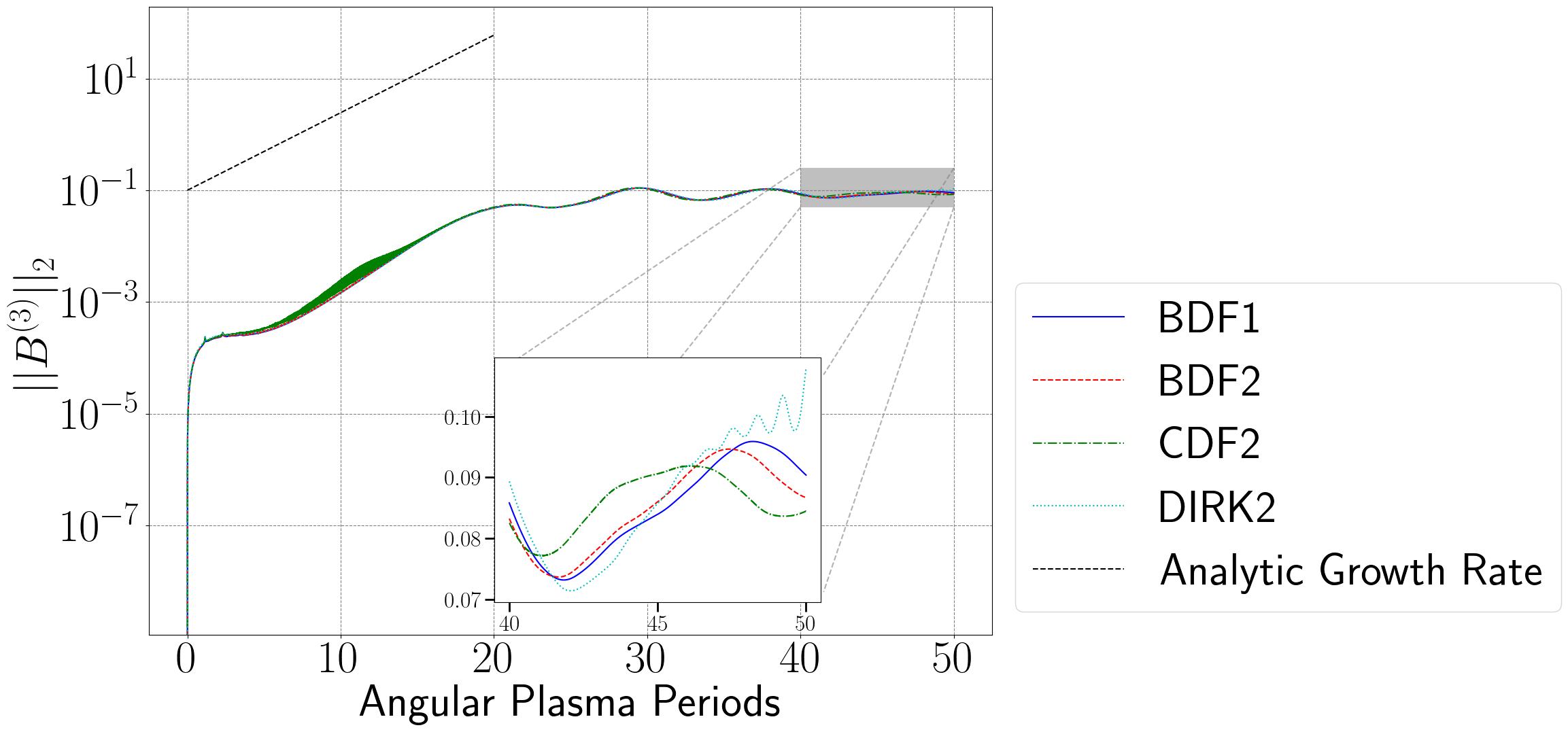}
    \label{fig:naive-quadratic-Bz-growth-rate}
    }
    \subfloat[][Charge Conserving Update, Quadratic Weighting]{
    \includegraphics[clip, trim={0cm, 0cm, 0cm, 0cm}, scale=0.185]{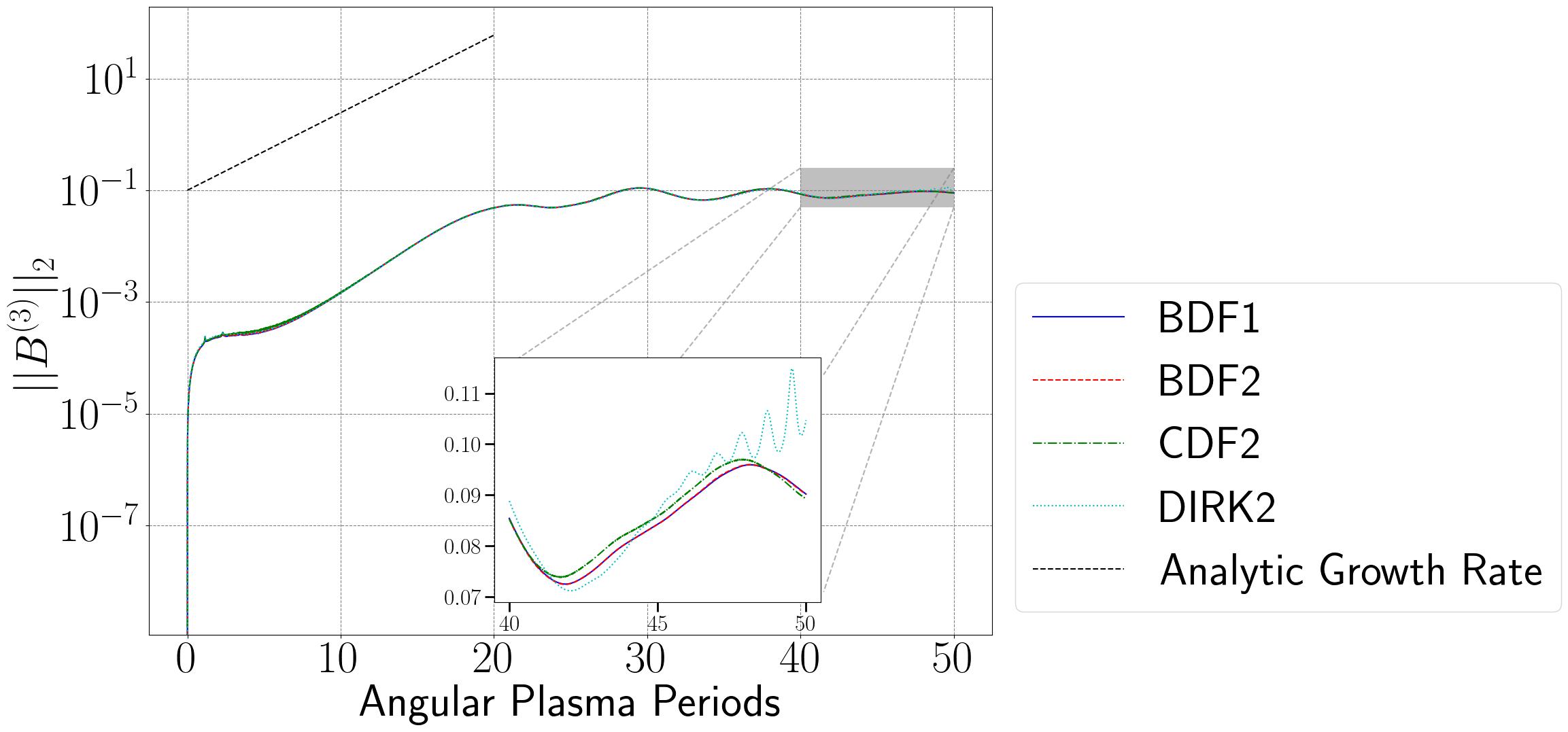}
    \label{fig:conserving-quadratic-Bz-growth-rate}
    } 
    \caption{The growth rate of the magnetic magnitude under the four methods under consideration. We note that both the naive update (left) and charge conserving update (right) follow the theoretical linear growth rate well, though the naive exhibits nonphysical oscillations and is clearly not as smooth. The naive update, unlike the charge conserving update, does not use the continuity equation to update $\rho$, but rather maps both $\mathbf{J}$ and $\rho$ to the mesh using quadratic particle weighting. When we enforce continuity, the oscillations observed in not only the growth period, but also the final state, are removed for all methods except DIRK-$2$, which triggers small two-stream instabilities (see Remark \ref{rem:DIRK2-Oscillations}).}
    \label{fig:quadratic-Bz-growth-rate}
\end{figure}

\begin{figure}[!ht]
    \centering
    \textbf{Weibel Instability Gauge Error vs Angular Plasma Period, Mesh Resolution: $128 \times 128$}
    \subfloat[][Naive Update, Quadratic Weighting]{
    \includegraphics[clip, trim={0cm, 0cm, 11cm, 0cm}, scale=0.185]{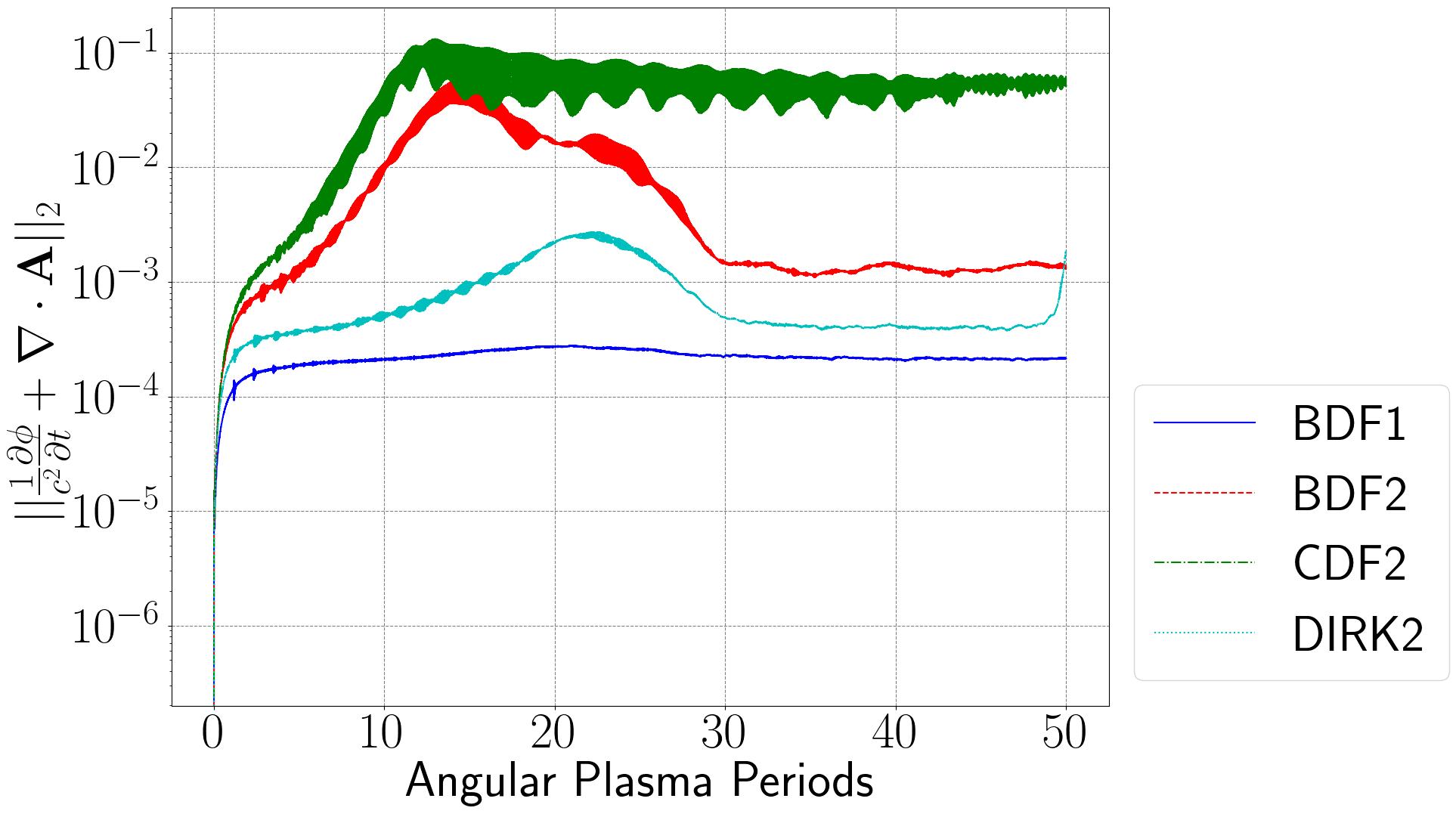}
    \label{fig:naive-conserving-gauge-naive}
    }
    \subfloat[][Charge Conserving Update, Quadratic Weighting]{
    \includegraphics[clip, trim={0cm, 0cm, 0cm, 0cm}, scale=0.185]{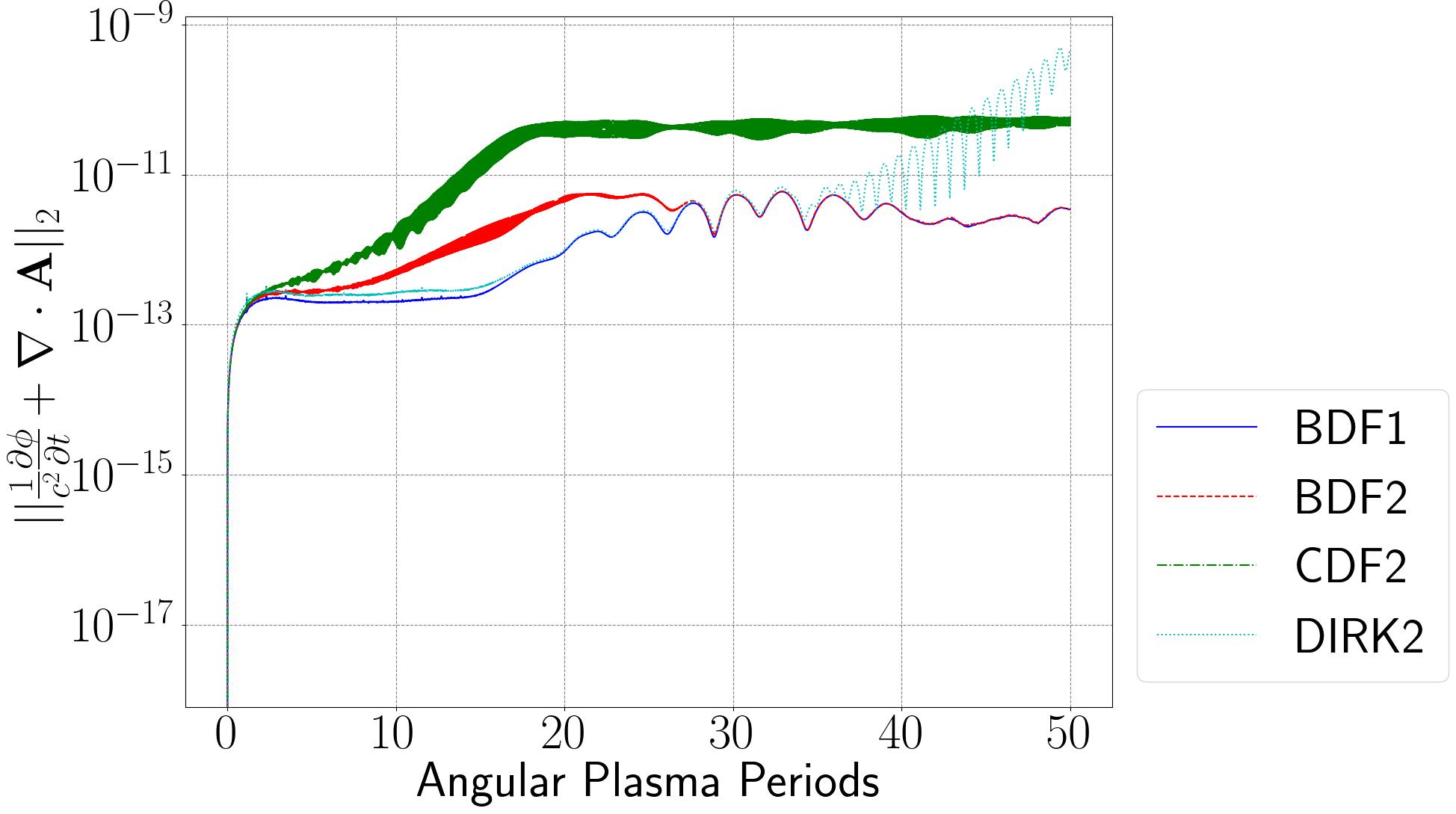}
    \label{fig:naive-conserving-gauge-conserving}
    } 
    \caption{The error in the Lorenz gauge condition for the Webeil Instability on a $128 \times 128$ grid. On the left we scatter charge and current density in a naive manner, on the right we have the setup in which we scatter current and update charge from the continuity equation. We see the error in the gauge is drastically reduced when we enforce continuity.}
    \label{fig:naive-conserving-gauge}
\end{figure}

As noted earlier, we make use of higher order particle weighting as we are using a spectral solver.  We see not only good alignment with the linear growth rate predicted by theory (Figures \ref{fig:quadratic-Bz-growth-rate}), but we also see good preservation of the Gauge condition (Figure \ref{fig:naive-conserving-gauge}) and Gauss's Law (Figure \ref{fig:naive-conserving-gauss}) in all cases other than DIRK-$2$, which as we have previously described in Remark \ref{rem:DIRK-Gauss-Law} does not have the latter property.
In Figure \ref{fig:quadratic-Bz-growth-rate} we show the growth rate of the magnetic magnitude as a function of angular plasma period.  Figure \ref{fig:naive-quadratic-Bz-growth-rate} is the growth rate when using the ``naive" weight where both charge density and current density are mapped to a mesh using the same weighting scheme on the collocated mesh.  Figure \ref{fig:conserving-quadratic-Bz-growth-rate} is the result of making use of the conservative method proposed in this paper.  This shows the physical benefits quite nicely, as we see the oscillations in the growth rate eliminated and the change in the magnetic magnitude over time become much smoother.  Likewise we exhibit the significant improvement in the Gauge condition and Gauss's law by comparing the naive weighting (Figures \ref{fig:naive-conserving-gauge-naive} and \ref{fig:naive-conserving-gauss-naive}, respectively) with the conserving scheme (Figures \ref{fig:naive-conserving-gauge-conserving} and \ref{fig:naive-conserving-gauss-conserving}, respectively).

\begin{remark}
    BDF-$1$, BDF-$2$, and CDF-$2$ with Crank-Nicolson all have diffusion, and as such this removes high frequencies, whereas the high frequencies are resolved by the Green's function solver used in \cite{christliebPIC2022pt2}.
\end{remark}
\begin{remark} \label{rem:DIRK2-dispersion}
    DIRK-$2$ behaves dispersively, and as such high frequencies will remain when we use a spectral solver.
\end{remark}
\begin{remark} \label{rem:DIRK2-Oscillations}
    We note that DIRK-$2$ exhibits oscillations towards the end that eventually saturate (Figure \ref{fig:weibel-DIRK2-refine-interpolation}). Due to its dispersive nature (Remark \ref{rem:DIRK2-dispersion}) it perturbs the cold direction ($v_\perp$) and triggers small two stream instabilities. Higher order particle weighting reduces the effect but does not eliminate it. 
\end{remark}


All methods, BDF-$1$, BDF-$2$, CDF-$2$, and DIRK-$2$, capture the Weibel instability, which is a more challenging problem due to the care needed in setting up the problem,  and additionally, once we enforce the semi-discrete continuity equation, they exhibit excellent gauge error and agreement with Gauss's Law (again, except DIRK-$2$ for Gauss).  The DIRK-$2$ method gives results that merit some comment. Figure \ref{fig:weibel-DIRK2-refine-interpolation} shows the gauge error increasing in an oscillatory fashion at around the thirtieth plasma period, eventually stabilizing around the end of the run.  DIRK-$2$ has no diffusion and cannot resolve high frequencies, as such it triggers nonphysical two stream instabilities in the cold direction (Remark \ref{rem:DIRK2-Oscillations}).  Refining the interpolation scheme mitigates this effect, but does not entirely remove it. We see the downstream effects of this in the magnetic magnitude growth in the quadratic scheme (Figure \ref{fig:weibel-DIRK2-refine-interpolation-B3}), though this matches the linear growth rate theory just as well as the other methods, towards the end the magnetic magnitude begin oscillating slightly.

There are a variety of ways we can approach this problem, the simplest of which would be adding numerical diffusion. Additionally, moving to a fully finite element basis based on $C^1$ elements or a new version of the Green's function solvers (eg \cite{christlieb_Gong_Yang_2024}) would mitigate the issue, and will be the subject of future work.  We are interested in developing Runge-Kutta methods that have the property of satisfying Gauss's Law if they satisfy the Gauge condition, and this too shall be the subject of future inquiry.

\begin{figure}[!ht]
    \centering
    \textbf{Weibel Instability Error in Gauss's Law vs Angular Plasma Period, Mesh Resolution: $128 \times 128$}
    \subfloat[][Naive Update, Quadratic Weighting]{
    \includegraphics[clip, trim={0cm, 0cm, 11cm, 0cm}, scale=0.185]{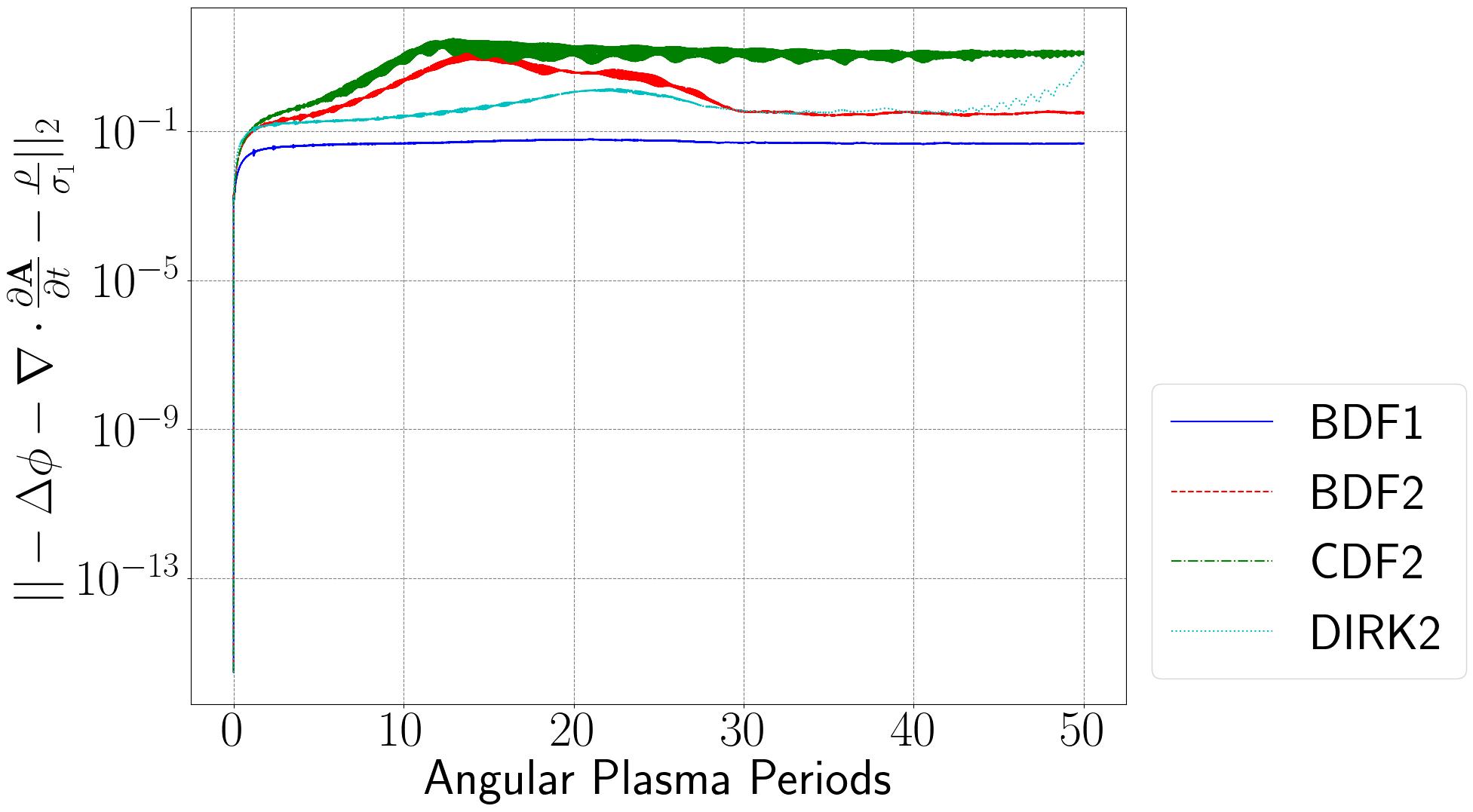}
    \label{fig:naive-conserving-gauss-naive}
    }
    \subfloat[][Charge Conserving Update, Quadratic Weighting]{
    \includegraphics[clip, trim={0cm, 0cm, 0cm, 0cm}, scale=0.185]{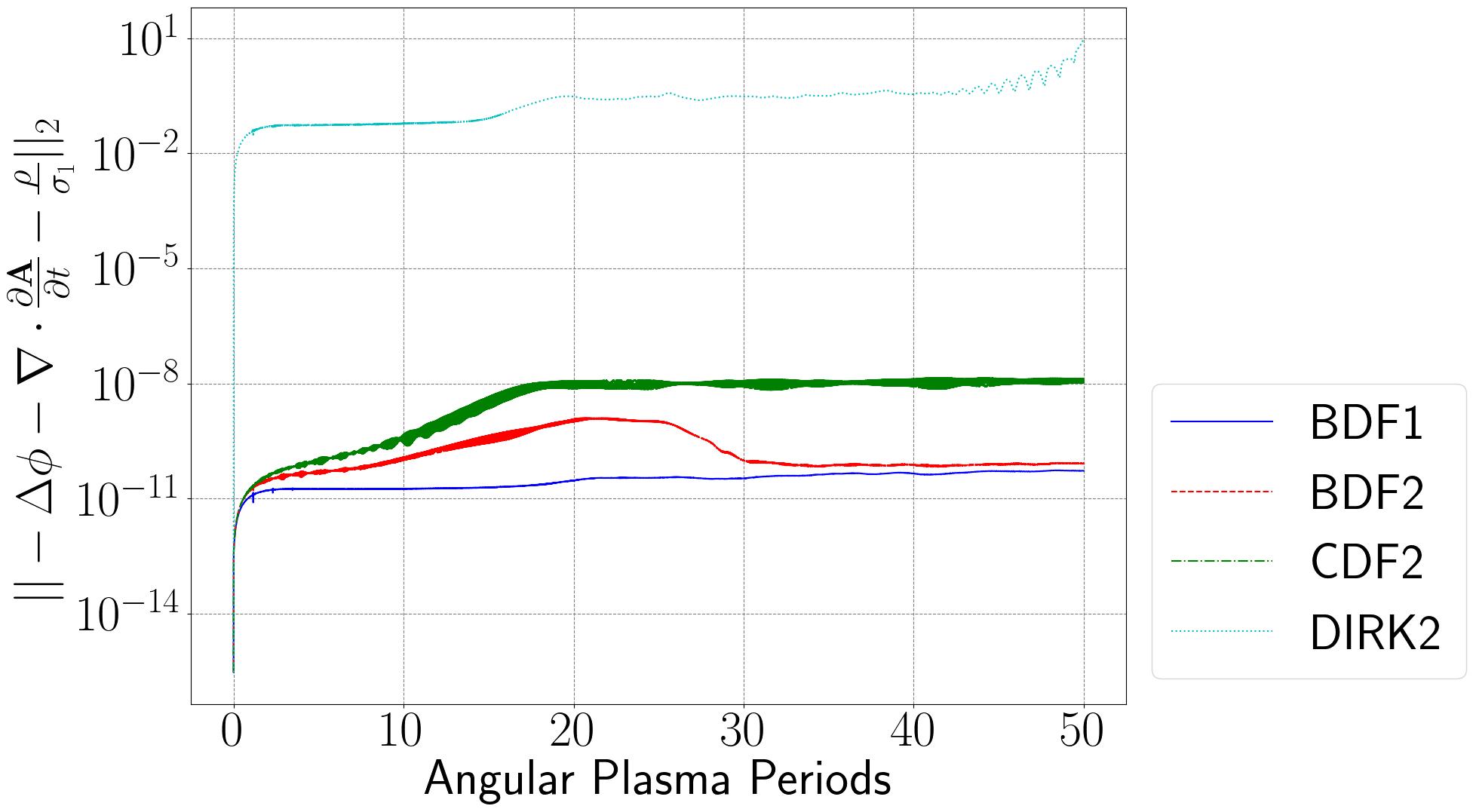}
    \label{fig:naive-conserving-gauss-conserving}
    } 
    \caption{The error in Gauss's Law on a $128 \times 128$ grid. On the left we have scatter charge and current density in a naive manner, on the right we have the setup in which we scatter current and update charge from the continuity equation. We see the error in Gauss's law is drastically reduced with the exception of DIRK-$2$ (see Remark \ref{rem:DIRK-Gauss-Law}).}
    \label{fig:naive-conserving-gauss}
\end{figure}

\begin{figure}[!ht]
    \centering
    \textbf{Weibel Instability, Charge Conserving DIRK-$2$, Mesh Resolution: $128 \times 128$}
    \subfloat[][Gauge Error vs Angular Plasma Period]{
    \includegraphics[clip, trim={0cm, 0cm, 22cm, 0cm}, scale=0.185]{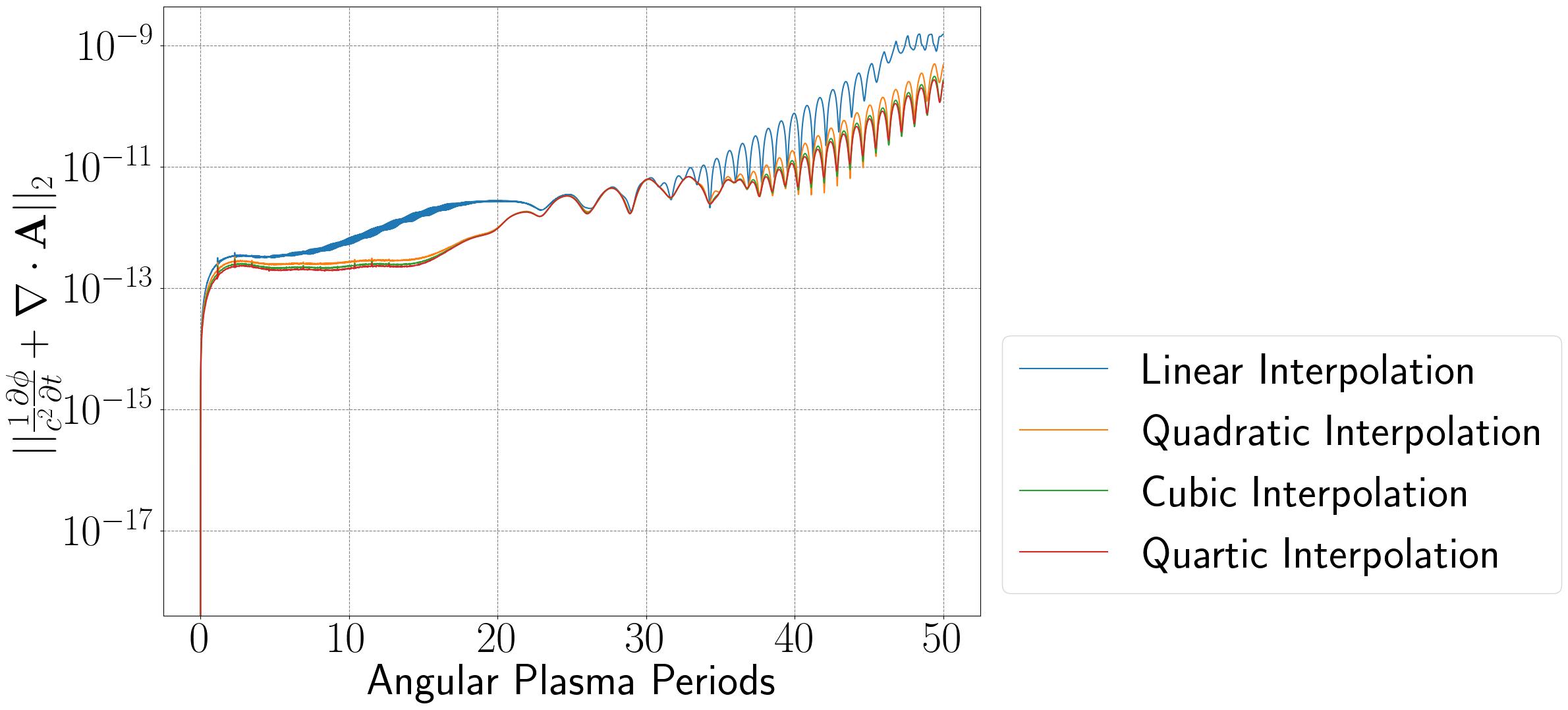}
    \label{fig:weibel-DIRK2-refine-interpolation-gauge}
    }
    \subfloat[][Magnetic Magnitude vs Angular Plasma Period]{
    \includegraphics[clip, trim={0cm, 0cm, 0cm, 0cm}, scale=0.185]{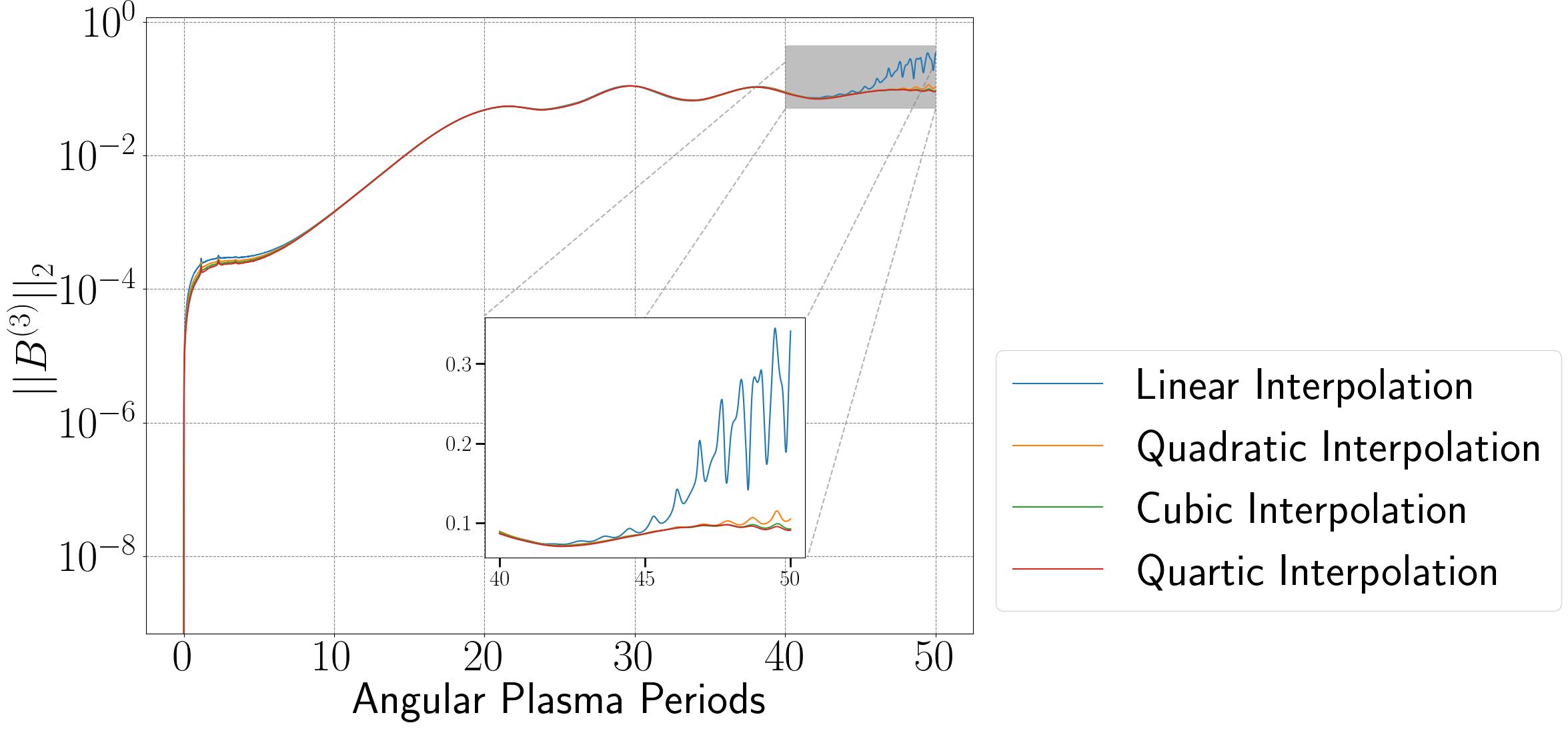}
    \label{fig:weibel-DIRK2-refine-interpolation-B3}
    } 
    \caption{We refine interpolation schemes by increasing the order of spline basis functions for the DIRK-$2$ method.  We see above excellent satisfaction for the gauge condition (left) until around the thirtieth plasma period, after which oscillations begin forming for around twenty plasma periods, eventually themselves stabilizing. This is unlike BDF-$1$, BDF-$2$, or CDF-$2$. Increasing the order of interpolation reduces this phenomenon, but does not remove it. This is a two-stream instability triggered by the DIRK-$2$ method's dispersive nature (see Remarks \ref{rem:DIRK2-dispersion} and \ref{rem:DIRK2-Oscillations}).}
    \label{fig:weibel-DIRK2-refine-interpolation}
\end{figure}

%
%
\subsection{Drifting Cloud of Electrons}
\label{subsec:moving-blob}

In a plasma that is macroscopically neutral, errors in Gauss's law may be mitigated, as an error in one direction caused by a negative particle may very well be cancelled by a likewise error from a positive particle nearby. In \cite{christliebPIC2022pt2} we introduced the $2D-2V$ problem of a drifting cloud of electrons. A stationary group of ions are distributed in a Gaussian about the center of the periodic domain, $\left[-\frac{L_x}{2},-\frac{L_y}{2}\right) \times \left[\frac{L_x}{2},\frac{L_y}{2}\right)$, with a mobile group of electrons distributed in the exact same manner. This group of electrons is given a drift of $c/100$ in the $x$ and $y$ direction, which is enough for some of them to escape the potential well of the ions, while the bulk of the electrons falls back in the well (see Figure \ref{fig:moving-cloud-scatter-plots} for a visualization). Over time they move into the more neutral spaces, spreading apart as they go, and as they do there is ample opportunity for violations in Gauss' law and the Lorenz gauge condition. In this work we modify the dimensions slightly from those in \cite{christliebPIC2022pt2}, increasing the domain from $[-.5,.5]^2$ to $[-8,8]^2$, where the units are in Debye lengths, to increase the amount of space that is traversed and make the problem more sensitive to errors in Gauss's Law. We increase the length of the run of the simulation to $100$ angular plasma periods. A summary of the plasma parameters may be found in Table \ref{tab:moving blob plasma parameters}.

In Figure \ref{fig:naive-conserving-moving-blob-gauge} we observe significant improvement in the gauge error when we compare the naive scatter method to the method enforcing charge conservation. Likewise when we compare satisfaction of Gauss's Law in Figure \ref{fig:naive-conserving-quadratic-moving-blob-gauss} between the naive scatter and charge conserving methods we find many orders of magnitude of improvement, confirming our theory. 

\begin{table}[!ht]
    \centering
    \def\arraystretch{1.2}
    \begin{tabular}{ | c || c | }
        \hline
        \textbf{Parameter}  & \textbf{Value} \\
        \hline
        Average number density ($\bar{n}$) [m$^{-3}$] & $1.0\times 10^{13}$ \\
        Average electron temperature ($\bar{T}$) [K] & $1.0\times 10^{5}$ \\
        Electron angular plasma period ($\omega_{pe}^{-1}$) [s/rad] & $5.6057\times 10^{-9}$ \\
        Electron skin depth ($c/\omega_{pe})$ [m] & $1.6806$ \\
        Electron drift velocity in $x$ and $y$ [m/s] & $c/100$ \\
        \hline
    \end{tabular}
    \caption{Plasma parameters used in the simulation of the Moving Cloud of Electrons problem. All simulation particles are prescribed a drift velocity corresponding to $v_{x} = v_{y} = \frac{c}{100}$ with a thermal velocity corresponding to a Maxwellian.}
    \label{tab:moving blob plasma parameters}
\end{table}

\begin{figure}[!ht]
    \centering
    \textbf{Moving Cloud Gauge Error vs Angular Plasma Period, Mesh Resolution: $128 \times 128$}
    \subfloat[][Naive Conserving Update, Quadratic Weighting]{
    \includegraphics[clip, trim={0cm, 0cm, 11cm, 0cm}, scale=0.185]{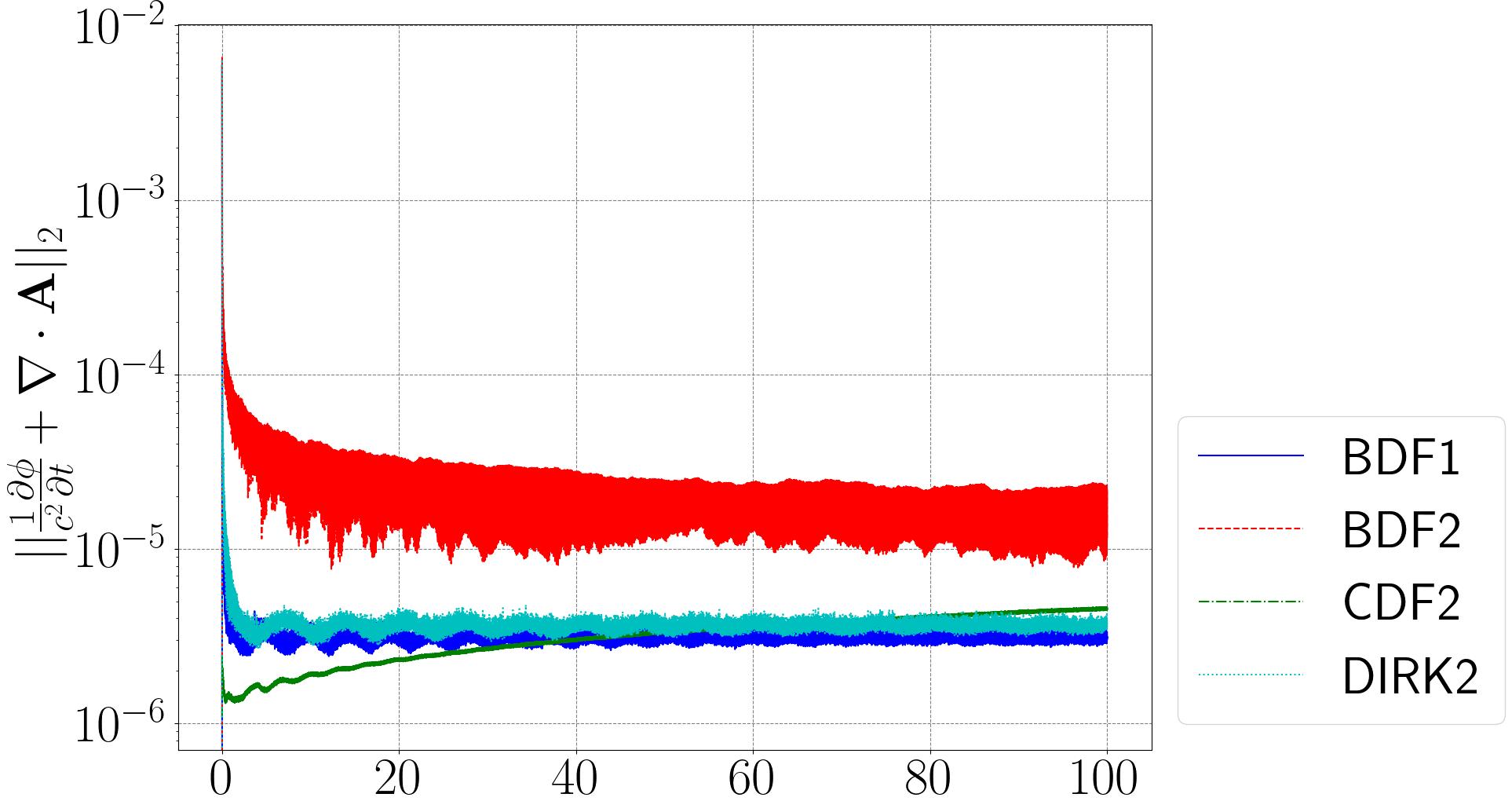}}
    \subfloat[][Charge Conserving Update, Quadratic Weighting]{
    \includegraphics[clip, trim={0cm, 0cm, 0cm, 0cm}, scale=0.185]{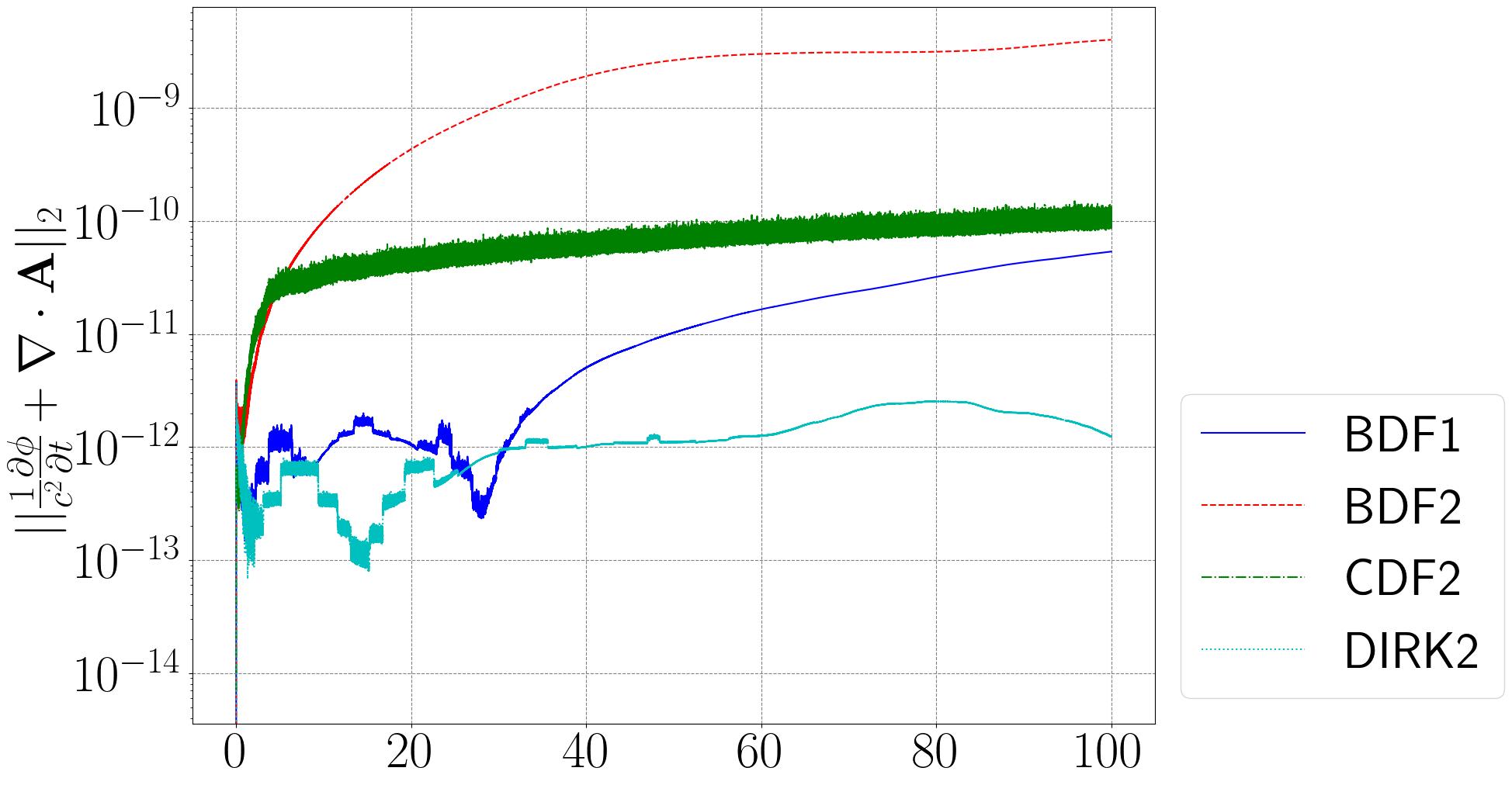}} 
    \caption{The error in the Lorenz gauge for the Moving Cloud of Electrons problem. We see on the left side relatively high errors in Gauss, whereas the right side shows significantly less.}
    \label{fig:naive-conserving-moving-blob-gauge}
\end{figure}

\begin{figure}[!ht]
    \centering
    \textbf{Moving Cloud Error in Gauss's Law vs Angular Plasma Period, Mesh Resolution: $128 \times 128$}
    \subfloat[][Naive Conserving Update, Quadratic Weighting]{
    \includegraphics[clip, trim={0cm, 0cm, 11cm, 0cm}, scale=0.185]{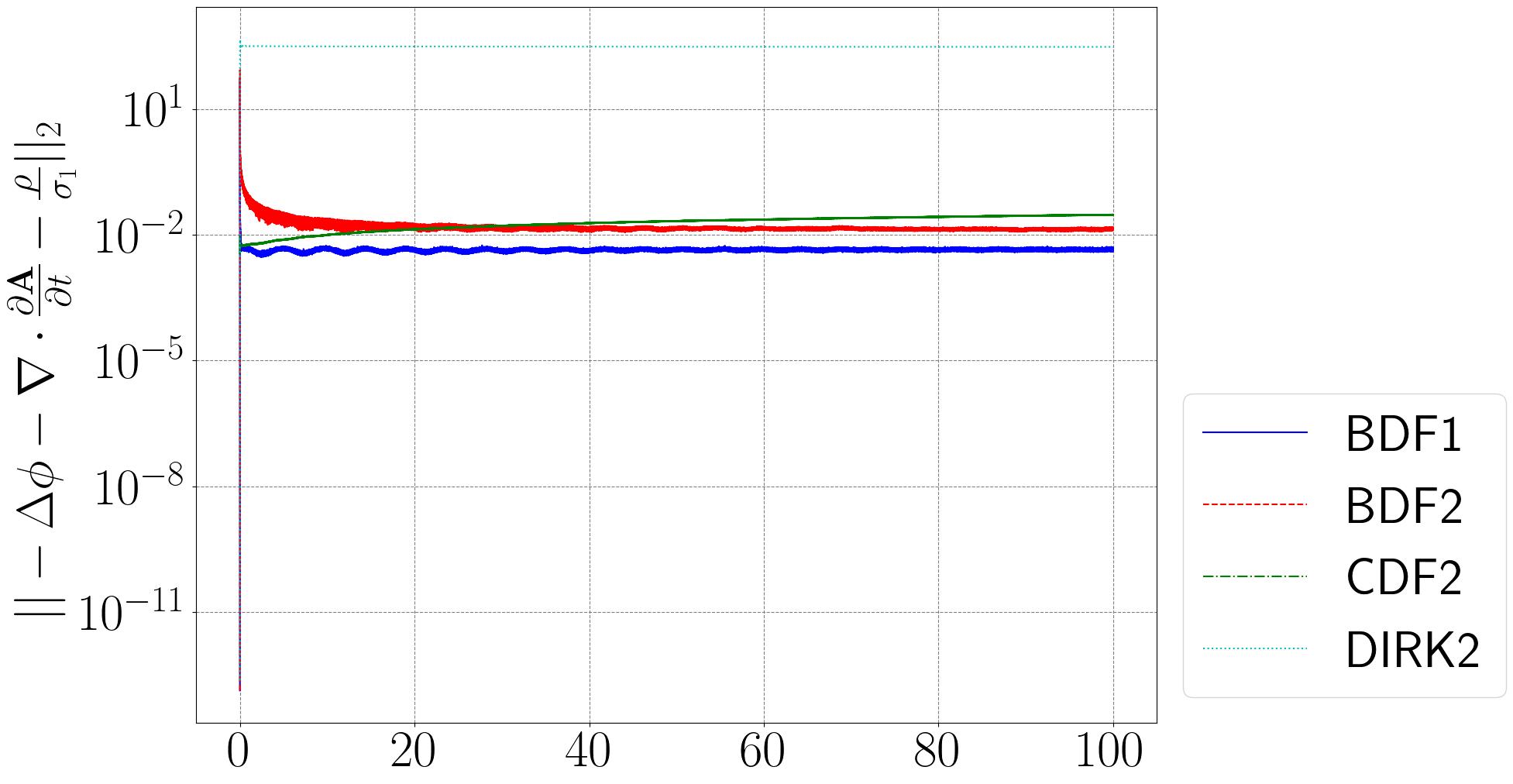}}
    \subfloat[][Charge Conserving Update, Quadratic Weighting]{
    \includegraphics[clip, trim={0cm, 0cm, 0cm, 0cm}, scale=0.185]{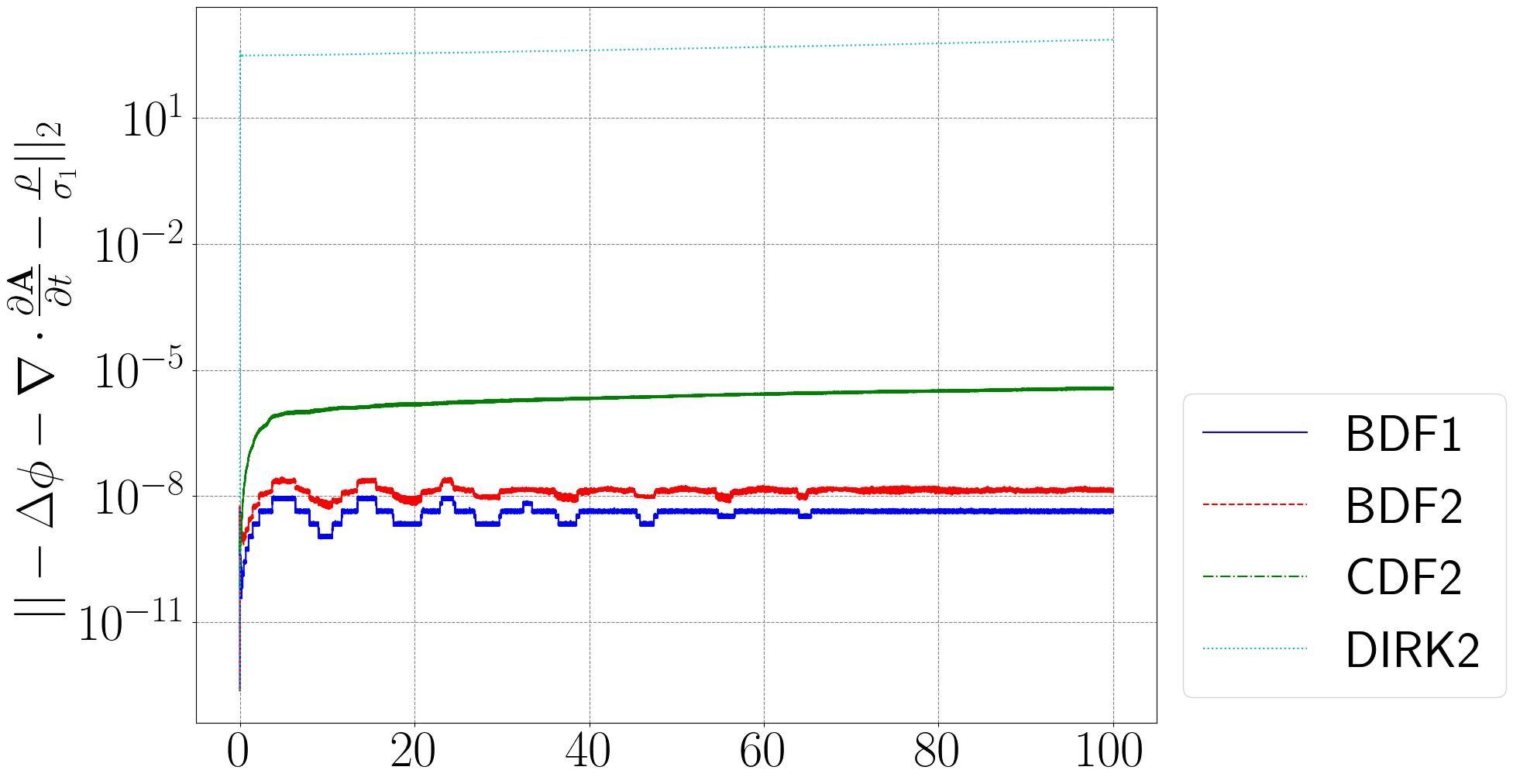}} 
    \caption{The error in Gauss's law for the Moving Cloud of Electrons problem. We see on the left side relatively high errors in Gauss, whereas the right side shows significantly less. As previously discussed, because the Runge-Kutta method is not a nested operator in time in a similar manner to the CDF or BDF methods, there is no straightforward way to enable DIRK methods to satisfy Gauss's Law, as we would need to apply a DIRK time difference to a time difference to get the second time derivative. See Remark \ref{rem:DIRK-Gauss-Law}.}
    \label{fig:naive-conserving-quadratic-moving-blob-gauss}
\end{figure}

\begin{figure}[!ht]
    \centering
    \subfloat[$t = 0$]{
    \includegraphics[clip, trim={0cm, 0cm, 0cm, 0cm}, scale=0.125]{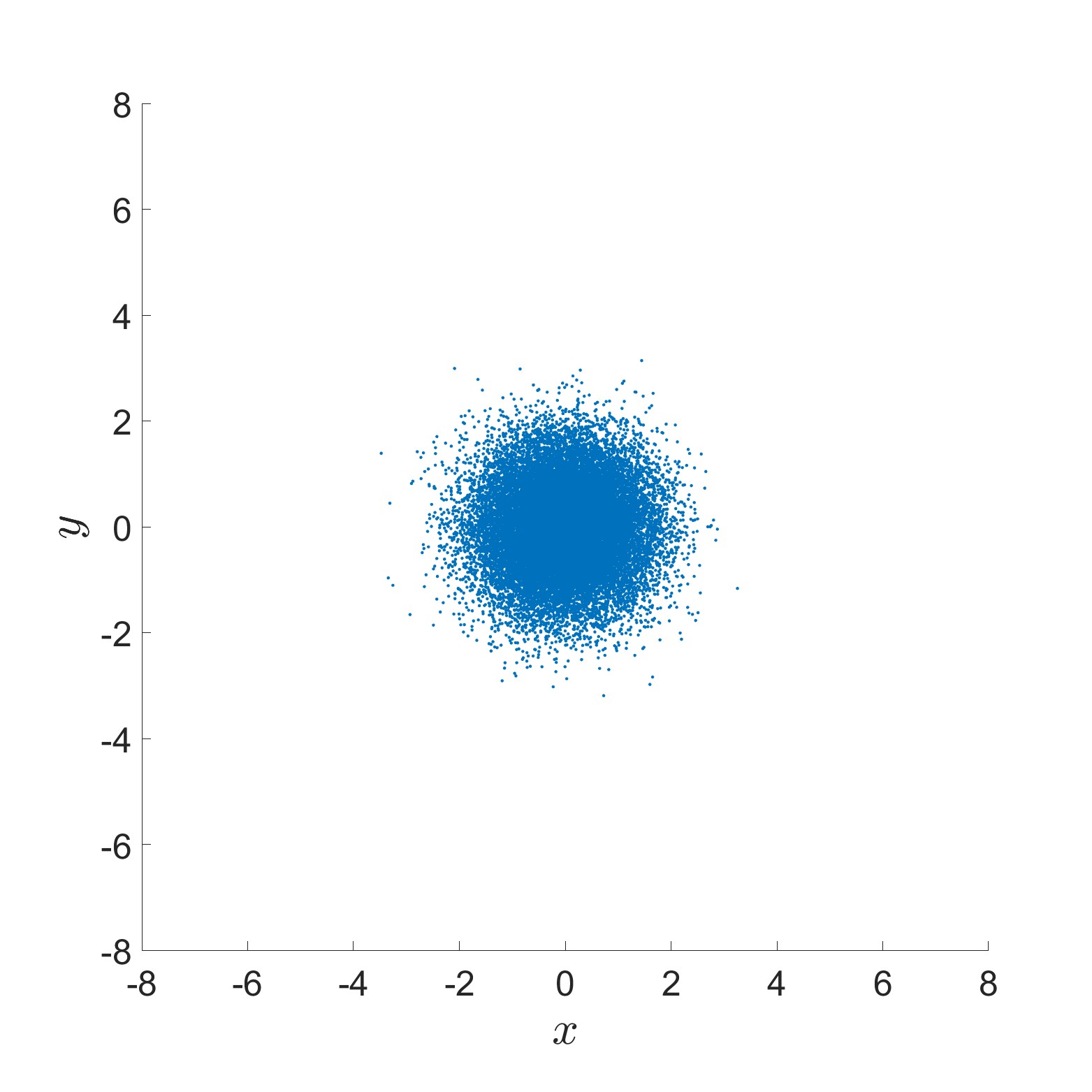}}
    \subfloat[$t = 0.5$]{
    \includegraphics[clip, trim={0cm, 0cm, 0cm, 0cm}, scale=0.125]{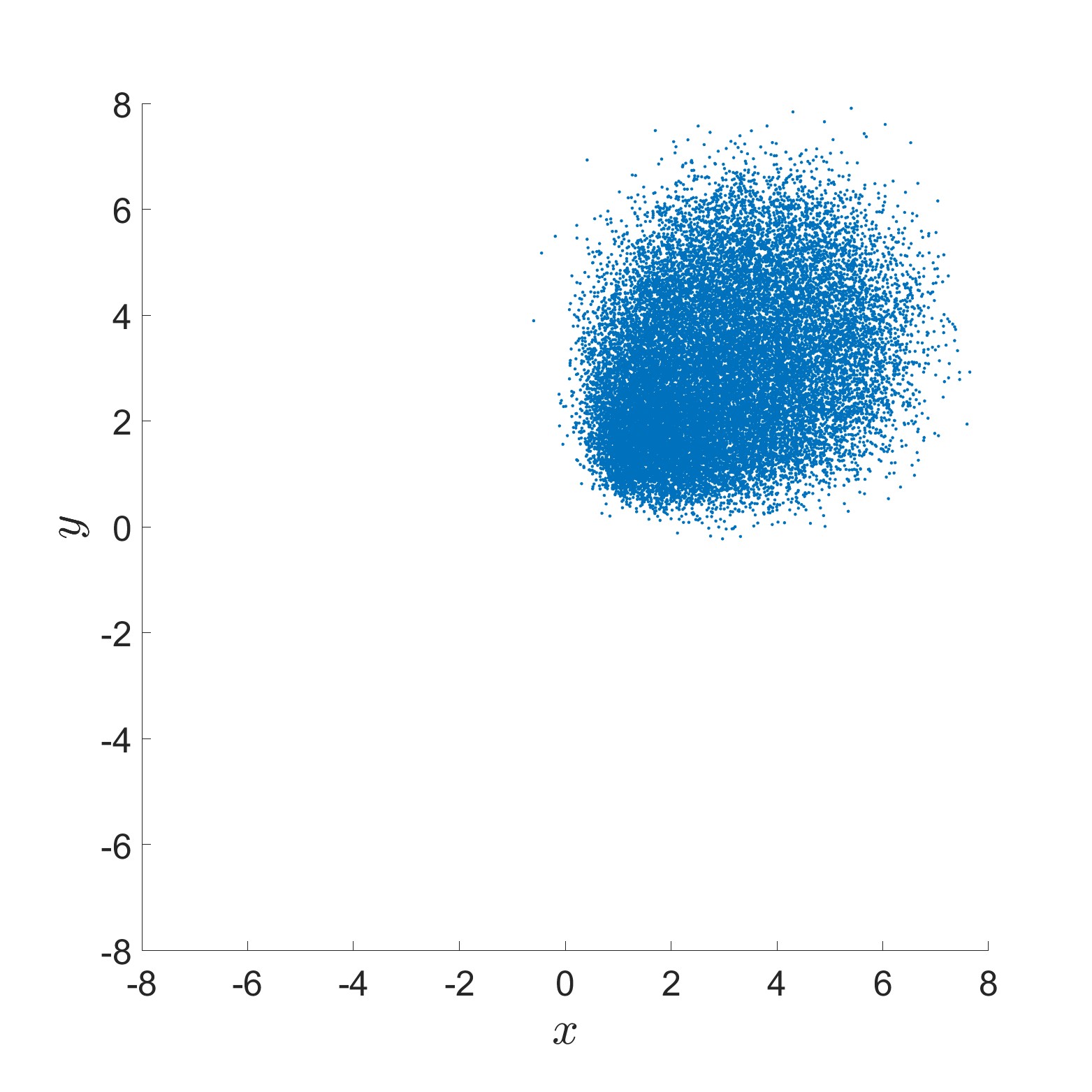}}
    
    \subfloat[$t = 1.0$]{
    \includegraphics[clip, trim={0cm, 0cm, 0cm, 0cm}, scale=0.125]{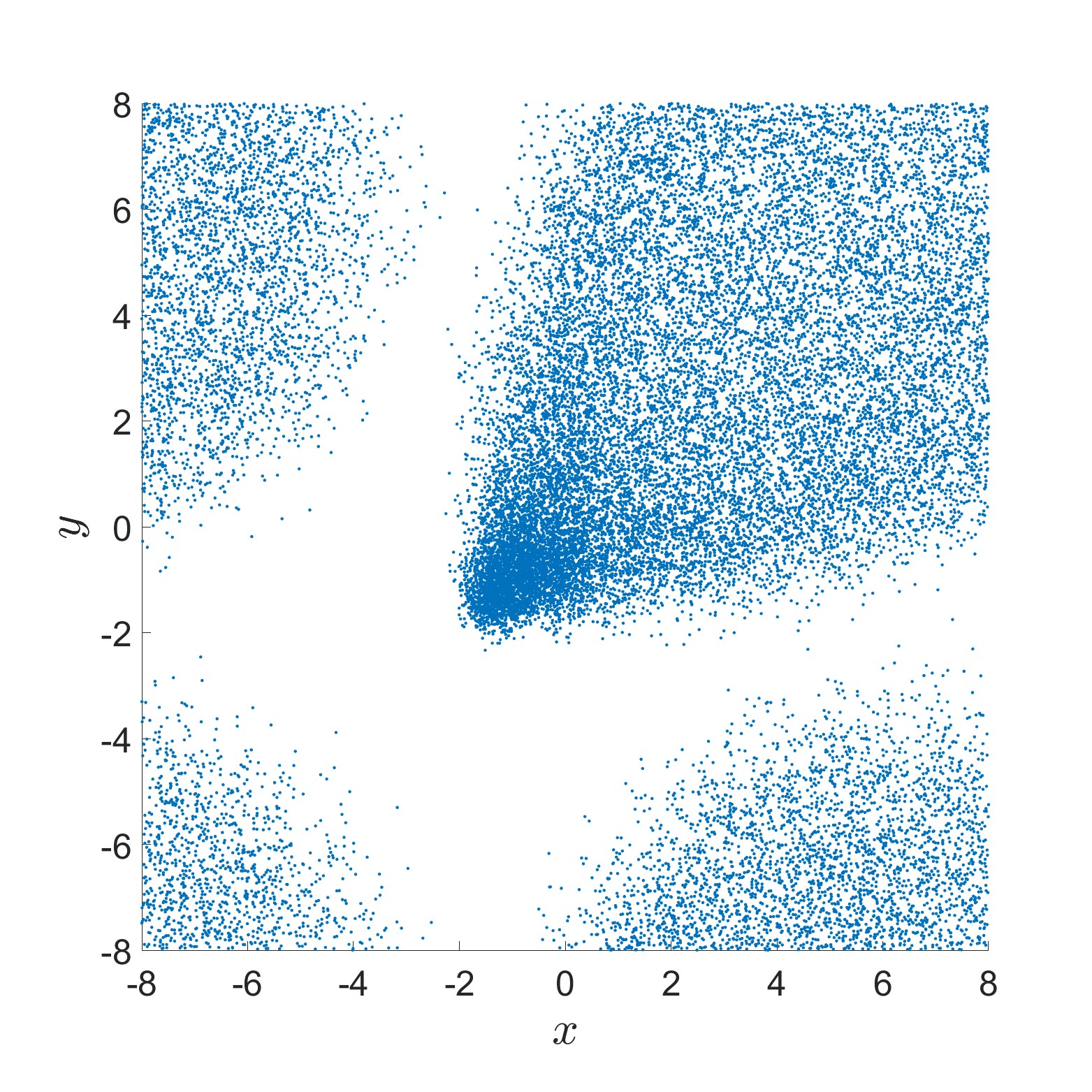}}
    \subfloat[$t = 2.0$]{
    \includegraphics[clip, trim={0cm, 0cm, 0cm, 0cm}, scale=0.125]{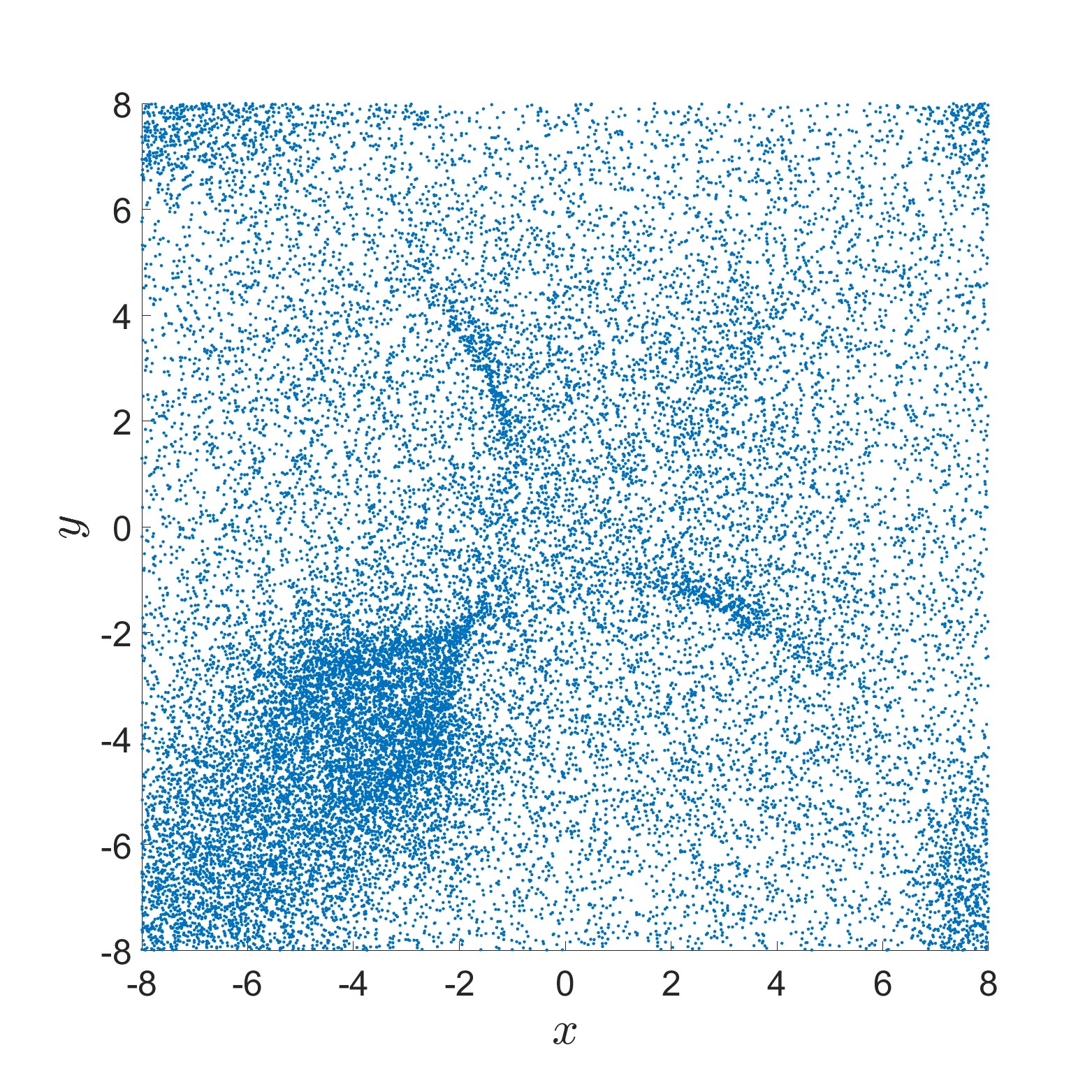}}
    \caption{A Gaussian distribution of electrons and ions are placed on a periodic domain at $t=0.0$. The electrons are given a drift velocity (see Table \ref{tab:moving blob plasma parameters}) in addition to their standard thermal Maxwellian velocity. Unlike the previous work's moving cloud of electrons problem \cite{christliebPIC2022pt2}, only some electrons escape the well ($t=.5$), the rest are pulled back in ($t=1.0$), causing a unique pattern to emerge as electrons are now travelling in opposing directions ($t=2.0$).}
    \label{fig:moving-cloud-scatter-plots}
\end{figure}

%
%
\subsection{Summary}
\label{subsec:5 Summary}

In this section we examined two numerical experiments. First we simulated the well known Weibel Instability, showing our numerical simulations have good agreement with the known analytic solutions. We then simulated a non-relativistic cloud of electrons, an experiment particularly sensitive to gauge errors. Both problems showed negligible variation in the semi-discrete Lorenz gauge and significant improvement in Gauss's Law. The results of both of these experiments are in line with the theoretical results we presented in section \ref{sec:3 wave solvers}.
%
%
%
%
\section{Conclusions}
\label{sec:6 Conclusion}

In our previous works \cite{christliebPIC2022pt1, christliebPIC2022pt2} we presented a new PIC method with a number of desirable properties, including an important theorem proving interdependence of the Lorenz gauge with conservation of charge, as well as proving satisfaction of the Lorenz gauge implies satisfaction of Gauss's law for electricity.  In this paper we generalized these theorems to all forms of the BDF scheme, a variety of second order CDF methods, and proved the one way relationship for arbitrary $s$-stage DIRK methods.  Following this, we chose a previously established method of conserving charge and examined the results from two numerical experiments, which gave numerical evidence for the theoretical results.

Establishing the relationship of charge conservation with the Lorenz gauge and Gauss's Law for electricity constitutes an advance forward for nonstaggered PIC methods and has many exciting avenues of exploration. We have developed wave solvers of arbitrary order accuracy, however, we noted that our particle pusher is still acting as a bottleneck for accuracy. Future work will involve extending the accuracy of this method. Additional work will explore tweaking the scheme to make it compatible with the time staggered nature of the CDF scheme.  We have also focused exclusively on periodic domains, work needs to be done to extend this to domains with more interesting boundary conditions, such as Dirichlet, Neumann, and Outflow, and domains with more interesting geometries than the square box we have thus far been examining.  An exploration into curvilinear coordinates or a finite element discretization would be beneficial, as implementation of $C^1$ finite elements would eliminate numerical dispersion. A more thorough exploration into a full implementation of MOL$^T$ using the fast convolution algorithm rather than the FFT in order to handle these non-periodic boundaries and nontrivial geometries would also be a fascinating study.  The higher order BDF methods, even when bottlenecked in accuracy by a lower order particle method, are known to mitigate numerical dispersion, so it would be interesting to incorporate these solvers in the context of collisional plasmas.  Extending the relation of satisfaction of the Lorenz gauge implying satisfaction of Gauss's law for general DIRK-$s$ methods remains to be done. 

%
%

\section{Acknowledgements}
Some of the simulations presented in this work were supported by computational resources made available through the Institute for Cyber-Enabled Research at Michigan State University. The authors would like to thank the Air Force Office of Scientific Research, the National Science Foundation, and the Department of Energy for their support though grants FA9550-19-1-0281, FA9550-17-1-0394, DMS-1912183, and DE-SC0023164. The authors would also like to thank the Michigan Institute for Plasma Science and Engineering for funding.

\newpage


\bibliography{main}

\end{document}